\documentclass[11pt,a4paper,twoside]{article}
\usepackage{amsfonts}
\usepackage{amssymb}
\usepackage[english]{babel}
\usepackage{amsmath,amsthm}
\usepackage{epsfig}
\usepackage{euscript}
\usepackage{color}
\usepackage{fancyvrb}
\usepackage{fullpage}
\usepackage{preamble}
\usepackage{graphicx}
\usepackage{float}
\usepackage{caption}
\usepackage{subcaption}
\usepackage{physics}
\usepackage{bbm}
\usepackage{hyperref}
\usepackage{tikz}
\usepackage{authblk}

\usepackage{thmtools, thm-restate}

\newtheorem{theorem}{Theorem}
\newtheorem{proposition}{Proposition}
\theoremstyle{definition}
\newtheorem{definition}[thm]{Definition}
\newcommand{\diam}{\mathrm{diam}}

\newcommand{\BB}{\mathcal B}

\newcommand{\HH}{\mathcal H}

\renewcommand{\AA}{\mathcal A}

\newcommand{\caal}{\caA_{\textrm{al}}}

\newcommand{\Adjoint}{\mathrm{Ad}}

\usepackage{xfrac}

\newcommand{\loopc}{\,\square\,}

\newcommand{\eder}{\equiv}

\newcommand{\nor}{|||}

\begin{document}
\title{A classification of $G$-charge Thouless pumps in 1D invertible states}

\author[1]{Sven Bachmann}
\author[2]{Wojciech De Roeck}
\author[3]{Martin Fraas}
\author[2]{Tijl Jappens}

\affil[1]{Department of Mathematics, University of British Columbia, Vancouver, BC V6T 1Z2, Canada}
\affil[2] {Institute of Theoretical Physics, K.U. Leuven, 3001 Leuven, Belgium }
\affil[3]{Department of Mathematics, University of California, Davis, Davis, CA, 95616, USA}
\date{\today}                 
\setcounter{Maxaffil}{0}
\renewcommand\Affilfont{\itshape\small}
\date{\today }

\maketitle
 
\begin{abstract} Recently, a theory has been proposed that classifies cyclic processes of symmetry protected topological (SPT) quantum states.  For the case of spin chains, i.e.\ one-dimensional bosonic SPT's, this theory implies that cyclic processes are classified by zero-dimensional SPT's. This is often described as a generalization of Thouless pumps, with the original Thouless pump corresponding to the case where the symmetry group is $U(1)$ and pumps are classified by an integer that corresponds to the charge pumped per cycle. 
In this paper, we review this one-dimensional theory in an explicit and rigorous setting and we provide a proof for the completeness of the proposed classification for compact symmetry groups $G$. 
\end{abstract}


\section{Introduction}   \label{sec: introduction}

Symmetry protected topological states are states of a spatially extended $d$-dimensional quantum many-body system that are symmetric with respect to a certain on-site symmetry group $G$ and that can be adiabatically connected to a product state, but only if the adiabatic path passes through states that are not $G$-symmetric. The idea is hence that the interesting topological characteristics of these states is only present due to the $G$-symmetry, which explains the nomenclature. Much progress has been made towards classifying such states: While the very notion of equivalence was introduced in~\cite{chen2010local}, classifications were first proposed in $d=1$ \cite{chen_gu_wen_2011,schuch2011classifying,pollmann2012symmetry}, with a first general proposal for all dimensions using group cohomology was given in~\cite{chen2013symmetry}. Among the parallel classification schemes, we mention string order~\cite{nijs:1989a,perez2008string} as well as the entanglement spectrum~\cite{li2008entanglement,pollmann2010entanglement}. More mathematical works are recent, with the cohomological classification given a fully rigorous treatment in~\cite{OgataZ2,ogata2021classification,kapustin2021classification} in one dimension, and \cite{Ogata2d,sopenko2021index} in two dimensions.  

In the present paper, we restrict ourselves to $d=1$ and to bosonic systems, i.e.\ spin chains. Given that restriction, we will consider a class of states that is a priori broader than SPT states, namely \emph{invertible} states on a spin chains, to be defined in Section \ref{sec: states}.  Let us immediately say that it is believed that for spin chains, the two classes (SPT states and invertible states) actually coincide, and for finite symmetry groups $G$ this was proven in \cite{kapustin2021classification}. For us, the class of invertible states is the natural class for which the main result of the paper can be readily formulated. 

Our main interest lies in classifying not the $G$-symmetric invertible states themselves, but \emph{loops of $G$-symmetric invertible states}. Physically, these loops can be thought of as describing periodic pumps. The most well-known example of this is given by considering pumping protocols where a $U(1)$-charge is conserved, often referred to as Thouless pumps.  The phenomenon of topological quantization of charge transport, which is relevant to the integer quantum Hall effect, is then \cite{thouless1983quantization,avron1994charge,De-Nittis:2016va,HastingsMichalakis,OurIndex,kapustin2020hall}:  \emph{The average charge transported per cycle is an integer} when expressed in an appropriate, pump-independent, unit. Of course, this statement needs a restrictive assumption that excludes fractional charge transport in states with non-trivial topological order, cf.\ the fractional Quantum Hall effect. In our setting, it is precisely the property of invertibility that excludes such topologically ordered states. 

In this paper, we construct an index associated with loops of $G$-symmetric invertible states. The index takes values in the first group cohomology group $H^1(G,\bbS^1)$ of the symmetry $G$, and we shall refer to its value as a charge. In order to classify the loops, we introduce an equivalence relation on the set of loops which is given by a new and specific form of homotopy of paths of states, see Section~\ref{sec: processes}. With this in hand, Theorem~\ref{thm: classification loops} states that the cohomological index yields a full classification of the loops: two loops with the same basepoint are equivalent if and only if they have the same index. We show that all possible values of the index can be realized. Finally, we prove that the group structure of $H^1(G,\bbS^1)$ reflects both the composition of cyclic processes and the stacking of physical systems. While this is reminiscent of the classification of SPT states discussed above, it is describing a different aspect of the manifold of states, namely its fundamental group. The discussion of Section~\ref{sec: processes} will however make the connection between the two, in relating the index of a loop in one dimension to an index of a particular SPT state in zero dimension. 
\begin{figure}[htb]
\begin{center}
\includegraphics[width=0.6\textwidth]{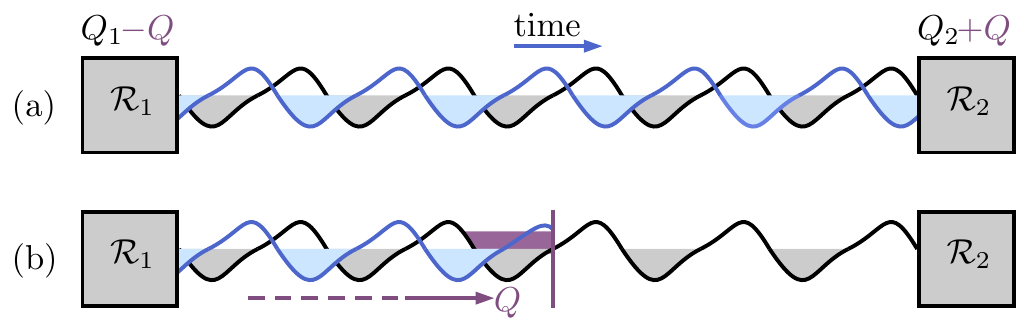}
\caption{Bulk (a) and edge (b) characterization of charge transport.} \label{fig: edge_characterization}
\end{center}
\end{figure} 

Before heading to the main text, we present the idea behind our classification in heuristic terms, and for the sake of clarity for $G=U(1)$, see Figure~\ref{fig: edge_characterization}. There are two ways to measure charge transport. In the `bulk' characterization used e.g. in the original work of Thouless, transport is obtained as the difference of charge in reservoirs before and after the process has taken place; equivalently, one can compute the integrated current. For the purpose of this paper, we prefer the `edge' characterization: Transport is determined by running the process truncated to only one half-line and by measuring the charge accumulated at the edge: specifically by measuring the difference between the amount of charge in the left part of the system before and after the truncated process took place.  In the setting of spin chains, the idea goes back to \cite{KitaevConf} and it was developed more explicitly in~\cite{Else_2014,Else_2016,potter2016classification,von2016phase}, often under the name of `Floquet phases', see also \cite{Xiong_2018} for a review. The problem of classifying loops of states (and more general families of states) has received much interest in the recent physics literature, see in particular~\cite{KunoHatsugai,kapustin2022local,wen2021flow,aasen2022adiabatic,shiozaki2022adiabatic}. 

Finally, we mention an auxiliary result, see Proposition~\ref{prop: uniqueness of ground state}, which may be of interest in its own, namely that small perturbations of a trivial on-site interaction have a unique gapped ground state. Results of that kind have been established in great generality with various techniques, see~\cite{yarotsky2006ground,bravyi2010topological,michalakis2013stability,nachtergaele2020quasi,del2021lie}. However, the specific claim that we need does not seem to appear in the literature, mainly because we consider the infinite-volume setup and because the class of interactions we consider is broader. The fact that the unperturbed ground state is a product state allows for a simpler argument --- but still along the lines of~\cite{michalakis2013stability}.

\subsection*{Plan of the paper} We develop the setup in Sections \ref{sec: setup} and \ref{sec: processes}. The results and relevant examples are stated in Sections \ref{sec: processes} and \ref{sec: pumping index}.  The rest of the paper is devoted to the proofs. In Sections \ref{sec: technical preliminaries} and \ref{sec: hilbert space theory} we state some general preliminaries, and the remaining Sections \ref{sec: trivial index loop is short}, \ref{sec: contractibility of short loops with product basepoint}, \ref{sec: classification for product loops}, \ref{sec: proofs of main} contain specific aspects of the proofs, see Section~\ref{sec:plan} for a more detailed overview.

\subsection*{Acknowledgements}
W.D.R thanks Dominic Else for explaining him the classification of periodic processes of SPT states, which is the main inspiration for this paper.    M.F. was supported in part by the NSF under grant DMS-1907435. W.D.R. and T.J. were supported in part by the FWO under grant G098919N. S.B. was upported by NSERC of Canada.

\subsection*{Data Availability Statement} Data sharing not applicable to this article as no datasets were generated or analysed during the current study.

\subsection*{Conflicts of Interests} The authors have no relevant financial or non-financial interests to disclose.

\section{Setup}   \label{sec: setup}
\subsection{Algebras}
We define a spin chain $C^*$-algebra $\caA$ in the standard way.
To any site $j \in \bbZ$, we associate a finite-dimensional algebra $\caA_j$ isomorphic to  $M_{n_j}(\bbC)$, the algebra of $n_j\times n_j$ matrices with complex entries. The algebra $M_{n_j}(\bbC)$ is equipped with its natural operator norm and $*$-operation (Hermititian adjoint of a matrix) making it into a $C^*$-algebra.  The spin chain algebra $\caA$ is the inductive limit of 
algebra's $\caA_S=\otimes_{j\in S} \caA_j$, with  $S$ finite subsets of $\bbZ$. 
It comes naturally equipped with subalgebra's $\caA_X, X\subset \bbZ$. We refer to standard references \cite{BratRob2,simon2014statistical,naaijkens2017quantum} for more background and details.

\subsection{Almost local evolutions and their generators}
We introduce the framework to discuss processes, i.e. time-evolutions.  Because of the infinite-volume setup, and in contrast to quantum mechanics with a finite number of degrees of freedom, most evolutions cannot be generated by a Hamiltonian that is itself an observable, i.e.\ an element of the  spin chain algebra $\caA$ and we need, instead, to consider so-called \emph{interactions}.  

\subsubsection{Interactions}\label{sec: interactions}
Let $\caF$ be the class of non-increasing, strictly positive functions $f:\bbN^+\to\bbR^+$, with $\bbN^+=\{1,2,\ldots\}$, satisfying the fast decay condition $\lim_{r\to\infty}r^pf(r)=0$ for any $p>0$. 
An interaction is a collections $H=(H_S)$ labelled by finite subsets $S$ of $\bbZ$ such that $H_S=H^*_S \in\caA_S$ and
$$
|| H ||_f  = \sup_{j\in\bbZ}  \sum_{S \ni j} \frac{||H_S||}{f(1+\diam(S))}
$$
is finite for some $f \in\caF$. This excludes for example power law interactions. As in the above expression, sums over $S$ or $S_1,S_2,\ldots$ will always be understood to run over  finite subsets of $\bbZ$.

\subsubsection{Almost local algebra $\caal$}\label{subsec:alal}

If $H$ is an interaction with finite $||\cdot||_f$ norm \emph{and} it satisfies the additional property that $H_S=0$ unless $0 \in S$\footnote{One can replace the site $0$ by any other $j\in\bbZ$ without changing the subalgebra obtained from the procedure described here},  then $\sum_{S} H_S$ is convergent in the topology of $\caA$ and we define
$$
\iota(H)= \sum_{S} H_S\in\caA,
$$
as a map from interactions to observables. The algebra generated by all such elements is a norm-dense subalgebra of $\caA$ that is usually called the \emph{almost local algebra} and that we denote by $\caal$, see also~\cite{kapustin2021classification}. The main interest of this algebra is that it is invariant under the adjoint action associated with an interaction. Indeed, if $A\in\caal$ and $||H||_f <\infty $, then $\sum_{S} [H_S,A]$ is norm convergent and its sum, denoted $[H,A] = [\iota(H),A]$, is again an element of $\caal$.

For later purposes, we also let $\iota_\Lambda(H) = \sum_{S\subset\Lambda} H_S$ whenever the sum is convergent. In particular $\iota_{\bbZ} = \iota$, and if $Z$ is a finite set, then $\iota_Z(H) \in \caA_Z$.
 
\subsubsection{Almost local evolutions}\label{ALEs}

We will consider families of interactions $H(s)$ parametrized by $s\in [0,1]$ and call them `time-dependent interactions' (TDI) provided that they satisfy some regularity conditions to be formulated below. Since the risk of confusion is small, we will often denote them $H$ as well. An interaction $H$ is a TDI if there is $f\in\caF$ such that $s\mapsto H(s)$ is $||\cdot||_f$ bounded and strongly measurable\footnote{namely, it is the limit of a sequence of simple functions, pointwise almost everywhere, where the limit at any $s$ is taken in $||\cdot||_f$-norm, see~\cite{diestel1978vector} }.  On such functions, we use the supremum norms 
$$
|||H |||_{f}=\sup_{s\in[0,1]} ||H(s)||_{f}.
$$
The role of a TDI $H$ is to generate \emph{almost local evolutions} $\alpha_H=(\alpha_H(s))_{s\in[0,1]}$, namely the one-parameter family of strongly continuous $^*$-automorphisms $\alpha_H(s)$ on $\caA$ defined as a particular solution of the Heisenberg evolution equations on $\caal$:
\begin{equation}\label{eq: heisenberg}
\alpha_H(s)[A]= A +i\int_0^s du\, \alpha_H(u) \{ [H(u),A] \}.
\end{equation} 
The integral on the right hand side is to be understood in the sense of Bochner, and the strong measurability of the integrand follows from strong measurability of $s\mapsto H(s)$ and the strong continuity of $s\mapsto \alpha_H(s)$. 

Specifically, $\alpha_H(s)[A]$ is defined as the limit in the topology of $\caA$ of solutions of~(\ref{eq: heisenberg}) where $H$ is replaced with $\iota_{\Lambda_n}(H)$ for an increasing and absorbing sequence $(\Lambda_n)_{n\in\bbN}$ of subsets of $\Gamma$. By a standard argument using the Lieb-Robinson bound \cite{Lieb:1972ts,nachtergaele2006propagation}, this procedure is well-defined and the limit solves~(\ref{eq: heisenberg}). 

We conclude this section by noting that all definitions introduced so far can be applied to the tensor product algebra $\widetilde \caA= \caA \otimes  \caA'$ of two spin chain algebras $\caA$ and $\caA'$. The site algebras $\widetilde \caA_j$ are defined as $\caA_j\otimes\caA'_j$ and $\widetilde\caA$ is then the inductive limit of (tensor products of) these site algebras.

\subsection{States}\label{sec: states}

States are normalized positive linear functionals on the spin chain algebra $\caA$. The set of states is denoted by $\caP(\caA)$ and a natural metric on states is derived from the Banach space norm
\begin{equation}\label{eq: metric on states}
||\psi-\psi'|| = \sup_{A \in\caA, ||A||=1} |\psi[A]-\psi'[A]|.
\end{equation}
The set $\caP(\caA)$ is convex and its extremal points are called the pure states. In all what follows, we consider only pure states.

A distinguished class of states is that of spatial product states $\phi$, i.e.\ states that satisfy
$$
\phi(A_jA_i)= \phi(A_j)\phi( A_i),\qquad  A_j \in\caA_j,A_i\in \caA_i, i\neq j.
$$  
\begin{definition}[Invertible state]
A pure state $\psi$ on a spin chain algebra $\caA$ is \emph{invertible} if there is a pure state $\overline{\psi}$ on a spin chain algebra $\caA'$, a TDI $H$ on $\widetilde \caA= \caA \otimes \caA'$ and a product state $\phi$ on $\widetilde \caA$ such that, 
\begin{equation}\label{eq: def invertibility}
\psi\otimes\overline{\psi} =\phi \circ \alpha_H(1).
\end{equation}
\end{definition}
In this context, the state $\overline{\psi}$ is usually referred to as an inverse of $\psi$.  Invertibility, first coined by \cite{KitaevConf} is a way of expressing that a state has no intrinsic topological order. For spin chains, invertibility has recently been shown \cite{kapustin2021classification} to be equivalent to a stronger property, namely being stably short-range entangled, which amounts to the additional requirement that $\overline{\psi}$ is itself a product state. The notion of invertibility is however essential when adding symmetries.

\subsection{Symmetries}\label{sec: symmetries}
Let $G$ be a compact topological group. We equip the chain algebra with a strongly continuous on-site action of $G$: For any $j\in \bbZ$, there is an automorphism $\gamma_{j}(g)$ of $\caA_j$ for each $g\in G$ such that $\gamma_{j}(g')\circ \gamma_{j}(g)=\gamma_{j}(g'g)$ and $g\mapsto \gamma_j(g)[A]$ is continuous for any $A\in \caA_j$. For any $A\in\caA_S$ with finite $S$, we let $\gamma(g)[A]=\otimes_{j\in S}\gamma_{j}(g)[A]$ and extend this action to a strongly continuous action  on $\caA$ by density. A state $\psi$ is $G$-invariant if  $\psi \circ \gamma(g)=\psi$ for all $g\in G$. 
A TDI $H$ is called $G$-invariant if all its local terms are $G$-invariant: 
$$
\gamma(g)[H_S(s)]=H_S(s), \quad \text{for all $s\in[0,1]$, any finite $S \subset \bbZ$ and all $g\in G$}.
$$
Note these notions of invariance depend on the choice of action $\gamma$, which is not well reflected in the nomenclature. In particular, any state is $G$-invariant if the action $\gamma$ is chosen trivial. 

\begin{definition}[$G$-invertible state]
A $G$-invariant pure state $\psi$ on a spin chain algebra $\caA$ is $G$-\emph{invertible} if it is invertible, and its inverse $\overline{\psi}$ and the interpolating family of states $\phi\circ\alpha_H(s)$ (see \eqref{eq: def invertibility}) can be chosen to be $G$-invariant as well. 
\end{definition}
As we shall see shortly, an equivalent condition for a $G$-invariant, invertible pure state $\psi$ to be $G$-invertible, is that the TDI $H$ in \eqref{eq: def invertibility} can be chosen to be $G$-invariant.

\section{Processes}\label{sec: processes}

Processes should intuitively be defined as continuous paths or loops of states.  This  presupposes in particular a suitable topology on the space of states. However, the standard topologies on states (e.g.\ norm topology or weak topology) do not fit physical intuition, in particular concerning locality properties, and it has become standard in this setting to view a curve of states as `continuous' if it can be generated by a TDI, even if there is no known topology that substantiates this.  This idea is sometimes referred to as \emph{automorphic equivalence} and it was introduced in the works \cite{hastings2005quasiadiabatic,bachmann2012automorphic}.  Concretely, we now define paths of pure states.
\begin{definition}[Paths] \label{def: paths}
Let $\caA$ be a spin chain algebra and let $\gamma$ be a group action. A \emph{path} is a function $[0,1] \ni s \mapsto \psi(s)\in\caP(\caA)$ where each $\psi(s)$ is a pure state, and there is a TDI $H$ such that $\psi(s)=\psi(0)\circ\alpha_H(s)$. A path is called a \emph{loop} if $\psi(0) = \psi(1)$.
Finally, a path $\psi(\cdot)$ is $G$-invariant if all $\psi(s)$ are G-invariant. 
\end{definition}
Two remarks are in order. First of all, the group action (as well as the algebra) are fixed along the path. Secondly, while we did not require that a $G$-invariant path is generated by a $G$-invariant TDI, it is always possible to choose the TDI $G$-invariant. Indeed, if a $G$-invariant path is generated by a TDI $H$, then the TDIs given by $\gamma(g)[H(s)_S]$ for $g\in G$ all generate the same path since $\psi$ and $\psi(s)$ are $G$-invariant, and so
$
\bar{H}_S(s)=   \int_G d\mu(g)  \gamma(g)[H_S(s)],
$
with $\mu(\cdot)$ the normalized Haar measure on $G$, generates the same path. Now $\bar{H}$ is $G$-invariant and $|||\bar H |||_f \leq |||H|||_f$ for any $f\in \caF$ since $\gamma(g)$ are automorphisms.

\subsection{Charge and equivalence}\label{sec: equivalence}

We can now use this notion of path to define connected components of the set of states on a spin chain algebra. We present this in the form of an equivalence relation.
\begin{definition}[Equivalence]
A pair of pure states $\psi,\psi'$  on a spin chain algebra $\caA$ are called \emph{equivalent} if and only if there is a connecting  path $s\mapsto \psi(s)$ such that $\psi(0)=\psi$ and $\psi(1)=\psi'$. If $\psi,\psi'$ are $G$-invariant, the pair is \emph{$G$-equivalent} if the connecting path is $G$-invariant. 
\end{definition}

We first settle the classification of the $G$-equivalence classes of the set of pure $G$-invariant states for zero-dimensional systems that we define below. In this work, we will not consider the classification problem on spin chains, see however Section~\ref{sec: stable equivalence}.

\subsubsection{Zero-dimensional systems}\label{sec: equivalence zero}
One relevant case where the classification problem is straightforward is for \emph{$0$-dimensional} pure $G$-invariant states, called $0$-dim $G$-states for brevity. 
By this we mean that $\caA = \caB(\caH)$, for some Hilbert space $\caH$. The $0$-dim $G$-states are $G$-invariant pure normal states on $\caB(\caH)$ and we will assume that there is at least one such $0$-dim $G$-state. 
It follows\footnote{See Lemma \ref{lem: gns} and its proof} that the $G$-action on $\caB(\caH)$ can be implemented by a family of unitaries $U(g)$, i.e.\  $ \gamma(g)[A] = U(g) A U(g)\str$, and $U(\cdot)$ is a strongly continuous unitary representation of $G$.

A pure normal states $\psi$ on $\caB(\caH)$ is given by a ray in Hilbert space as  $\psi(B)=\langle\Psi, B\Psi\rangle$ for some $\Psi \in \caH$ with $||\Psi||=1$, which we call a representative of $\psi$.
 The equivalence question for $0$-dim $G$-states is stated more easily than in the case of states on spin chain algebras, because now the natural topologies are physically relevant. 
We simply call a pair of $0$-dim $G$-states $\psi,\psi'$  $G$-equivalent if they can be connected by a norm-continuous path of $0$-dim $G$-states\footnote{Choosing a weaker topology would not change the classification}.  
Note that if the action $\gamma$ is trivial, then any pair of $0$-dim $G$-states is $G$-equivalent, by the connectedness of the projectivization of~$\caH$. For general compact groups $G$, we will provide a full classification of $G$-equivalence classes of $0$-dim $G$-states  in Proposition \ref{prop: zerodim}.
 
 \subsubsection{G-charge}\label{sec: g charge}

We consider the group $H^1(G)=\mathrm{Hom}(G,\bbS^1)$ of continuous group homomorphisms $G\to \bbS^1$, equipped with the topology of uniform convergence. We use additive notation for the Abelian group $H^1(G)$. Because of the compactness of $G$,  $H^1(G)$ is a discrete space in this topology \cite{hofmann1965topologische}, and this fact plays a central role in our reasoning.

The next proposition shows that $H^1(G)$ completely classifies the $0$-dim $G$-states. Recall that equivalence is defined between states of the same algebra equipped with the same group action. 
\begin{proposition}\label{prop: zerodim}
To every ordered pair of $0$-dim $G$-states $\psi_1,\psi_2$ of $\caB(\caH)$ equipped with $G$-action $\gamma$ we can associate an element $h_{\psi_2/\psi_1} \in H^1(G)$ such that 
\begin{enumerate}
\item $\psi_1,\psi_2$ are $G$-equivalent if and only if  $h_{\psi_2/\psi_1}=0$
\item $h_{\psi_3/\psi_2}+h_{\psi_2/\psi_1}=h_{\psi_3/\psi_1}$
\item If $\psi_1',\psi_2'$ are $G$-invariant states of $\caB(\caH')$ equipped with $G$-action $\gamma'$, then $h_{(\psi_2 \otimes \psi'_2)/(\psi_1 \otimes \psi'_1)}= h_{\psi_2/\psi_1}+ h_{\psi'_2/\psi'_1}$
\item If $\delta:\caB(\caH')\to \caB(\caH) $ is a $*$-isomorphism such that $\delta\circ\gamma'(g)=\gamma(g)\circ\delta$, then $h_{\psi_2/\psi_1}= h_{(\psi_2\circ\delta)/(\psi_1\circ\delta)}$
\end{enumerate} 
\end{proposition}
We call henceforth $h_{\psi_2/\psi_1}$ the relative charge of $\psi_2$ and $\psi_1$. 
An alternative view on $G$-charge would be to restrict the algebra to $G$-invariant elements, in which case we could phrase the above proposition as describing the \emph{superselection sectors}.
\begin{proof}
We choose a unitary representation $U(\cdot)$ of $G$ that implements $\gamma(\cdot)$ on $\caB(\caH)$.  
If a state $\psi_j$ is $G$-invariant then any of its representatives $\Psi_j$ is invariant up to phase: $U(g)\Psi_j= z_j(g)\Psi_j$, and $G\to U(1): g\mapsto z_j(g)$ is a group homomorphism. We then define $h_{\psi_2/\psi_1}(g)\in\bbS^1$ by
$$
e^{ih_{\psi_2/\psi_1}(g)}=\frac{z_2(g)}{z_1(g)}. 
$$ 
It is apparent that $z_j(g)$ does not depend on the choice of representatives $\Psi_j$. The unitary representation $U(\cdot)$ is uniquely determined up to multiplication with a $U(1)$ representation of $G$, which does affect $z_j(g)$ but not $h_{\psi_2/\psi_1}(g)$, and we conclude that $h_{\psi_2/\psi_1} \in H^1(G)$ is well-defined. To prove $i)$, we note that if two states are $G$-equivalent with interpolating path $\psi(s)$, then $s\mapsto h_{\psi(s)/\psi_1} $ is continuous and hence, by the discreteness of $H^1(G)$, constant. Reciprocally, if two states have zero relative charge, then there is a continuous path of $G$-invariant unitaries $V_s$ acting non-trivially only in the subspace spanned by their representatives $\Psi_1,\Psi_2$ and interpolating between the two. The vector states given by $V_s\Psi_1$ form the connecting path. $ii)$ is checked using the relation $\frac{z_3}{z_1} = \frac{z_3}{z_2}\frac{z_2}{z_1}$ and $iii)$ follows similarly. To prove $iv)$, we pick a unitary representation $U'$ implementing $\gamma'$ on  $\caB(\caH')$. By the intertwining property we have then that $\gamma(g)[\cdot]=\delta(U'(g))[\cdot]\delta(U'(g))\str$ and the claim follows from $z_j(g) = \psi_j(U(g)) = \psi_j(\delta(U'(g))$.
\end{proof}

\subsection{Homotopy}\label{sec: homotopy}

We introduce a notion of a homotopy of two loops, say $\psi_0(\cdot),\psi_1(\cdot)$, generated by the TDIs $H_0, H_1$, respectively. 
If we had an appropriate topology on the set of states, we would say that a homotopy of those loops is a continuous function $(s,\lambda)\mapsto \psi_\lambda(s)$, such that $\psi_\lambda$ is a loop for each $\lambda$, reducing to the two given loops for $\lambda=0$ and $\lambda=1$. 
Though there is no topology,  we have  already specified in Definition \ref{def: paths} what it means for a state-valued function to be `continuous' and we impose that definition here for the dependence on the parameter $\lambda$ as well.  That is, we demand that $s\mapsto \psi_\lambda(s),\lambda\mapsto \psi_\lambda(s) $ are generated by TDIs that we note as $H_\lambda(\cdot), F_s(\cdot)$;
\begin{equation}\label{eq: homotopy}
\psi_\lambda(s)= \psi_\lambda(0) \circ \alpha_{H_\lambda}(s),\qquad
\psi_\lambda(s)= \psi_0(s) \circ \alpha_{F_s}(\lambda),\qquad \text{for any $s, \lambda \in[0,1]$}
\end{equation}
 Additionally, we demand a uniform bound on the generating TDIs:
\begin{equation}\label{Homotopy: Uniform bound}
\sup_{\lambda} ||| H_\lambda |||_f <  \infty,\qquad \sup_s  ||| F_s |||_f  <\infty
\end{equation}
A more intuitive way of phrasing our definition of homotopy is to say that motion on the 
 two-dimensional sheet of states $(s,\lambda)\mapsto \psi_\lambda(s)$ is generated by an interaction-valued vector field $H\frac{\partial}{\partial s} + F\frac{\partial}{\partial \lambda} $, see Figure~\ref{fig: homotopy states}, but we do not formalize this, as it would require technical conditions that are ultimately not relevant for our goals (but see \cite{kapustin2022local}).
 
\begin{figure}[h]
\begin{center}
   \includegraphics{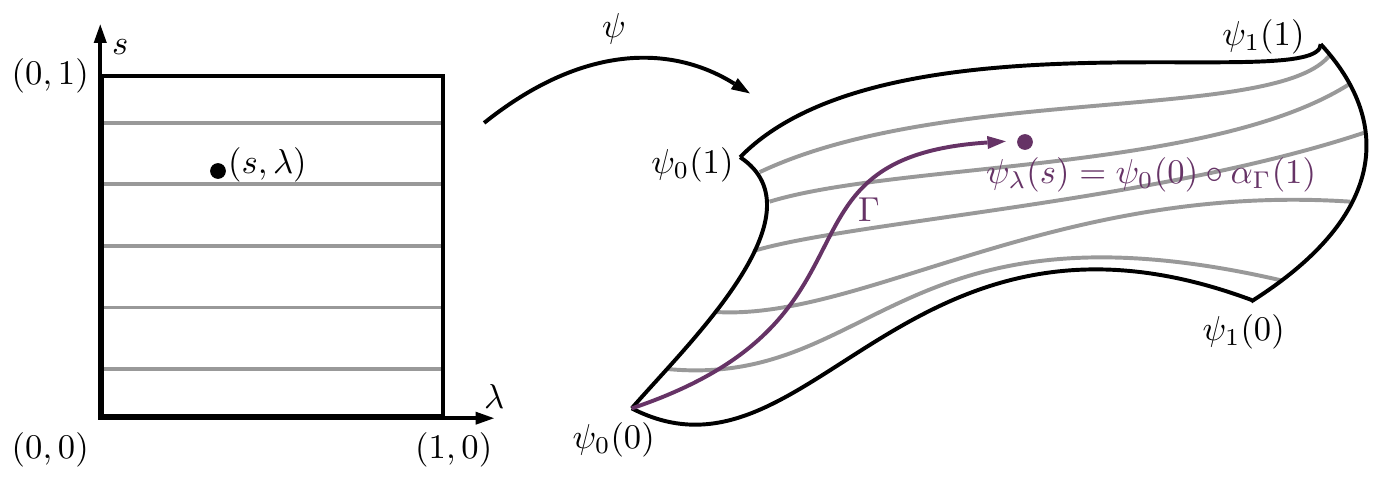}
   \caption{Intuitive notion of homotopy: the state $\psi$ at the endpoint of the curve $\Gamma: [0,1]\to [0,1]^2$ is given by  $\psi=\psi_0(0)\circ \alpha_\Gamma(1)$, with $\alpha_\Gamma$ the almost local evolution generated by the TDI $L(z), z\in [0,1]$  defined as i.e.\  
$L(z)= H\frac{\partial \Gamma(z)}{\partial s} + F\frac{\partial \Gamma(z)}{\partial \lambda} $. }
\label{fig: homotopy states}
\end{center}
\end{figure} 
 
\begin{definition}[Homotopy and $G$-homotopy]
A pair of loops $\psi(\cdot)$ and $\psi'(\cdot)$ is homotopic if there exists a homotopy $(s,\lambda)\mapsto \psi_\lambda(s)$ in the sense of eqs.\ (\ref{eq: homotopy}, \ref{Homotopy: Uniform bound}) such that $\psi_0(\cdot)=\psi(\cdot)$ and $\psi_1(\cdot)=\psi'(\cdot)$.   
A pair of loops is $G$-homotopic if the homotopy can be chosen such that $ \psi_\lambda(s)$ are $G$-invariant  for all $s,\lambda$.   
\end{definition}

One could opt for a different definition of paths and their homotopies by requiring that the map $(s,\lambda)\to H_\lambda(s)$ is sufficiently smooth (allowing e.g.\ for isolated singularities) in $||\cdot||_f$-norm, for some $f\in\caF$. Such a condition would imply our definition of homotopy and all our results would still hold if we stuck to this definition but we find it unnatural as it puts the stress squarely on the generators rather then on the states.

On the other hand, it is not possible to drop the boundedness property \eqref{Homotopy: Uniform bound}. The necessity of this boundedness property is best understood on an example. Let the TDI $H$ generate a loop $\psi(\cdot)$. Then, for any $\lambda$, 
$$
H_\lambda(s)= \begin{cases} 0  &  s\leq \lambda \\
   \tfrac{1}{1-\lambda} H(\tfrac{s-\lambda}{1-\lambda})&  s>\lambda
  \end{cases}
$$ 
generates a loop  $\psi_\lambda(\cdot)$. 
For any $\lambda<1$, this loop is simply a reparametrization of the original loop, namely 
$$
\psi_\lambda(s)= \begin{cases} \psi(0)  &  s\leq \lambda \\
  \psi( \tfrac{s-\lambda}{1-\lambda})  &  s>\lambda
  \end{cases}
$$ 
and hence these loops are homotopic to each other by any reasonable definition. However, for $\lambda=1$, we get the constant loop $\psi_1(s)=\psi(0)$. Hence, we arguably need some uniform boundedness property on the generators, akin to \eqref{Homotopy: Uniform bound}, to avoid the conclusion that any loop is homotopic to a constant loop.

\subsubsection{Concatenation of loops} \label{sec: concatenation loops}

If $\psi(\cdot)$ is a loop, we refer to the state $\psi(0)$ is its basepoint. We define the concatenation $\psi' \loopc \psi$ of two loops $\psi',\psi$ with the same basepoint, in the obvious way
\begin{equation} \label{eq: concatenation}
(\psi' \loopc \psi) (s) = \begin{cases} \psi(2s) &  s \leq \tfrac{1}{2}  \\
\psi'(2s-1) &  s \geq \tfrac{1}{2}  \\
  \end{cases}
\end{equation}
and we note that $\psi' \loopc \psi$ is also a loop with the same basepoint as $\psi,\psi'$.
It is quite intuitive that, if $\psi'$ is a constant loop, then $\psi' \loopc \psi$ and $\psi \loopc \psi'$ are homotopic to $\psi$. Along the same lines, for a loop $\psi$, we can define the time-reversed loop ${\psi^{\theta}}(s)=\psi(1-s)$ and it turns out that ${\psi^{\theta}} \loopc\psi$ and $\psi\loopc{\psi^{\theta}}$ are  homotopic to a constant loop. Tools to prove such statements will be furnished in Section \ref{sec: technical preliminaries}. Hence, just as for the conventional notion of homotopy, our notions of homotopy and concatenation lead to a group structure  on homotopy equivalence classes of loops.  

\subsection{Stable equivalence and homotopy}  \label{sec: stable equivalence and homotopy}
\subsubsection{Stable equivalence}   \label{sec: stable equivalence}
We return to the question of equivalence of states.  There are some works classifying certain $G$-invariant pure states on spin chain algebras, see in particular~\cite{ogata2021classification}, 
but most authors \cite{kitaev2009periodic,chen_gu_wen_2011,kapustin2021classification} choose to relax the notion of equivalence to allow for tensoring with a product state on an auxiliary spin chain algebra, analogous to our 
definition
of invertibility \eqref{eq: def invertibility}: One identifies $\psi$ on a spin chain algebras $\caA$ with $\psi \otimes \phi$ on $\caA\otimes\caA'$ if $\phi$ is a pure product state on $\caA'$. This gives states on spin chain algebras the structure of a monoid, with product states being the unit element. In that sense invertible states are indeed those that have an inverse in this monoid.   
\begin{definition}[Stable ($G$)-equivalence]
A pair of states $\psi_1,\psi_2$ on spin chain algebras $\caA_1,\caA_2$ are $(G)$-stably equivalent if and only if there are a ($G$-invariant) product states $\phi_1,\phi_2$ on spin chain algebras $\caA'_1,\caA'_2$  such that the states $\psi_1\otimes \phi_1, \psi_2\otimes \phi_2$ are ($G$)-equivalent.
\end{definition}  
 The relaxation of the equivalence property by stabilization is natural, if only because it allows to compare states defined on spin chain algebras with distinct on-site dimensions $n_j$.
On the other hand, a case could be make that our notion of equivalence is too generous and that one should only allow for tensoring with product states of which all the factors have zero $G$-charge with respect to some fixed reference state, see \cite{kapustin2021classification} for a similar restriction.  For example, consider a $U(1)$-invariant spin chain where, for each site $i$, we specify a reference zero-dimensional pure $G$-invariant state $\phi_i$ and we consider the $G$-invariant product state $\psi=\otimes_i \psi_i$ such that 
$h_{\psi_i/\phi_i}=i$, i.e.\ the charges grow linearly as one moves outwards.   By our definition, $\psi$ is trivially stably $G$-equivalent to $\phi$, but it is not by the alternative definition, which seems in this case physically more relevant. 

The classification of states up to stable ($G$-)equivalence is usually referred to as the classification of (symmetry protected) topological phases. It has been a central theme in the recent literature. 

\subsubsection{Stable Homotopy} \label{sec: stable homotopy}

Just as for the notion of equivalence, it is natural to relax the notion of homotopy by allowing tensoring with product states. 

\begin{definition}[Stable ($G$)-homotopy]
A pair of loops $\psi_1(\cdot),\psi_2(\cdot)$ on spin chain algebras $\caA_1,\caA_2$ are $(G)$-stably homotopic if and only if there are  ($G$-invariant) product states $\phi_1,\phi_2$ on spin chain algebras $\caA'_1,\caA'_2$ such that the following loops are ($G$)-homotopic 
$\psi_1(\cdot)\otimes \phi_1, \psi_2(\cdot)\otimes \phi_2$
\end{definition} 

We note that the adjoined loops on $\caA'_1,\caA'_2$ are constant and that the role of the unit element in the monoidal structure of loops is played  by constant loops with product basepoint.

\subsection{Classification of loops}\label{sec: classification of loops}

To streamline the discussion of our results, we introduce some terminology tailored to our results.   This also serves to summarize the definitions made in the previous sections. 
\begin{definition}[$G$-state]
A $G$-state is a triple $(\caA,\gamma,\psi)$ where
\begin{enumerate}
\item $\caA$ is a spin chain algebra.
\item $\gamma$ is a strongly continuous $G$-action on $\caA$ by on-site $*$-automorphisms. 
\item $\psi$ is a pure, $G$-invariant and $G$-invertible state on $\caA$. 
\end{enumerate}
\end{definition}
The extension to loops is then self-explanatory. 
\begin{definition}[$G$-loop]
A $G$-loop is a triple $(\caA,\gamma,\psi(\cdot))$ where
\begin{enumerate}
\item For any $s\in [0,1]$, $(\caA,\gamma,\psi(s))$ is a $G$-state. 
\item $s\mapsto \psi(s)$ is a loop, as defined in Section~\ref{sec: homotopy}.  
\end{enumerate}
\end{definition}

We note that both $G$-states and $G$-loops come equipped with the monoidal structure of taking tensor products. If $\psi,\psi'$ are $G$-loops, then $s\mapsto \psi(s)\otimes\psi'(s)$ is a $G$-loop. Additionally, $G$-loops $\psi,\psi'$ that  have a common basepoint, 
admit another natural operation, namely concatenation, see Section~\ref{sec: concatenation loops}.  To be precise, saying that two $G$-loops `have a common basepoint' presupposes that they are defined on a common spin chain algebra $\caA$, equipped with the same $G$-action $\gamma$.

We now  come to our main result, the classification of loops, up to stable $G$-homotopy.

\begin{theorem}\label{thm: classification loops} [Classification of loops]
There is a map $h$ assigning to any $G$-loop $\psi$ an element  $h(\psi) \in H^1(G)$ such that 
\begin{enumerate}
\item If $\psi$ is a constant $G$-loop, i.e. $\psi(s)=\psi(0)$ for $s\in[0,1]$, then $h(\psi)=0$. 
\item For any $h \in H^1(G)$, there is a $G$-loop $\psi$ such that $h(\psi)=h$. 
\item If a pair of $G$-loops $\psi$ and $\psi'$ is stably $G$-homotopic, then $h(\psi)=h(\psi')$.
\item If the $G$-loops $\psi,\psi'$ have the same basepoint,  then  $h(\psi'\loopc \psi)=h(\psi')+h(\psi)$.
\item For $G$-loops  $\psi,\psi'$, we have  $h(\psi'\otimes \psi)=h(\psi')+h(\psi)$.
\item Let $\psi,\psi'$ be $G$-loops whose basepoints are stably $G$-equivalent and satisfying $h(\psi)=h(\psi')$, then $\psi,\psi'$ are stably $G$-homotopic. 
\end{enumerate} 
\end{theorem}

\noindent We point out that the surjectivity $ii$ presupposes that $\caA$ and the action $\gamma$ can be varied.

\section{Pumping index}\label{sec: pumping index}

In this section, we provide an explicit construction of the map $h$ whose existence was postulated in Theorem \ref{thm: classification loops}. This echoes the informal
discussion of the Thouless pump in the introduction.

\subsection{Relative charge associated to cut loop}\label{sec: relative charge for cut loop}

We consider a loop $\psi(\cdot)$ generated by the $G$-invariant TDI $H$. 
We restrict the generating TDI $H$ at the origin to get the truncated TDI $H'$ given by
\begin{equation}\label{eq: tdi cut loop}
H'_{S}(s)= \begin{cases}   H_S(s)  &  S \subset \{\ldots,-2, -1,0\} \\  0 & \text{otherwise} \end{cases}  
\end{equation}
This TDI is again $G$-invariant, but the path it generates is in general no longer a loop because it can fail to be closed. That is, its endpoint, which shall henceforth call the \emph{pumped state},
$$
\psi':=  \psi(0) \circ \alpha_{H'}(1) 
$$
need not be equal to the basepoint $\psi(0)$.  Nevertheless, far to the left of the origin,  $\psi'$ resembles $\psi(0) \circ \alpha_{H}(1)=\psi(0)$ because far to the left $H$ is equal to $H'$. Far to the right,  $\psi'$ trivially resembles $\psi(0)$. Making these observations precise, we will show that
\begin{equation}\label{eq: same away from cut}
 ||(\psi(0)-\psi')_{< -r} || \leq f(r),\qquad  ||(\psi(0)-\psi')_{{\geq r}} || \leq f(r)
\end{equation}
for some $f\in \caF$. Here we abbreviated $\zeta_X=\zeta|_{\caA_X}$ for the restriction of functionals $\zeta$ on $\caA$ to the subalgebra $\caA_X$, with $X=\{{< -r}\}$ or $X=\{{\geq r}\}$ in the above examples. Let $\caH$ be the Hilbert space carrying the GNS representation associated with $\psi(0)$. Then the above bounds lead to
\begin{proposition}\label{prop: equivalence of states}
If the pure state $\psi(0)$ is invertible, then the pumped state $\psi'$, is normal with respect to $\psi(0)$. That is, it is represented by a density matrix on $\caH$.
\end{proposition}
\noindent We postpone the proof to Section~\ref{sec: proofs of main}. 
 By standard considerations (see Lemma \ref{lem: gns}), we obtain a unitary action $g\mapsto U(g)$ satisfying the requirements of Proposition \ref{prop: zerodim} in Section \ref{sec: equivalence zero},  and therefore  the pumped charge $h_{\psi'/{\psi(0)}}$ is well-defined. We then define the index
$$
I(\psi(0),H):=h_{\psi'/{\psi(0)}}.
$$
We will show that it only depends on the loop $\psi$, not on the TDI $H$ generating the loop and so the map $\psi \mapsto I(\psi(0),H) $ is well-defined as a map on $G$-loops.
\begin{theorem}\label{thm: pump index}
The map $h : \psi \mapsto I(\psi(0),H) $ satisfies all the properties of Theorem \ref{thm: classification loops}.
\end{theorem}
We will therefore write the index as $h(\psi)=I(\psi(0),H)$. We will use the notation $I(\psi(0), H)$ whenever we locally want to emphasize the dependence on $H$.

\subsection{Example}\label{sec: examples}
We present a class of examples that realize a loop for every element $ h \in H^1(G)$.
 We take the on-site algebra $\AA_j\simeq \BB(\HH)$ with  $\HH$ a $3$-dim Hilbert space with a distinghuished orthonormal basis, denoted by $\{|-h\rangle,  | 0\rangle,  |h\rangle\}$. The group action is implemented by the unitary on-site operator $U(g)$ defined on the basis by 
\begin{equation}
	U(g) | h'\rangle = e^{-i h'(g)}| h'\rangle,
\end{equation}
and hence the automorphism on the on-site algebra $\BB(\HH)$ is given by $\gamma(g)[A]=U(g)AU^*(g)$
Recalling the classification of zero-dimensional systems, we have on each site the pure states $\langle h'| \cdot |h'\rangle$ with charge $h'$ relative to the state $\langle 0| \cdot |0\rangle$. 

Given an operator $ O \in \BB(\HH)$, we write $O_j$ for $O\otimes \id_{\caA_{j^c}}$, viewing $O$ as an element in $\caA_j$. Similarly, given an operator $O \in \BB(\HH) \otimes \BB(\HH) $, we write $O_{j,j+1}$ for $O\otimes \id_{\caA_{\{j,j+1\}^c}}\in\caA_{ \{j,j+1\}  }$. This will be used now to construct the TDI.
We define the charge pair creation and annihilation operators $a(h),a^*(h)$ in $\BB(\HH) \otimes \BB(\HH) $ by 
\begin{align}
 a(h')=\ket{0,0}\bra{-h',h'},\qquad  	a^*(h')&=\ket{-h',h'}\bra{0,0}.
\end{align}
Now we define two families of interactions $E(h'),O(h')$ parameterized by $h'$:
\begin{align}
	E(h')_{\{i,i+1 \}} & =  (a(h')+a^*(h'))_{i,i+1},\qquad \text{if $i$ is even}     \\
	O(h')_{\{i,i+1 \}} & = (a(h')+a^*(h'))_{i,i+1},\qquad \text{if $i$ is odd}  .
\end{align}
and $E(h')_S,O(h')_S=0$ for all other $S$. 
The TDI  $H=H(s)$ is then 
\begin{equation}
	H(s)= 	\pi \begin{cases}
	  E(h)& s\in[0,1/2]\\
		  O(-h)&s \in[1/2,1]
	\end{cases}
\end{equation}
The basepoint of the loop is chosen as the product state $\phi$ over zero-charge states defined by $\phi[A_{i}]  = \langle 0| A |0\rangle$ for any $A \in \caB(\caH)$ and $\phi[A_{i}A'_{j}] =  \phi[A_{i}]  \phi[A'_{j}]$ for any $A,A' \in \caB(\caH)$.

To understand the loop, it is helpful to graphically represent $ \phi$ as 
$$
\ldots \ket{0} \otimes \ket{0}\otimes \ket{0} \otimes \ket{0}  \ldots 
$$
During the timespan $ [0,1/2]$, the vector $\ket{0,0}$ is rotated into $\ket{-h,h}$ at pairs $(i,i+1)$ with $i$ even, and vice versa. Vectors orthogonal to the $2$-dimensional space spanned by $\ket{0,0}$ and $\ket{-h,h}$, are left invariant. 
Similarly, during the time $ [1/2,1]$, the vector  $\ket{h,-h}$ is rotated into $\ket{0,0}$ at pairs $(i,i+1)$ with $i$ odd, and vice versa.
This gives hence
$$
\ldots \ket{0} \otimes \ket{0}\otimes \ket{0} \otimes \ket{0} \ldots  \Rightarrow   \ldots \ket{-h} \otimes \ket{h}\otimes \ket{-h} \otimes \ket{h} 
\ldots  \Rightarrow \ldots
  \ket{0} \otimes \ket{0}\otimes \ket{0} \otimes \ket{0} \ldots
$$
 In Figure \ref{fig:ExampleLoop} we give a graphical representation of $\phi,$ $\phi\circ\alpha_H(1/2)$ and $\phi\circ\alpha_H(1)$. This shows that indeed $\phi=\phi\circ\alpha_H(1)$ and that therefore $H$ does generate a loop. To determine the index of the loop, one goes through a similar calculation, to find
$$
\ldots \ket{0} \otimes \ket{0}\otimes \ket{0} \otimes \ket{0} \ldots \Rightarrow \ldots \ket{-h} \otimes \ket{h}\otimes \ket{0} \otimes \ket{0} 
\ldots \Rightarrow \ldots
  \ket{0} \otimes \ket{h}\otimes \ket{0} \otimes \ket{0} \ldots 
$$
The apparant violation of charge conservation is because the opposite charge $-h$ has been evacuated to $-\infty$. It is not visible in the final state, and the index simply measures the charge $h$ of the final state.

\begin{figure}[ht] 
\begin{center}
\includegraphics[width=0.8\textwidth]{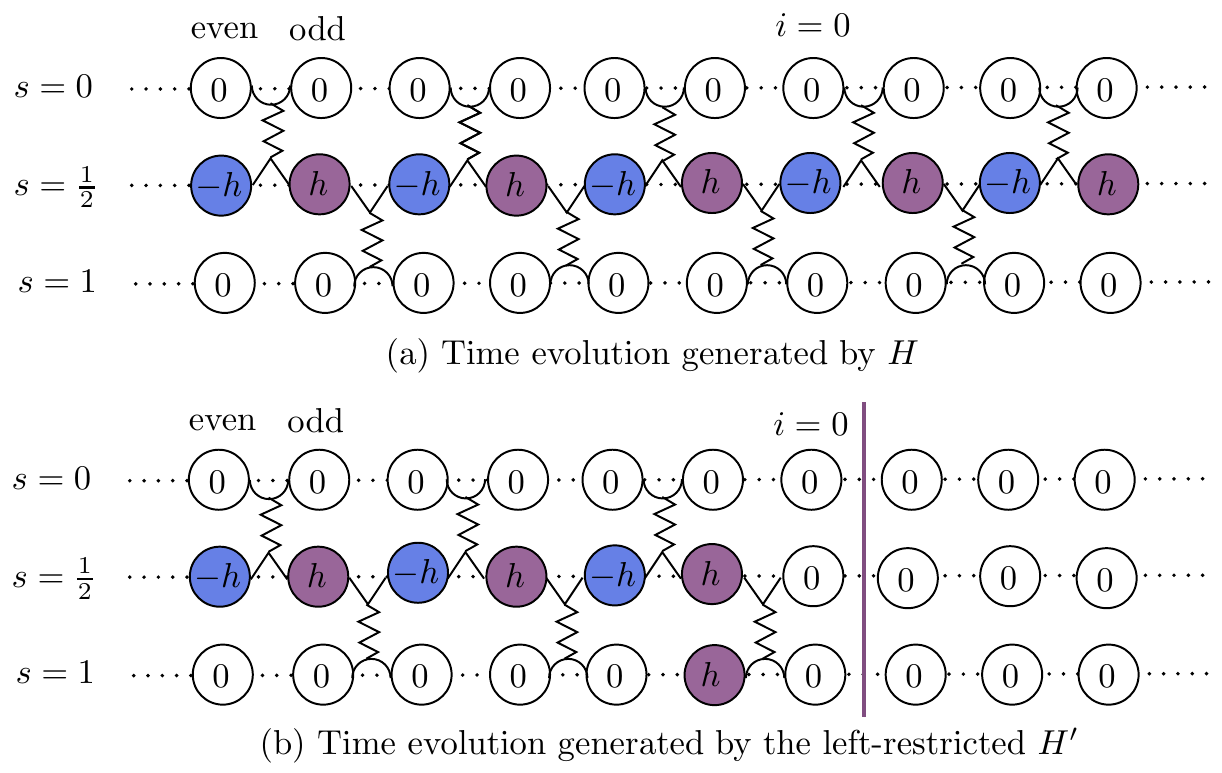}
\caption{Schematic representation of the examples presented in Section \ref{sec: examples}.}
\label{fig:ExampleLoop}
\end{center}
\end{figure}

\subsection{Overview of the proof}\label{sec:plan}
We are now ready to turn to proofs. We start in Section~\ref{sec: technical preliminaries} with locality results on almost local evolutions; although they are known to some level from the literature, full proofs are not readily available in the precise setting we are considering here. In Section~\ref{sec: hilbert space theory}, we describe the Hilbert space setting obtained through the GNS representation. By comparing the initial state and the pumped state, this allows us to define an index following the lines sketched above and show that it is invariant under small perturbations of the TDI generating it. Section~\ref{sec: trivial index loop is short} starts the analysis of homotopies by showing that a loop with product basepoint can be deformed to what we call a short loop, namely a loop generated by a TDI with small norm. We show that this is doable through a $G$-homotopy provided the index vanishes. Section~\ref{sec: contractibility of short loops with product basepoint} finalizes this step by proving that a short loop can be further deformed to a constant loop; unlike the previous section which relies on splitting the chain, this section uses perturbative techniques on gapped ground states of quantum spin chains. Section~\ref{sec: classification for product loops} concludes the proof of our main result in the case of a fixed product basepoint, by proving the claimed additivities on the one hand, and showing that loops have the same index if and only if they are homotopic. Finally, Section~\ref{sec: proofs of main} extends the results to loops with arbitrary ($G$-invertible) basepoints, concluding the proof of Theorem~\ref{thm: classification loops}.

\section{Technical preliminaries}   \label{sec: technical preliminaries}
In this section, we establish some technical facts and tools concerning interactions and TDIs.

\subsection{Anchored interactions and TDIs} \label{sec: anchored interactions}

In Section \ref{sec: setup}, we defined interactions and we equipped them with a family of norms $||\cdot||_f$. We will often need to express that interactions are restricted to the vicinity of a region $X\subset\bbZ$. We say that an interaction $F$ is \emph{$X$-anchored} if 
$$
S \cap X=\emptyset\Rightarrow  F_S=0
$$
and it will be convenient to have an associated norm
\begin{equation}\label{AnchoredInt}
||F||_{X,f}=\begin{cases} ||F||_{f}   & \text{$F$ is $X$-anchored} \\
 \infty  &   \text{otherwise}
\end{cases}
\end{equation}
If  $ F$ is $X$-anchored for a finite $X$, 
then $F$ determines in a natural way an element $\iota(F)$ of $\caA$, see Section~\ref{subsec:alal}, where it was used primarily to define the almost local algebra $\caal$.

Anchored TDIs are defined in a similar fashion, with norm $
\nor H \nor_{X,f}=\sup_{s\in[0,1]} ||H(s)||_{X,f}$. We note that if $X$ is a finite set and $\nor H \nor_{X,f}<\infty$, the almost local evolution is generated by a time-dependent family of self-adjoint elements of $\caA$, namely $\iota(H(s))$, see Lemma~\ref{lem: evolution by similar tdi}.

\subsection{Manipulating interactions} \label{sec: manipulating interactions}
We need some results about manipulating interactions and TDIs.
We first introduce these manipulations algebraically and then we state most of the relevant bounds. 
\paragraph{Commutators} Given interactions $H,H'$ we define the commutator interaction
\begin{equation}\label{commutator of interactions}
([H,H'])_S= \sum_{\substack{ S_1,S_2:  S_1\cup S_2=S \\  S_1 \cap S_2 \neq\emptyset} }   [H_{S_1},H'_{S_2}]
\end{equation}
The same symbol was already used in Section \ref{subsec:alal}  to denote the action of an interaction on $\caal$. These two usages are consistent with each other:  If $H'$ is anchored in a finite set, then $\iota(H')\in\caal$. In particular, $[H,\iota(H')]$ is well defined and $\iota([H,H'])=[H,\iota(H')]$. If also $H$ is anchored in a finite set, then the latter expression equals  $[\iota(H),\iota(H')]$ with $[\cdot,\cdot]$ the standard commutator on $\caA$.

\paragraph{Derivations}\label{sec: derivations}
The action of an interaction $H$ on $\caal$ by $[H,\cdot]$ is an (unbounded) derivation. It is the derivation, rather then the interaction itself, that determines the almost local evolution $\alpha_H$. 
One should notice that distinct $H\neq H'$ can give rise to the same derivation, which complicates some of our calculations. Therefore, we introduce the relation $\eder$ to indicate equality of derivations;
$$
H \eder H' \quad \text{iff} \quad   [H,A]=[H',A] \quad \forall A \in\caal
$$

\paragraph{Weak sums} 
Given a family of interactions $H^{(n)}$ indexed by $n\in \bbZ$, we define the weak sum $\sum_n H^{(n)}$ by
\begin{equation}\label{Def WeakSum}
   (\sum_n H^{(n)})_S=  \sum_n H^{(n)}_S, 
\end{equation}
provided the right-hand side is absolutely norm-convergent for every $S$. 
\subsection{Automorphisms}\label{sec: automorphisms}
Since we consider spin chain algebras $\caA$ where for every finite $S$, $\caA_S$ is finite-dimensional, there is a well-defined \emph{tracial} state $\tau$ on $\caA$. For a region X, we write  $ \tau_{X^c}:\caA \to \caA$ for the corresponding conditional expectation $\id_{\caA_X}\otimes \tau_{X^c}$. This linear map is contracting:  $|| \tau_{X^c}[A] || \leq ||A||$.  We define a canonical decomposition of observables  $A\in\caA$ into finitely suppported terms centered at a given site $j \in\bbZ$, $A=\sum_{k\in\bbN} A_{j,k}$, defined by 
 \begin{equation}\label{eq: decomposition finite}
A_{j,k}=\begin{cases}  \tau_{B^c_{k}(j)}[A]- \tau_{B^c_{k-1}(j)} [A]  &  k>0  \\
 \tau_{B^c_{0}(j)}[A] & k=0
 \end{cases},
\end{equation}
where we used the balls $B_k(j)=\{  \dist(\cdot,j)\leq k\}$. 
Given an interaction $F$ that is $X$-anchored and a $*$-automorphism $\beta$, we define $\beta[F]$ by
\begin{equation} \label{eq: evolved interaction}
(\beta[F])_{B_k(j)}=\sum_{S:  j\in S \cap X} \frac{1}{|X\cap S|}     (\beta[F_S])_{j,k}
\end{equation}
for any $k\in\bbN$ and $j\in X$. For $S$ that are not of the form $B_k(j)$ with $j\in X$, we set $(\beta[F])_{S}=0$. 
This definition satisfies the following constraint on the associated derivations. 
\begin{equation} \label{eq: commutation autos}
[\beta[F],\cdot]= \beta \circ [F,\cdot] \circ \beta^{-1}
\end{equation}
and the integrated version for a TDI $H$ follows from Definition~(\ref{eq: heisenberg})
\begin{equation} \label{eq: commutation autos integrated}
\alpha_{\beta[H]}= \beta \circ \alpha_H \circ \beta^{-1}
\end{equation}
where $\beta[H](s)=\beta[H(s)]$. 
However, the above definition of $\beta[F]$ is of course not the unique definition satisfying these properties. For example, if $F$ is $X$-anchored, then it is also $X'$-anchored, whenever $ X\subset X'$. The transformed interaction resulting  from the definition \eqref{eq: evolved interaction} does in general change upon replacing $X$ by $X'$, whereas the associated derivations are equal.  
Let us illustrate a particular  and relevant consequence of this. Let $\alpha_H$ be an almost local evolution and $F$ an interaction, so that the family 
$\alpha_H(s)[F]$ is defined by \eqref{eq: evolved interaction} for every $s\in [0,1]$.  Then in general the Heisenberg-like evolution equation $
 \alpha_H(s)[F] - F=  i\int_{0}^{s}  du \,   \alpha_H(u) \{[H(u),F]\}
$ 
is \emph{not} correct as an equation between interactions, but the equality does hold on the level of derivations, i.e.\ we do have
$$
 \alpha_H(s)[F] - F \eder  i\int_{0}^{s} du  \,  \alpha_H(u)  \{[H(u),F]\}
$$ 
with $\eder$ introduced in Section \ref{sec: derivations}.

\subsection{Bounds on transformed interactions}\label{sec: bounds on transformed interactions}

To state bounds on transformed interactions, we will need to adjust the functions $f\in\caF$ as we go along.   It would be too cumbersome to define each time explicity the adjusted function, and therefore we introduce the following generic notion:

\begin{definition}\label{def: derived function}
 We say that $\hat{f}$ is \emph{affiliated} to $f$ (and depends on some parameter $C$) if
\begin{enumerate}
\item Both $f,\hat{f}$ are elements of $\caF$ and $\hat{f}$ depends only on $f$ (and on the the parameter $C$).  
\item If $f(r)\leq C_\beta e^{-c_\beta r^\beta}$ for some $\beta<1$ and $c_\beta,C_\beta>0$, then $\hat{f}(r)\leq C'_\beta e^{-c'_\beta r^\beta}$ for some $C'_\beta,c'_\beta>0$.
\end{enumerate}
\end{definition}
We will also use repeatedly that $\max(f_1,f_2) \in \caF$ if $f_1,f_2 \in\caF$.

\begin{lemma}\label{lem: loc and liebrobinson}
Let $F,F_{n}, n\in \bbZ$ be interactions and $H$ a TDI. Let $X,X_n$ be (possibly infinite) regions and let $f,f_1,f_2$ be elements of $\caF$. 
\begin{enumerate}
\item There is a function $\hat f$ affiliated to $\max\{f_1,f_2\}$ such that 
\begin{equation*}
||[F_1,F_2]||_{X_1,\hat f}\leq \hat f(1+\mathrm{dist}(X_1,X_2))||F_1||_{X_1,f_1}||F_2||_{X_2,f_2}.
\end{equation*}
\item If all $X_n$ are mutually disjoint, then 
$$||\sum_{n}  F^{(n)}||_{X,\hat{f}} \leq \sup_n || F^{(n)}||_{X_n,f},\qquad  X=\cup_n X_n$$
for $\hat{f}$ affiliated to $f$.
\item  Assume that $|||H|||_{f_1} \leq C_H $ for some $C_H<\infty$. Then
$$ ||  \alpha_H(s)[F] ||_{X,\hat{f}} \leq     ||F||_{X,f_2},  \qquad s\in [0,1]
$$
for $\hat{f}$ affiliated to $f=\max(f_1,f_2)$ and depending on $C_H$.
\end{enumerate}
\end{lemma}

\begin{proof}
(i) Let 
$$
f=\max(f_1,f_2), \qquad g(r)=2\max_{r_1+r_2=r+1} f(r_1)f(r_2), \qquad \hat f=\sqrt{g}.
$$
One checks that $f, g$ and $\hat{f}$ are in $\caF$. We also write $f_i(S_i)=f_i(1+\diam(S_i))$ for brevity. Definition \eqref{commutator of interactions} implies immediately that the commutator is anchored in $X_1$. We then estimate
\begin{align*}
 || [F_1,F_2]||_{X_1,\hat f}  &\leq  \sum_{ \substack{S_1 \cap X_1 \neq \emptyset , S_2 \cap X_2 \neq \emptyset \\   S_1 \cap S_2 \neq \emptyset}}      \frac{||(F_1)_{S_1}||}{f_1(S_1)}   \frac{||(F_2)_{S_2}||}{f_1(S_2)}   \frac{2 f_1(S_1)f_2(S_2)}{\hat f(1+\diam(S_1\cup S_2))}  .
\end{align*}
We now bound
$$2f_1(S_1)f_2(S_2)\leq g(1+\diam(S_1\cup S_2)) \leq \hat f(1+\dist(X_1,X_2)) \hat f(1+\diam(S_1\cup S_2)), $$ where we used that $\diam(S_1\cup S_2)\geq \dist(X_1,X_2)$,  and the claim follows.

(ii) Recalling Definition~(\ref{Def WeakSum}), we pick $S\subset\bbZ$ and we must estimate the norm of $\sum_{n} F^{(n)}_S$. Since the sets $X_n$ are mutually disjoint, there are at most $\mathrm{diam}(S)$ indices $n$ in the sum such that $F^{(n)}_S\neq0$. Hence
\begin{equation*}
    \Big\Vert\sum_n F^{(n)}_S \Big\Vert
    \leq \mathrm{diam}(S) \sup_{n}\Vert F^{(n)}_S \Vert
    \leq \mathrm{diam}(S) f(1+\mathrm{diam}(S))\sup_{n}  || F^{(n)} ||_{X_n,f},
\end{equation*}
which yields the claim since $\mathrm{diam}(S) f(1+\mathrm{diam}(S))\leq \hat{f}(1+\mathrm{diam}(S))$ with $\hat f$ affiliated to~$f$. Item (iii) is a corollary of the Lieb-Robinson bound \cite{Lieb:1972ts} . We refer to similar\footnote{One of the differences with the cited references is the use of slightly different norms. In the appendix of~\cite{Stability} it is explained how to accommodate this.} statements in~\cite{nachtergaele2019quasi,bachmann2022trotter} and their proofs, see also the proof of Lemma~\ref{lem: evolution by similar tdi}.
\end{proof}

\subsection{Inversion, time-reversal, and composition of almost local evolutions}\label{sec:composition}

We will often need to invert, time-reverse, and compose almost local evolutions. For TDIs $H,H_1,H_2$, we set
\begin{align}
(\alpha_H^{-1})(s) &= (\alpha_H(s))^{-1},\label{def: inverse} \\
(\alpha^{\theta}_H)(s) & = \alpha^{-1}_H(1)\circ \alpha_H(1-s), \label{def: reversal ale} \\
(\alpha_{H_2}\circ\alpha_{H_1})(s) &= \alpha_{H_2}(s)\circ\alpha_{H_1}(s).  \label{def: composition}
\end{align}
 
\begin{lemma}\label{lem: manipulation of evolutions}
The families of automorphisms defined in \eqref{def: reversal ale},\eqref{def: inverse} and \eqref{def: composition}, parametrized by $s\in [0,1]$ are almost local evolutions. More precisely: 
\begin{enumerate}
\item The TDI $-\alpha_H(s)[H(s)]$ generates the inverse $\alpha_H^{-1}$.
\item  The TDI $-H(1-s)$ generates  $\alpha^{\theta}_H$, the time-reversal of $\alpha_H$. 
\item The TDI $H_1(s)+\alpha_{H_1}^{-1}(s)[H_2(s)]$ generates the composition $ \alpha_{H_2}\circ\alpha_{H_1}$
\end{enumerate}
\end{lemma}
\begin{proof}
The proposed families of interactions have a finite $|||\cdot|||_{f}$-norm for some $f\in\caF$ by Lemma~\ref{lem: loc and liebrobinson}. They are strongly measurable: If $H(s)$ is strongly measurable, it is the pointwise limit of a piecewise constant interaction; The definition~(\ref{eq: heisenberg}) and the bounds of Lemma~\ref{lem: loc and liebrobinson} further imply that $\alpha_H(s)$ can be approximated by a piecewise constant automorphism. The fact that they generate the given families of automorphisms follows from a direct computation.
\end{proof}
If $\psi(\cdot)$ is a loop generated by a TDI $H$, then the time-reversal $\alpha^{\theta}_H$ generates the time-reversed loop ${\psi^{\theta}}(\cdot)$ that was defined in Section \ref{sec: concatenation loops}. 
Indeed, using that $\psi$ and $\psi^\theta$ have the same basepoint $\psi(0)$, we verify
$$
 {\psi}(0)\circ \alpha^{\theta}_H(s)= {\psi}(0)\circ    \alpha^{-1}_H(1)\circ \alpha_H(1-s) =  {\psi}(0)\circ  \alpha_H(1-s)= \psi(1-s)=\psi^\theta(s). 
$$

We now introduce the notion of a cocyle, associated to an almost local evolution.  
The cocycle $\alpha_H(s,u)$, with $u\leq s$, is defined as
\begin{equation}\label{cocylce bis}
\alpha_H(s,u) = \alpha_H^{-1}(u)\circ\alpha_H(s)
\end{equation}
so that in particular $\alpha_H(s)=\alpha_H(s,0)$. It corresponds to the evolution from time $u$ to time $s$. The cocyle satisfies the cocyle  equation
\begin{equation}\label{cocylce equation}
 \alpha_H(s,u)=    \alpha_H(s',u) \circ \alpha_H(s,s').
\end{equation}
and the evolution equation
$$
 \alpha_H(s,u)[A] = A+i \int_u^s ds' \alpha_H(s',u)\{[H(s'),A]\}.
$$ 
In particular, $\alpha_H(s,u)$ satisfies the same bound as $\alpha_H(s)$ in Lemma~\ref{lem: loc and liebrobinson}(iii). Now we can then state the Duhamel principle:
\begin{equation}\label{Duhamel}
 \alpha_{H_2}(s)- \alpha_{H_1}(s)=  i\int_0^s \alpha_{H_2}(u)\big\{\big[H_2(u) - {H_1}(u),\alpha_{H_1}(s,u)\{\cdot\}\big]\big\}du
\end{equation}
which is established in the usual way from the Heisenberg equation \eqref{eq: heisenberg}. 
\begin{lemma}\label{lem: inverse small diff}
If ${H_1},{H_2}$ are TDIs, then 
\begin{enumerate}
\item The TDI  $\alpha_{{H_1}}(s)[{H_2}(s)-{H_1}(s)]$ generates the almost local evolution $ \alpha_{{H_2}}\circ\alpha_{{H_1}}^{-1}$.
\item The TDI 
\begin{equation}
\label{Difference}
{H_2}(s)-{H_1}(s)+ i \alpha_{H_2}^{-1}(s)\left[\int_0^s \alpha_{{H_2}}(u)\Big\{\big[{H_2}(u) - H_1(u),\alpha_{H_1}(s,u)\{H_1(s)\}\big]\Big\}du \right]
\end{equation}
generates the almost local evolution $\alpha_{{H_1}}^{-1}\circ\alpha_{{H_2}}$. 
\end{enumerate}
\end{lemma}
\begin{proof}
Both claims follow by applying Lemma \ref{lem: manipulation of evolutions}. For the second claim, one uses the Duhamel principle applied to interactions, more concretely,
\begin{equation} \label{eq: duhamel}
 \alpha_{{H_2}}(s)[H(s)] - \alpha_H(s)[H(s)] \eder i\int_0^s \alpha_{{H_2}}(u)\Big\{\big[{H_2}(u) - H(u),\alpha_{H}(s,u)\{{H}(s)\}\big]\Big\}du.
\end{equation}
\end{proof}

\subsection{Applications of the Duhamel principle}

We establish a few consequences of the bounds in the previous sections. Even though they are straightforward, they will be used so often that it is worthwile to highlight them. First of all, we introduce the $L^1$ norm on TDIs, given by 
$$
||| H |||^{(1)}_{X,f}:= \int_0^1 ds  || H(s) ||_{X,f}
$$
Of, course, we always have $||| H |||^{(1)}_{X,f}  \leq ||| H |||_{X,f}$, but some bounds are stated more naturally in terms of the $L^1$-norm. 
\begin{lemma} \label{lem: evolution by similar tdi}
Let $X,Y$ be (possibly infinite) intervals or complements thereof and let $A\in\caA_Y$.
\begin{enumerate}
\item If $F$ is an interaction anchored in $X$, then
\begin{equation}\label{eq: most simple com bound}
||[F,A]|| \leq  ||A|| ||F||_{X,f} \hat{f}(1+\dist(X,Y)) (1+|X\cap Y|) 
\end{equation}
with $\hat{f}$ affiliated to $f$.
\item
Let $H_1,H_2$ be TDIs satisfying $|||H_1|||_f \leq C_H, |||H_2|||_f \leq C_H$ for some constant $C_H$. If $H_1-H_2$ is anchored in $X$, then 
\begin{equation}\label{eq: duhamel practical}
||\alpha_{H_1}(s)[A]-\alpha_{H_2}(s) [A] || \leq  ||A|| |||H_1-H_2|||^{(1)}_{X,f} \hat{f}(1+\dist(X,Y)) (1+|X\cap Y|)
\end{equation}
with $\hat{f}$ affiliated to $f$ and depending on $C_H$. 
\end{enumerate}
\end{lemma}
\noindent Note that the right hand side of~(\ref{eq: most simple com bound}) is allowed to be infinite, but if it is finite then $[F,A]$ is a well-defined element of $\caA$.
\begin{proof}
We obtain
\begin{align*}
||[F,A]|| &\leq \sum_{y\in Y}\sum_{S \ni y, S \cap X \neq \emptyset } 2 ||F_S|| ||A|| 
 \leq  2 ||A|| \sum_{y\in Y}  f(1+\dist(y,X)) ||F||_{X,f}.
\end{align*}
The bound $2\sum_{y\in Y}  f(1+\dist(y,X))\leq\hat{f}(1+\dist(X,Y)) (1+|X\cap Y|)$ now holds in the particular case of intervals. 
For $ii)$ we use Duhamel's formula~\eqref{Duhamel} to get
\begin{equation}
\alpha_{H_1}(s)[A]-\alpha_{H_2}(s) [A] 
= \iu\int_0^s du \,   \alpha_{H_1}(u) \circ \alpha_{H_2}(s,u)  \big\{[ \alpha^{-1}_{H_2}(s,u)\{H_1(s)-H_2(s)\}, A]\big\}  \label{eq: duhamel practical derivation}
\end{equation}
By the results above, $\alpha^{-1}_{H_2}(s,u)$ is generated by a TDI whose norm is upper bounded by an expression depending on $C_H$. Therefore, we invoke Lemma \ref{lem: loc and liebrobinson} item iii) to get $|| \alpha^{-1}_{H_2}(s,u)\{H_1(s)-H_2(s)\} ||_{X,\hat{f}}\leq   ||H_1(s)-H_2(s) ||_{X,f} $ uniformly in $u,s \leq 1$. 
Using now item i) above, he argument of $\{\cdot\} $ in \eqref{eq: duhamel practical derivation} is in norm bounded by 
$$
 ||H_1(s)-H_2(s) ||_{X,f}  ||A|| \hat{f}(1+\dist(X,Y)) (1+|X\cap Y|),
$$
which proves the claim since $\alpha_{H_1}(u)$ and $ \alpha_{H_2}(s,u)$ preserve the norm.
\end{proof}
The next lemma elaborates on the situation when a TDI $K$ is anchored in a set $X$.   If $X$ is finite, then $\alpha_K(s)$ is an \emph{inner automorphism}, i.e.\ we retrieve the formulas familiar from finite quantum systems:
\begin{equation}\label{eq: finite tdi}
\alpha_K(s)[A]= \Adjoint(V(s))[A]
\end{equation}
with $V(s) \in \caA$ a unitary family that solves  
$$
V(s)=\id+i\int^s_0 du V(u)\iota(K(u)).
$$
\begin{lemma}\label{lem: local perturbation tdi}
Let $H$ be a TDI and $B$ a TDI anchored in $X$, 
\begin{enumerate}
\item There are TDIs $K,K'$ anchored in $X$ such that 
$$
\alpha_{H+B}=\alpha_H \circ  \alpha_K =   \alpha_{K'} \circ  \alpha_H.
$$
 If $H,B$ are $G$-invariant, then so are $K,K'$.  
\item  Let $\psi$ be a pure state and let $X$ be finite. Then the states $\psi\circ\alpha_{H+B} $ and $\psi\circ \alpha_H$ are mutually normal. If $\psi$ and $H,B$ are $G$-invariant, then the relative charge of
$\psi\circ\alpha_{H+B} $ and $\psi\circ \alpha_H$ is zero.
\end{enumerate}
\end{lemma}
\begin{proof}
Item i).  We use Lemma~\ref{lem: manipulation of evolutions} and Lemma~\ref{lem: inverse small diff} to derive expressions for $K,K'$. Then we use Lemma \ref{lem: loc and liebrobinson} to pass the property of being anchored in $X$ from $B$ to $K,K'$. 
Item ii). By i) and the fact that $K$ is anchored on a finite set, the remark preceding the lemma ensures that there is a unitary $V\in\caA$ such that $\psi'=\psi\circ\Adjoint(V)$. In particular they are normal with respect to each other. If $V$ is $G$-invariant, then they have zero relative charge.    
\end{proof}

\subsection{Concatenation of almost local evolutions}\label{sec:Concatenation of evolutions}

First, we remark that one can often construct almost local evolutions by time-rescaling. This is used so often that we put it in a lemma.
\begin{lemma}\label{lem: time-rescaling almost local evolutions}
Let $H$ be a TDI and let $j: [0,1]\to [0,1]$ be a piecewise smooth function, with  $|\frac{dj(s)}{ds}|\leq C_j$ except for a finite set of times $s$. Then $\alpha_H^{-1}(j(0))\circ\alpha_H(j(\cdot))$ is an almost local evolution.
\end{lemma}     
\begin{proof}
We set $K=\frac{d j(s)}{ds} H(s)$, except for the finite set of times where the bound $|\frac{dj(s)}{ds}|\leq C_j$ does not hold, where we set $K(s)=0$.
We see that $|||K|||_f\leq C_j |||H|||_f $ for any $f\in \caF$ so that the TDI satisfies the necessary bound. The fact that it generates the family $\alpha^{-1}_H(j(0))\circ\alpha_H(j(\cdot))$  follows from direct computation.      
\end{proof}
For almost local evolutions $\alpha,\alpha'$, we define the concatenated evolution $\alpha\loopc \alpha'$ 
by 
\begin{equation}\label{concatenation}
(\alpha\loopc \alpha')(s)= \begin{cases} \alpha(2s) & 0\leq s <1/2 \\
\alpha(1)\circ\alpha'(2s-1)   & 1/2\leq s <1   
  \end{cases}
\end{equation}
The family of automorphisms $(\alpha\loopc \alpha')(s)$ is again an almost local evolution by Lemma \ref{lem: manipulation of evolutions}.    
We also note the relation to the concatenation of loops defined in \eqref{eq: concatenation}:  If $\psi,\psi'$ are loops with common basepoint and generated by $\alpha,\alpha'$, then $\psi' \loopc \psi$ is a loop with the same basepoint and generated by $\alpha\loopc \alpha'$. 
We will now state that composition of almost local evolutions is, in a certain  sense, homotopic, to concatenation of almost local evolutions.
\begin{definition}\label{def: homo of ale}
A family $(s,\lambda)\mapsto \alpha^{(\lambda)}(s)$ is a $G$-homotopy of almost local evolutions if 
$$
\sigma^{(0)}(\lambda) \circ \alpha^{(\lambda)}(s) 
=  \alpha^{(0)}(s) \circ \sigma^{(s)}(\lambda)  ,\qquad \text{for any $s,\lambda \in [0,1]$}
$$  
with $\alpha^{(\lambda)},\sigma^{(s)}$, for each $\lambda,s$,  generated by $G$-invariant TDIs $H_\lambda, F_s$ satisfying the uniform boundedness property \eqref{Homotopy: Uniform bound} and such that 
\begin{equation}\label{eq: extra condition homotopy}
\sigma^{(0)}(\lambda)=\sigma^{(1)}(\lambda)=\Id
\end{equation}
\end{definition}
In Section~\ref{sec: homotopy} the homotopy was defined as a property of a function $(s,\lambda)\mapsto \psi_\lambda(s)$ of \emph{states}. The connection with the above definition is the following: If $(s,\lambda)\mapsto \alpha^{(\lambda)}(s)$ is a homotopy of almost local evolutions, then $(s,\lambda)\mapsto \nu\circ\alpha^{(\lambda)}(s)$ is a homotopy of states, for any choice of state $\nu$. Moreover, if $s \mapsto \nu\circ\alpha^{(\lambda)}(s)$ is a loop for some $\lambda$, then it is a loop (with the same basepoint $\nu$) for any $\lambda$.
To complete the vocabulary, two almost local evolutions $\alpha^{(0)},\alpha^{(1)}$ are called $G$-homotopic if there exists a $G$-homotopy $\alpha^{(\lambda)}$ reducing to $\alpha^{(0)}$ for $\lambda=0$ and to $\alpha^{(1)}$ for $\lambda=1$. 
The main reason to consider homotopy of almost local evolutions, is given by the following lemma.
\begin{lemma}\label{lem: composition of loops is homotopic to concatenation}
\begin{enumerate}
\item Let $H$ be a $G$-invariant TDI and let $j: [0,1]\to [0,1]$ be a piecewise smooth function, whose derivative is bounded (except for a finite set), and such that $j(0)=0$ and $j(1)=1$. Then $\alpha_H(\cdot)$ and $\alpha_H(j(\cdot))$ are $G$-homotopic. 
\item Let $H_1,H_2$ be two $G$-invariant TDIs.  Then $\alpha_{H_2} \circ \alpha_{H_1}$ is $G$-homotopic to~$\alpha_{H_2} \square \alpha_{H_1}$.
\end{enumerate}
 \end{lemma}
\begin{proof}
Item i)  We consider the family
$$
\alpha^{\lambda}(s) =      \alpha( k(s,\lambda)), \qquad   k(s,\lambda)=  (1-\lambda)s+ \lambda j(s)
$$
For any $s$, the function $k(s,\cdot)$ satisfies the conditions of Lemma \ref{lem: time-rescaling almost local evolutions} and hence it yields the required almost local evolution 
$\sigma^{(s)}(\cdot)=\alpha^{-1}(s) \alpha(k(s,\cdot))$.
The condition \eqref{eq: extra condition homotopy} follows from the fact that $j(0)=0$ and $j(1)=1$.\\
Item ii)  For simplicity, let us denote $\alpha(s)=\alpha_{H_1}(s)$ and $\beta(s)=\alpha_{H_2}(s)$.
We will construct a family $\alpha^\lambda(s)$ interpolating between $\beta\circ\alpha$ and $\beta \loopc \alpha$ as $\lambda$ ranges between $0$ and $1$. Let 
$$
d^{(\lambda)}=\tfrac{1}{2}(1+\lambda,1-\lambda) \in [0,1]^2,\qquad \lambda\in[0,1].
$$
Let $j^{(\lambda)}: [0,1] \mapsto [0,1]^2:  s \mapsto j^{(\lambda)}(s)=(j_1^{\lambda}(s),j_2^{(\lambda)}(s)) $ be the continuous and piecewise smooth map defined by 
$$
j^{(\lambda)}(s)=
\begin{cases}
  2s d^{(\lambda)}  &   s\leq 1/2 \\
  d^{(\lambda)} + 2(s-\tfrac{1}{2})\left((1,1) - d^{(\lambda)} \right)  & s>1/2  
\end{cases}.
$$
We note that
\begin{equation}\label{eq: extreme values for j}
j^{(0)}(s)=(s,s),\qquad    j^{(1)}(s)=
\begin{cases}
 (2s,0)  &   s\leq 1/2 \\
  (1,2s-1)  & s>1/2  
\end{cases}
\end{equation}
Then we define the family of automorphisms
$$ 
\alpha^{(\lambda)}(s)=   \beta(j_1^{(\lambda)}(s) ) \circ \alpha(j_2^{(\lambda)}(s) ) ,\qquad \lambda\in[0,1].
$$
We verify
\begin{enumerate}
\item $\alpha^{(\lambda)}(0)= \Id$ and $\alpha^{(\lambda)}(1)= \beta(1 ) \circ \alpha(1 )$ because $j^{(\lambda)}(0)=(0,0), j^{(\lambda)}(1)=(1,1)$.
\item $\alpha^{(0)}(s)=  (\beta\circ \alpha)(s )$. 
\item $\alpha^{(1)}(s)=  (\beta \loopc \alpha)(s)$. 
\item For each $\lambda$, $\alpha^{(\lambda)}$ is an almost local evolution.
\item $\alpha^{(\lambda)}(s)=    \alpha^{(0)}(s) \circ \sigma^{(s)}(\lambda) $, with 
$$
\sigma^{(s)}(\lambda)=\alpha(s )^{-1} \circ  \beta(s )^{-1}  \circ \beta(j_1^{(\lambda)}(s) ) \circ \alpha(j_2^{(\lambda)}(s) )    
$$
\item For each $s$,  $\sigma^{(s)}$ is an almost local evolution. 
\item   $\sigma^{(0)}(\lambda)=\sigma^{(1)}(\lambda)=\Id $, so item (v) means that $(s,\lambda)\mapsto \alpha^{(\lambda)}(s)$ is a homotopy.
\end{enumerate}
Item (i--iii) and (vii) are immediate consequences of the definition, using~\eqref{eq: extreme values for j}.  To check~(iv), 
we use Lemma \ref{lem: time-rescaling almost local evolutions} to conclude that both $\beta(j_1^{(\lambda)}(\cdot) )$ and $ \alpha(j_2^{(\lambda)}(\cdot) ) $ are almost local evolutions and so is their composition by Lemma~\ref{lem: manipulation of evolutions}.  
Item (v) follows from the definition of $\alpha^{(\lambda)}(s)$ and~(ii).
Item (vi) follows from Lemma  \ref{lem: time-rescaling almost local evolutions} upon noting that $\lambda \mapsto j^{(\lambda)}(s)$ has uniformly bounded derivatives and $j^{(0)}(s)=(s,s)$.
\end{proof}

\section{Hilbert space theory for states equivalent to a product state} \label{sec: hilbert space theory}

Although most of our reasoning stays on the level of the spin chain algebra $\caA$, some steps are easier taken within the GNS representation. In this section, we establish the tools that we will need from Hilbert space theory.

\subsection{GNS construction}\label{sec: representations}
We will make use of the GNS representation of state $\psi$ on $\caA$, see \cite{BratRob}. As before, we denote the GNS triple associated to $\psi$ by   $(\caH,\Omega,\pi)$ 
 (Hilbert space $\caH$,  normalized vector $\Omega\in\caH$,  $*$-representation  $\pi:\caA\to\caB(\caH)$), such that 
$$\psi(A)=\langle\Omega, \pi[A]\Omega\rangle, \qquad \text{ for any $A\in\caA$}.  $$     From the purity of $\psi$, it follows that $\pi(\caA)'=\bbC 1$ and hence the von Neumann algebra $\pi(\caA)''$ equals the full algebra $\caB(\caH)$. The following lemma, already used in Section \ref{sec: g charge}, verifies that we are in the setting outlined in Section \ref{sec: equivalence zero}; We omit its standard proof which relies on the uniqueness of the GNS representation.
\begin{lemma}\label{lem: gns} If $\psi$ is $G$-invariant, then there is a unique strongly continuous unitary representation $G\to \caU(\caH): g\mapsto U(g)$, such that $U(g)\Omega=\Omega$ and
\begin{equation} \label{eq: rep of g action}
\pi\circ \gamma(g)= \Adjoint(U(g))\circ \pi
\end{equation}
\end{lemma} 
\noindent We further recall that two pure states $\psi,\psi'$ are said to be mutually normal\footnote{The term `equivalent' is used more often in the literature, but we avoid this term since it has already another natural meaning within this paper. However, in the zero-dimensional case, the two notions actually coincide.} iff\  $\psi'$ is represented as a pure density matrix $\rho_{\psi'}$ in the GNS representation of $\psi$: $\psi'(A) = \Tr_\caH(\rho_{\psi'}\pi(A))$. It then follows that the same holds with the roles of $\psi,\psi'$ reversed. 

\subsection{GNS representation of a pure product state}\label{sec:GNS of product state}
We construct the GNS representation of a pure, $G$-invariant product state $\phi$ in a pedagogical way to render the objects that follow more tangible.
Recall that $\caA_i$ is a matrix algebra of $n_i\times n_i$ complex matrices, that we write as $\caA_i= \caB(\bbC^{n_i})$ For every $i$, we choose an orthonormal basis of $\bbC^{n_i}$ labelled by $\sigma=0,\ldots, n_i-1$ such that 
$$
\phi(A)=\langle 0|  A | 0\rangle,\qquad A\in \caA_i.
$$
The state $\phi$ corresponds hence intuitively to the --- a priori ill-defined --- vector 
\begin{equation}\label{eq: gns intuitive}
\ldots \otimes | 0\rangle_{i-1} \otimes | 0\rangle_i \otimes | 0\rangle_{i+1} \otimes \ldots
\end{equation}
The representation Hilbert space $\caH$ is chosen as 
$$
\caH := l^2(M_\bbZ)
$$
where $M_X$, for $X\subset\bbZ$, is the space of sequences $m:X\to \{0,1,2\ldots\}$ such that $m(j) < n_j$ and the set $\{j: m(j)\neq 0\}$ is finite. We write henceforth $\caH_X=\ell^2(M_X)$  and we note the tensor product structure $\caH=\caH_X \otimes \caH_{X^c}$.
We now define an isometry $K_S: \otimes_{j\in S}\bbC^{n_j}\to\caH_S$ for finite $S$, by 
$$
K_S(\otimes_{j\in S}|m_j\rangle)(m')=\begin{cases}  1 &  m'=m \\
0 &   m'\neq m
 \end{cases},\qquad   m,m'\in M_S
$$
and we lift it to an isomorphism of $C^*$-algebras by setting
$$
\pi_S: \caA_S \to \caB(\caH_S):  A \mapsto K_SA K_S^{-1}.
$$
 Then the representation $\pi:\caA\to\caB(\caH)$ is fixed by the requirement that its restrictions 
$$
\pi|_{\caA_S}:  \caA_S \otimes \id_{S^c} \to   \caB(\caH_S))\otimes \id_{\caH_{S^c}}
$$
coincide with $\pi_S$, for finite $S$.  We note that for finite $S$, $\pi_S$ is an isomorphism between $\caA_S$ and $\caB(\caH_S)$ , whereas for infinite $X$,  $\pi(\caA_X)$ is a strict subset of $\caB(\caH_X)$. In the former case, the dual space of bounded linear functionals on $\caA_S$ is isomorphic to the trace-class operators on $\caH_S$, equipped with the trace norm $||\cdot||_1$.
%
%
The state $\phi$ is now represented by the vector $\Omega=\delta_{m_0,\cdot}\in\caH$ with $m_0(j)=0$ for all $j$; 
$$
\phi(A)=\langle\Omega, \pi[A]\Omega\rangle
$$
and we see that $\Omega$ indeed corresponds to the heuristic \eqref{eq: gns intuitive}, and it factors as $\Omega=\Omega_X\otimes\Omega_{X^c}$.  We denote by $\Pi_X$ (and again $\Pi = \Pi_\bbZ$) the orthogonal projection on the range of $\Omega_X$ and we abuse notation by also writing $\Pi_X$ to denote $\Pi_X\otimes 1_{X^c}$.

We now state a relevant condition for a pure state $\psi$ to be normal with respect to the product state $\phi$. We fix intervals in $\bbZ$ centered on the sites $\{0,1\}$;
\begin{equation}\label{eq: intervals} 
I_1=\emptyset, \qquad I_r=[2-r,r-1],\qquad  r>1.
\end{equation}
Then
\begin{lemma}\label{lem: condition for normality}
Let $\phi,\psi$ be pure states on $\caA$ and let $\phi$ be a product.
If $||(\psi-\phi)_{I_r^c}|| \to 0$ as $r\to \infty$, then $\psi$ is normal w.r.t.\ $\phi$. 
\end{lemma}
\begin{proof}
We use the notation introduced above. 
Let $\rho_{I_r}$ be the density matrix on the finite-dimensional Hilbert space $\caH_{I_r}$ representing the state $\psi_{I_r}$. 
We consider the sequence
\begin{equation}\label{eq: converging sequence of states}
\rho_{I_r} \otimes \Pi_{I_r^c}
\end{equation}
of density matrices on $\caB(\caH)$ indexed by $r$,
and we prove now that this sequence is a Cauchy sequence in trace-norm $||\cdot||_1$.
Let us denote $\delta(r)=||(\psi-\phi)_{I_r^c}|| $ and abbreviate 
$\rho_r=\rho_{I_r}$, $\Pi_{r,r'}=\Pi_{I_{r'}\setminus I_r}$ for $r<r'$, and let $\Tr_{r'}, \Tr_{r,r'}$ be the traces on $\caH_{I_{r'}}, \caH_{I_{r'}\setminus I_r}$, respectively. 
The key observation is that, writing $P=1-\Pi_{r,r'}$,
$$
  || \rho_{r'} P ||^2_1  \leq  \Tr_{r'}[\rho_{r'} P ] = \psi(P) = (\psi - \phi)(P)
$$
where we used the noncommutative H{\"o}lder inequality  $||AB||_1\leq  ||A||_2||B||_2 $ with $A=\sqrt{\rho_r'},B=\sqrt{\rho_r'} P  $ and $P=P^2$, to get the first inequality.
Therefore, 
\begin{equation}\label{eq: workhorse state restrictions}
    || \rho_{r'} (1-\Pi_{r,r'}) ||_1 \leq \sqrt{\delta}.  
\end{equation}
Since $\rho_{r}=  \Tr_{r,r'}[\rho_{r'}]$, we have that
\begin{align}
 \rho_{r}\otimes \Pi_{r,r'}&=  \Tr_{r,r'}[\Pi_{r,r'}\rho_{r'} \Pi_{r,r'}] \otimes \Pi_{r,r'} +E_1 \\
 &=  \Pi_{r,r'}\rho_{r'} \Pi_{r,r'} +E_1
=  \rho_{r'} +E_1+E_2
 \end{align} 
where $||E_1||_1 \leq \sqrt{\delta}$ and $||E_2||_1 \leq 3\sqrt{\delta}$, by \eqref{eq: workhorse state restrictions}. 
This shows that \eqref{eq: converging sequence of states} is indeed Cauchy, and hence it converges to a density matrix $\rho$. Therefore, the normal states on $\caA$, induced by the density matrices \eqref{eq: converging sequence of states} form a Cauchy sequence as well and they have a limit. The limit is necessarily equal to $\psi$ since $\psi$ is obviously a limit point (hence the unique limit point)  in the weak$^*$-topology. We conclude that~$\psi$ is represented by~$\rho$.  
\end{proof}

As remarked in Section \ref{sec: relative charge for cut loop},  this lemma also follows from Corollary~2.6.11 in~\cite{BratRob}.  Unlike there, the proof given here yields quantitative bounds, which we collect now. We use the notations above, in particular $\delta(r)=||(\psi-\phi)_{{I_{r}^c}}||$. Then, taking $r'\to\infty$ in the proof of Lemma \ref{lem: condition for normality}, we have
\begin{equation} \label{eq: loc bound 1}
||\rho-\rho_{I_r} \otimes \Pi_{I_{r}^c}||_1 \leq  4\sqrt{\delta(r)},
\end{equation}
\begin{equation}\label{eq: loc bound 2}
||\rho- \Pi_{I_{r}^c}\rho \Pi_{I_{r}^c} ||_1 \leq  3\sqrt{\delta(r)}.
\end{equation}
From \eqref{eq: loc bound 2}, we also deduce 
\begin{equation} \label{eq: loc bound 3}
||\rho- \frac{\Pi_{I_{r}^c}\rho \Pi_{I_{r}^c}}{\Tr[\Pi_{I_{r}^c}\rho \Pi_{I_{r}^c}]}||_1\leq    12 \sqrt{\delta(r)},  
\end{equation}
provided that $\Tr[\Pi_{I_{r}^c}\rho \Pi_{I_{r}^c}]\neq 0$. The claim is trivial if $12 \sqrt{\delta(r)}\geq 2$. If $6 \sqrt{\delta(r)}< 1$, we denote $\tilde\rho = \Pi_{I_{r}^c}\rho \Pi_{I_{r}^c}$ and have that $\Vert\rho - (\Tr\tilde\rho^{-1})\tilde\rho\Vert_1\leq 3\sqrt{\delta(r)} + \vert\Tr\tilde\rho\vert^{-1}\Vert \tilde\rho\Vert_1\vert 1-\Tr\tilde\rho\vert$, which yields the claim since (\ref{eq: loc bound 2}) implies $\vert 1- \Tr\tilde\rho\vert\leq 3\sqrt{\delta(r)} \leq 1/2$ and further $\Tr\tilde\rho\geq 1/2$.

\subsection{The pumped state and its index}\label{sec: states cut loop}

In this section we analyze the pumped state that was introduced in Section \ref{sec: relative charge for cut loop} for the special case in which the basepoint of the loop is a pure $G$-invariant product $\phi$, as in the previous subsections. 
We fix this state $\phi$ in what follows, and we consider a $G$-invariant TDI $H$ that generates a loop $\phi\circ\alpha_H(\cdot)$, satisfying  $|||H|||_f\leq C_H$ for some $f\in\caF$.  As in Section \ref{sec: relative charge for cut loop}, 
we cut the TDI by defining $H'$ to contain only the terms $H_S$ supported on the left, $S \leq 0$,  obtaining the pumped state
\begin{equation}\label{eq: pumped state}
\psi= \phi\circ \alpha_{H'}(1).
\end{equation}
We first establish that $\psi$ is normal w.r.t\ $\phi$.
\begin{lemma}\label{lem: cut state}
There is $\hat{f}$ affiliated to $f$ and depending on $C_H$ such that 
\begin{equation}\label{eq: two-sided bound}
|| (\psi-\phi)_{I_r^c} || \leq  \hat{f}(r).
\end{equation}
Therefore, as a consequence of Lemma \ref{lem: condition for normality}, $\psi$ is normal w.r.t.\ $\phi$
\end{lemma}
\begin{proof}
Clearly, $\psi$ is a product $\psi=\psi_{\leq 0}\otimes \psi_{\geq 1}$, and $\psi_{\geq 1}=\phi_{\geq 1}$. It follows that 
$$
||(\psi-\phi)_{I_r^c} ||= || (\psi-\phi)_{<2-r} ||=\sup_{A \in\caA_{<2-r}, ||A|| \leq 1}  |\psi(A)-\phi(A)|
$$
and we now estimate the last quantity:
\begin{equation}
| \psi(A) - \phi(A)| =|\phi( \alpha_{H'}(1)[A] - \alpha_H(1)[A]) | \leq  \hat{f}(r) ||A||
 \label{eq: onesided bound}
    \end{equation}
where the equality follows because  $\phi\circ\alpha_{H}(\cdot)$ is a loop,
and the inequality follows from Lemma~\ref{lem: evolution by similar tdi} item ii), with the input that $H-H'$ is anchored in $\{\geq 1\}$ and that all constants that appear can be bounded by functions of $C_H$.  
\end{proof}
 
Since $\psi$ is normal w.r.t.\ $\phi$, and the $G$-action is implemented by a unitary representation, see Lemma \ref{lem: gns},  we can indeed define the index $I(\phi,H)$ introduced in Section \ref{sec: relative charge for cut loop}.  We will now derive some topological properties of the map $H\mapsto I(\phi,H) $.

\subsubsection{Properties of the index map $I(\phi,\cdot)$}

We fix again the product state $\phi$  and its GNS triple. We also fix $f\in\caF$ and a constant $C_H$.  
We let  $\caS=\caS(C_H,f,\phi)$ be the set of $G$-invariant TDIs such that
\begin{enumerate}
\item $|||H|||_f \leq C_H$
\item $\phi\circ\alpha_H(\cdot)$ is a loop
\end{enumerate}
 To any $H \in \caS$, we associate the pumped state $\psi=\phi\circ\alpha_{H'}(1)$. 
 Let  $\caT_{1,*}(\caH)$ the set of pure $G$-invariant density matrices on $\caH$. In particular,  $\rho_\psi$, the density matrix representing the state $\psi$ on $\caH$, is in $\caT_{1,*}(\caH)$.
 We define the function
 $$
 \hat\rho: \caS\to \caT_{1,*}(\caH): H\mapsto \rho_\psi.
 $$
Next, let $g\mapsto U(g)$ be the unitary representation on $\caH$ given in Lemma \ref{lem: gns}. Since $\rho$ is pure, it is a one-dimensional projection and the $G$-invariance implies that $U(g)$ reduces to a one-dimensional representation on the range of the projection. In particular $\Tr( \rho U(g))$ is a complex number on the unit circle. We set 
$$
\hat{h}:  \caT_{1,*}(\caH) \to H^1(G): \rho \mapsto \hat{h}(\rho)
$$
where $\hat{h}(\rho)$ is the homomorphism $G\mapsto \bbS^1$ defined by 
\begin{equation}\label{expression for charge}
 e^{i \hat{h}(\rho)(g)}=\Tr( \rho U(g)) \in U(1).
\end{equation}
Because of the specific choice of $U(\cdot)$, we find that $\hat{h}(\rho_\psi)=h_{\psi/\phi}$. 
We consider $\caS,\caT_{1,*}(\caH), H^1(G)$ with  topologies induced by, respectively, the norm $|||\cdot|||^{(1)}_f$, the trace norm $||\cdot ||_1$, and the metric $d_\infty(h,h')=\sup_{g\in G}|h(g)-h'(g)|_{\bbS^1}$.  With these choices, the map $\hat{h}$ is Lipschitz. Indeed, for any $\rho^{(1)},\rho^{(2)} \in \caT_{1,*}(\caH) $,
$$
d_{\infty}\left(\hat{h}(\rho^{(1)}),\hat{h}(\rho^{(2)})\right)
\leq 
C_{\bbS^1} \sup_{g\in G} \left|\Tr [\rho^{(1)} U(g)] -\Tr [\rho^{(2)} U(g)] \right| \leq C_{\bbS^1} ||\rho^{(1)}-\rho^{(2)}||_1. 
$$
with $ C_{\bbS^1}=\sup_{\theta\in [0,\pi] } \frac{\theta}{ |e^{i\theta}-1|}  <\infty $. 
We now show 
\begin{lemma}\label{lem: uniform cont index}
The maps
$ \hat{\rho}: \caS \to \caT_1(\caH)$ and 
$$
I(\phi,\cdot)=\hat h\circ \hat\rho:   \caS \to H^1(G)  
$$
are uniformly continuous. 
\end{lemma}
\noindent Note that the value of $\hat h\circ \hat\rho$ is the same as the value of the index in Theoem~\ref{thm: pump index} but the map $h$ there differs from $\hat h\circ \hat\rho$ here because it is a function of a loop of states there, while it is here a function of the Hamiltonians generating a particular loop with fixed product basepoint.
\begin{proof}
We have already remarked that $\hat{h}$ is Lipschitz, so both claims will be proven once we show that $\hat{\rho}$ is uniformly continuous. 
Consider $H^{(1)},H^{(2)}\in \caS$ and let $\rho^{(j)}=\hat{\rho}(H^{(j)})$.
We have
\begin{align*}
 ||\rho^{(1)}-\rho^{(2)}||_1  &\leq    ||\rho^{(1)}_{I_r} \otimes\Pi_{I_r^c}-\rho^{(2)}_{I_r}\otimes\Pi_{I_r^c}||_1 +  ||\rho^{(1)}_{I_r}\otimes\Pi_{I_r^c}-\rho^{(1)}||_1  +   ||\rho^{(2)}_{I_r}\otimes\Pi_{I_r^c}-\rho^{(2)}||_1 \\
  &\leq  ||\rho^{(1)}_{I_r} -\rho^{(2)}_{I_r}||_1   +  8\sqrt{\hat{f}(r)}
\end{align*}
where we used the bound \eqref{eq: loc bound 1} with $\delta(r)=\hat{f}(r)$, by Lemma \ref{lem: cut state}. 
To bound  the difference
$||\rho^{(1)}_{I_r} -\rho^{(2)}_{I_r}||_1 $, we rewrite it as
\begin{align}
&\sup_{A \in \caA_{I_r}, ||A||\leq 1} |\phi\circ\alpha_{(H^{(1)})'}(1)[A]-\phi\circ\alpha_{(H^{(2)})'}(1)[A]| 
\end{align}
which is bounded by $ \hat{f}(1)|I_r| |||H^{(1)}-H^{(2)}|||^{(1)}_f   $ by Lemma \ref{lem: evolution by similar tdi}. 
We hence get 
$$
 ||\rho^{(1)}-\rho^{(2)}||_1 \leq    \hat{f}(1)|I_r|\, |||H^{(1)}-H^{(2)}|||^{(1)}_f +   8 \sqrt{\hat f(r)}        
$$
It it important to note that, since $f$ is fixed, $\hat{f}$ can be fixed as well.  Since $\hat f(r) \to 0$ as $r\to 0$, we can now find a modulus of continuity for the function $\hat{\rho}$.
\end{proof}

Finally, we formulate a proposition that will be used crucially in the proof of homotopy invariance of the index map $I(\phi,\cdot)$. 
\begin{proposition}\label{prop: finite index set}
 There is an $ \epsilon>0$, such that, if $|||H^{(1)}-H^{(2)}|||^{(1)}_f \leq \epsilon$ for $H^{(1)}, H^{(2)}\in\caS$, then $I(\phi,H^{(1)})=I(\phi,H^{(2)})$. 
\end{proposition}
\begin{proof}
Let $h\in H^1(G)$. The discreteness of $H^1(G)$, see Section~\ref{sec: g charge}, implies that there is an $\epsilon(h)>0$ such that
$d_{\infty}(h,h')\leq \epsilon(h)$ implies $h=h'$.  However, since the metric on $H^1(G)$ is homogenous, namely $d_{\infty}(h,h') = d_{\infty}(h-h',0)$, we find  $\epsilon(h)=\epsilon(0)$ and one can choose $\epsilon$ uniformly on $H^1(G)$.  The claim of the proposition follows then from the uniform continuity of the index map~$I(\phi,\cdot)$ on $\caS$, i.e.\ from Lemma \ref{lem: uniform cont index}. 
\end{proof}

\subsection{Connecting states by adiabatic flow}

States that are mutually normal, can be connected via an almost local evolution whose generating TDI is anchored in a finite set. In other words, the almost local evolution consists of inner automorphisms. In this section, we establish this statement in a quantitative way, but restricted to the specific setting in which we will need it.  A result very similar to the crucial Lemma \ref{lem: parallel transport} below already appeared in \cite{kapustin2021classification}.

 \subsubsection{Parallel transport for $0$-dimensional systems}\label{Parallel in 0d}
 
 We first exhibit how to connect states on $0$-dimensional systems, i.e.\ states on $\caB(\caH)$, with $\caH$ a Hilbert space. We say that a bounded, measurable function $[0,1]\ni s\mapsto E(s)$ with $E(s)=E^*(s)\in \caB(\caH) $ is a zero-dimensional TDI and we set $|||E|||=\sup_{s\in [0,1]} ||E(s)||$.
The evolution $s\mapsto \alpha_E(s)$ is then defined as the family of automorphisms that is the unique solution to the equation 
$$\alpha_E(s)=\id+i\int_0^sdu \alpha_E(s)\{[E(u),\cdot]\}.$$ 
\begin{lemma}\label{lem: parallel transport simple}
Let $\nu,\omega$ be pure normal states on $\caB(\caH)$.
Then there exist zero-dimensional TDI $E$, such that 
\begin{enumerate}
\item   $|||E||| < 8 || \nu-\omega ||$
\item $\nu = \omega\circ \alpha_E(1)$
\item  If both $\nu,\omega$ are invariant under a $G$-action by automorphisms $\gamma(g)$, and $h_{\nu/\omega}=0$, then $E$ can be chosen $G$-invariant as well.
\end{enumerate}
\end{lemma}
The proof of this lemma uses only basic linear algebra and we postpone it to Appendix \ref{app: tricks}.

\subsubsection{Parallel transport for states normal to $\phi$}
Let $\psi$ be a pure state normal w.r.t.\ the pure product state $\phi$.
Recall the definition of the intervals $I_r$ in \eqref{eq: intervals}. 
\begin{lemma}\label{lem: parallel transport}
Let $\psi$ be a  pure state and $\phi$ a pure product state, such that
\begin{equation}\label{eq: one-sided bounds repeat}
|| (\psi-\phi)_{I_r^c} || \leq f(r),\qquad f \in \caF.
\end{equation}
Then there exist a  TDI $K$ such that 
\begin{enumerate}
\item   $|||K|||_{\{0,1\},\hat{f}} <1$ for $\hat{f}$ affiliated to $f$.  
\item $\psi = \phi\circ \alpha_K(1)$
\item  If both $\phi,\psi$ are $G$-invariant, with zero relative charge $h_{\psi/\phi}=0$, then $K$ can be chosen $G$-invariant.
\end{enumerate}
\end{lemma}
 \begin{proof}
By Lemma~\ref{lem: condition for normality}, $\psi$ is normal w.r.t. $\phi$. Let $\rho$ be the pure density matrix on $\caH$ representing $\psi$.
We set 
\begin{equation} \label{eq: def rho r}
\tilde\mu_r= \frac{\Pi_{I_{r}^c}\rho \Pi_{I_{r}^c}}{\Tr[\Pi_{I_{r}^c}\rho \Pi_{I_{r}^c}]}, \qquad \text{provided  $\Tr[\Pi_{I_{r}^c}\rho \Pi_{I_{r}^c}] \neq 0$  }  
\end{equation}  
By the bounds \eqref{eq: loc bound 2} and \eqref{eq: loc bound 3} with $\delta(r)=f(r)$ and the fact that $f$ is non-increasing, there is a finite $r_0>0$ such that for $r\geq r_0$, the denominator  \eqref{eq: def rho r} is not zero.
For $1 \leq r\leq r_0$, we set simply $\tilde\mu_r=\Pi$ (which coincides with the definition \eqref{eq: def rho r} for $r=1$).  In any case, the density matrices $\tilde\mu_r$ are pure, by the purity of $\rho$, and  of the form $\tilde\mu_r=\mu_r\otimes \Pi_{I_r^c}$ with $\mu_r$ pure as well. 
We have
\begin{align*}
||\mu_{r+1}-\mu_r\otimes \Pi_{I_{r+1}\setminus I_r}||_1 &= ||\tilde\mu_{r+1}-\tilde\mu_r||_1 \\
 &\leq ||\tilde\mu_r-\rho||_1 +||\tilde\mu_{r+1}-\rho||_1 \\
 &\leq  4\sqrt{f(r)}+4\sqrt{f(r+1)} \leq 8\sqrt{f(r)}
\end{align*}
by \eqref{eq: loc bound 3} with  $\delta(r)=f(r)$.

We now apply Lemma \ref{lem: parallel transport simple} to obtain a zero-dimensional TDI $E_r(\cdot)$ of operators on $\caH_{I_{r+1}}$ (see Section~\ref{sec:GNS of product state}) transporting $\mu_r\otimes \Pi_{I_{r+1}\setminus I_r}$ into  $\mu_{r+1}$, satisfying the bound 
\begin{equation}\label{eq: bound on operators l}
\sup_{s}||E_r(s)|| \leq 8 \sqrt{f(r)}
\end{equation}
 and we identify $E_r(\cdot)$ with $E_r(\cdot)\otimes \id_{I_{r+1}^c}$ so that $E_r(\cdot)$ are viewed as operators on $\caH$ transporting $\rho_r$ into $\rho_{r+1}$. Since the restricted GNS representation $\pi|_{\caA_X}:\caA_X\to \caB(\caH_X)$ is bijective for $X$ finite, we can define
$$
K_r(s)=\pi^{-1}(E_r(s))  \in \caA_{I_{r+1}}
$$
It remains to assemble this family of time-dependent operators into a single TDI.
Let $T_r \subset [0,1]$ be half-open intervals $[t,t')$ of length $(\pi^2/6)^{-1} \times 1/r^2$ such that $\min T_1=0$ and $\sup T_{r}=\min T_{r+1}$. Note that $\cup_{r} T_r=[0,1)$. Then, for $s\in [0,1]$,
$$
K_{S}(s)= \begin{cases}  \frac{1}{|T_r|}  K_r(\frac{s}{|T_r|})   &    S=I_{r+1}, s\in T_r \quad \text{for some $r\in \bbN^+$} \\
0 & \text{otherwise}. \end{cases}
$$
By construction, $\phi\circ\alpha_K(\sup T_r)$ is the state corresponding to the density matrix $\rho_{r+1}$ and so $\phi\circ\alpha_K(\sup T_r)\to \psi$ as $r\to\infty$. The bound $ |||K|||_{\{0,1\},\hat{f}}<1$ now follows from \eqref{eq: bound on operators l}.  The other claims simply carry over from Lemma~\ref{lem: parallel transport simple}. 
\end{proof}

 \section{A loop with trivial index and product basepoint is homotopic to a short loop}  \label{sec: trivial index loop is short}
 
By a \emph{short} loop, we mean a loop with a the generator of small norm. 
\begin{proposition}\label{prop: loop contractible to small loop}
Let $\phi$ be a $G$-invariant pure product state on $\caA$ and let $H$ be a $G$-invariant TDI such that $\phi\circ \alpha_H(\cdot)$ is a loop with trivial index, $I(\phi,H)= 0$. Assume that $|||H |||_f \leq C_H$ for some $f\in\caF$. Then there exists $G$-invariant TDIs $\widetilde H^{(R)}$, with $R\in 2\bbN$ such that
$\phi\circ\alpha_{\widetilde H^{(R)}}(\cdot)$ is a loop that is $G$-homotopic to $\phi\circ\alpha_{H}(\cdot)$, and such that
 $$
 \lim_{R\to\infty} \nor \widetilde H^{(R)}\nor_{\hat{f}} =0.
 $$  
 for $\hat{f}$ affiliated to $f$ and depending on $C_H$. The homotopy can be chosen to have fixed basepoint~$\phi$. 
\end{proposition} 
We work throughout the section under the assumptions of the proposition. All constants $C$ that appear can depend on $C_H$ and $f$. 
In particular, we write $\hat{f}$ for a generic function affiliated to $f$ and depending on $C_H$.
Just like for constants $C$, it is understood that the function $\hat{f}$   can change from line to line.   
 
\subsection{Splitting of the loop at a single edge}\label{sec: splitting at single edge}

We say that a TDI $F$ is split at an edge $(j,j+1)$ iff\
$$
 \min S \leq j <\max S  \Rightarrow \forall s: F_S(s)=0.
$$
In such a case, we can write $F=F_{\leq j}+F_{>j}$ where, obviously, $F_{\leq j}, F_{>j}$ consist of terms $F_S$ with $\max S\leq j$ and $\min S>j$, respectively. The aim of this section is to modify the TDI $H$ locally, $H \to H+E^{(j)}$, such that $H+E^{(j)}$ is split at $(j,j+1)$  while still generating a loop with basepoint $\phi$.
We write 
$$
H=H_{\leq j}+H_{>j} + B_j,\qquad \text{where}\quad (B_j)_S= \chi( \min S \leq j <\max S)  H_S
$$ 
such that $H-B_j$ is split at $(j,j+1)$. 
Note that for $j=0$, we previously used the notation $H'$ for $H_{\leq 0}$. 
We recall that the pumped state $\phi\circ\alpha_{H_{\leq 0}}(1)$ is assumed to have zero charge with respect to $\phi$. The following lemma uses and strenghtens this statement.
\begin{lemma}\label{lem: zero charge both sides}
Let 
$$
\psi_-=\phi\circ\alpha_{H_{\leq j}}(1), \qquad \psi_+=\phi\circ\alpha_{H_{> j}}(1).
$$
Then the states $\psi_-|_{\leq j},\psi_+|_{> j} $ are normal and have zero relative charge with respect to $\phi|_{\leq j},\phi|_{> j} $, respectively. 
\end{lemma}
\begin{proof}
Since $\psi_-=\psi_-|_{\leq j} \otimes \phi|_{> j}$ and  $\psi_+=\phi|_{\leq j} \otimes \psi_+|_{> j}$, we have  
\begin{equation}\label{eq: charge invariant under truncation}
h_{\left(\sfrac{\psi_-|_{\leq j}}{\phi|_{\leq j}}\right)} =  h_{\left(\sfrac{\psi_-}{\phi}\right)},\qquad   h_{\left(\sfrac{\psi_+|_{>j}}{\phi|_{> j}}\right)}=  h_{\left(\sfrac{\psi_+}{\phi}\right)}
\end{equation}
from item iii) of Proposition \ref{prop: zerodim}. 
Therefore, in what follows, we abuse notation and we write $\psi_-,\psi_+ $ also for $\psi_-|_{\leq j},\psi_+|_{> j} $, as the potential for confusion is small. 
 Then
$$
\psi_-\otimes\psi_+= \phi\circ \alpha_{H-B_j}(1)= \phi\circ\alpha_{H}(1)\circ   \alpha_L(1)=  \phi  \circ \alpha_L(1)
$$
where, by Lemma \ref{lem: local perturbation tdi}, $L$ is a G-invariant TDI anchored in $\{j,j+1\}$. The last equality follows because $H$ generates a loop. By invoking the last item of Lemma \ref{lem: local perturbation tdi}, we conclude that $\psi_- \otimes\psi_+$ has zero charge wrt.\ $\phi$. Therefore, using \eqref{eq: charge invariant under truncation} and Proposition \ref{prop: zerodim}, we have $h_{\psi_-/\phi}+h_{\psi_+/\phi}=0$.  
If $j=0$, then the state $\psi_-$ is the pumped state $\psi'$ defined in Section \eqref{sec: relative charge for cut loop}. For $j\neq0$, TDIs $H'=H_{\leq 0}$ and $H_{\leq j}$ differ by an interaction anchored in a finite set, therefore the resulting states have equal charge, invoking again Lemma \ref{lem: local perturbation tdi}. Since we assumed that the loop generated by $H$ has index zero, the claims  follow. 
\end{proof}

\begin{lemma}\label{lem: splitting single edge}
For any site $j \in \bbZ$, there is a $G$-invariant TDI $E^{(j)}$ with 
$$
\nor E^{(j)} \nor_{\{j,j+1\},\hat{f}} \leq  C,
$$
such that $ H+E^{(j)}$  is split at the edge $(j,j+1)$ and it generates a loop with basepoint $\phi$.
\end{lemma}
In particular the affiliated function $\hat{f}$ can be chosen independently of $j$.
\begin{proof}
We fix $j$ and we use again the notation $\psi_-,\psi_+$ as above. 
By the same analysis as the one in the proof of Lemma \ref{lem: cut state}, we get
$$
||(\phi-\psi_-)|_{\leq j-r-2} || \leq \hat{f}(r), \qquad ||(\phi-\psi_+)|_{> j+r} ||\leq \hat{f}(r).
$$
We will use Lemma \ref{lem: parallel transport} to connect these states. Even though Lemma \ref{lem: parallel transport} is formulated for states on a spin chain algebra, i.e.\ referring to the graph $\caA=\caA_\bbZ$, it of course holds just as well for states on $\caA_X$ with $X\subset\bbZ$. We apply it once with $X=\{\leq j\}$ and once with $X=\{> j\}$, to obtain TDIs $K_-,K_+$,  supported on $\{\leq j\}$ and  $\{> j\}$, and such that 
$$
\psi_-\otimes\psi_+  = \phi \circ \alpha_K(1), \qquad K=K_-+K_+, \qquad \nor K_{\pm}\nor_{\{j,j+1\},\hat{f}} \leq 1,
$$ 
with $\hat{f}$ chosen independently of $j$: Indeed, in both Lemmas~\ref{lem: cut state} and~\ref{lem: parallel transport}, the function $\hat f$ depends only on $f$ and $C_H$, in particular not on $j$. Hence we have now obtained that the path of states
$$
\phi\circ \alpha_{H-B_j} \circ \alpha_K^{-1}(\cdot)
$$ is a loop with basepoint $\phi$ and it is factorized between $\{\leq j\}$ and $\{>j\}$. 
The TDI generating this loop is, by Lemma \ref{lem: inverse small diff},
$$
H+E^{(j)}, \qquad E^{(j)}(s)= \alpha_K(s)[- K(s) - B_j(s)] + ( \alpha_{K}(s)[H(s)] - H(s)).
$$
Since $H-B_j$ and $K$ are split at $(j,j+1)$, we deduce that $H+E^{(j)}$ is split as well. 
To get the desired bounds on the TDI $E^{(j)}$, we use Lemma~\ref{lem: loc and liebrobinson} and Duhamel's formula.
\end{proof}

\subsection{Approximate splitting into factors}\label{sec: approximate splitting}

In the previous subsection, we constructed anchored interactions $E^{(j)}$ that each split the loop at the edge $(j,j+1)$. 
Now we will try to split a loop at all edges $(Rn,Rn+1)$ at once, with  $n\in \bbZ$ and $R$ a large even integer.  We will however only succeed to do this approximatively, and we will obtain the representation
$$
\alpha_H=   \alpha_{F} \circ  \alpha_{W} \circ \alpha_Z
$$ 
where $F,W,Z$ are $G$-invariant TDIs depending on the parameter $R$ such that
\begin{enumerate}
\item $W$ is small in the sense that $\nor W \nor_{\hat{f}}\leq \hat{f}(R)$.
\item $F$ is factorized over the intervals $I_n$  and $Z$ is factorized over the intervals $J_n$ (both $I_n,J_n$ are defined below and  have length $R$, see Figure~\ref{fig: intervals}).
\item $F$ is close to generating a loop: for each $n$, 
$$||(\phi-\phi\circ\alpha_{F}(1))_{I_n}|| \leq \hat{f}(R), $$
and analogously for $Z$ and the intervals $J_n$.
\end{enumerate}

In the present section \ref{sec: approximate splitting} and the next section \ref{sec: final stage proof}, all TDIs and all states will be $G$-invariant, and we will not repeat this.

\subsubsection{Splitting the loop into factorized quasiloops}\label{sec:factorized quasi-loops}

Let $R\in\bbN$. We introduce the intervals 
$$
I_n= [Rn+1,R(n+1)],\qquad   J_n= [R(n-\tfrac{1}{2})-1,R(n+\tfrac{1}{2})]\qquad(n\in \bbZ).
$$

\begin{figure}[htb]
\begin{center}
\includegraphics[width=0.75\textwidth]{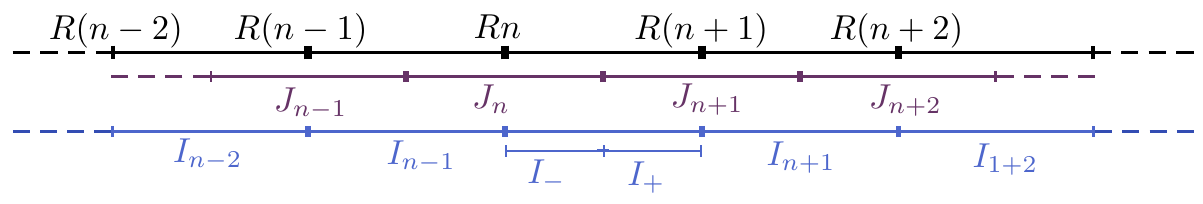}
\caption{The various intervals used in the splitting procedure}
\label{fig: intervals}
\end{center}
\end{figure}
 
We consider  the TDI 
\begin{equation}\label{F tilde}
\widetilde F=H+\sum_{n \in \bbZ} E^{(Rn)}
\end{equation}
where $E^{(Rn)}$ are defined as in Lemma~\ref{lem: splitting single edge}. The infinite sum is understood in a weak sense, as in Section \ref{sec: manipulating interactions}. Since $E^{(Rn)}$ is anchored at $\{Rn,Rn+1\}$, Lemma \ref{lem: loc and liebrobinson}, item ii)  implies that $\nor\widetilde F\nor_{\hat{f}}\leq C$.
Because of the tails of the anchored interactions $E^{(Rn)}$, $\widetilde F$ is not exactly split over the intervals $I_n$.
 Therefore, we explicitly define the split interaction $F$ by
$$
F_S=\begin{cases}  \widetilde F_S & S \subset I_n \quad \text{for some $n$} \\
0  & \text{otherwise}
\end{cases}
$$
Let $\widetilde Z$ be a TDI 
such that  $\alpha_{\widetilde Z}=\alpha_{F}^{-1} \circ \alpha_H$. We similarly define its split version, but unlike $F$, the terms are now supported on the intervals $J_n$:
$$
Z_S=\begin{cases}  \widetilde Z_S & S \subset J_n \quad \text{for some $n$} \\
0  & \text{otherwise}
\end{cases}
$$
We will now express that the differences $F-\widetilde F, Z-\widetilde Z$ are small. We note first that all four of these interactions depend on the parameter $R$ and this is the parameter that controls the smallness, namely
\begin{lemma}\label{lem: small difference f}
$$
\nor F-\widetilde F\nor_{\hat{f}} \leq \hat{f}(R)
$$
\end{lemma}
\noindent To avoid confusion, we remind that the two functions called $\hat{f}$ in the above expression need not be equal,  as explained at the beginning of Section \ref{sec: trivial index loop is short}.
\begin{proof} 
The difference $F-\widetilde F$ contains only those terms that cross at least one of the cuts between two intervals $I_n$. Since $H+E^{(Rn)}$ is split at site $Rn$, the only interaction terms in $\widetilde F=H+\sum_{n'} E^{(Rn')}$ that contain $\{Rn,Rn+1\}$ arise from $ E^{(Rn')}$ with $n' \neq n$ and hence they correspond to a set $S$ with $\diam(S)\geq R$ since $ E^{(Rn')}$ is anchored at $\{Rn',Rn'+1\}$. It follows that every term in $F-\widetilde F$ is of the form $\widetilde F_S$  with  $\diam(S)\geq R$.  Therefore
$$\nor F-\widetilde F\nor_{\sqrt{\hat{f}}} 
\leq \sqrt{\hat{f}(R)}  \nor \widetilde F \nor_{\hat{f}}. $$
As we already remarked that $\nor \widetilde F \nor_{\hat{f}}\leq C$, the lemma is proven.
\end{proof}

\begin{lemma}\label{lem: small differences k}
$$
  \nor Z-\widetilde Z\nor_{\hat{f}} \leq \hat{f}(R).
$$
\end{lemma}
\begin{proof}
Since $\widetilde Z$ is such that $\alpha_{\widetilde Z}=\alpha_{F}^{-1} \circ \alpha_H$, Lemma \ref{lem: inverse small diff} yields
\begin{equation}\label{eq: difference widetilde k}
\widetilde Z(s)=H(s)-F(s)- i \alpha^{-1}_{H}(s)\left[\int_0^s \alpha_{H}(u)\Big\{\big[H(u) - F(u),\alpha_F(s,u)\{F(s)\}\big]\Big\}du \right]
\end{equation}
We write
\begin{equation}\label{Decomp of H-F}
    H-F = D_1+D_2,\qquad  D_2= H - \widetilde F = \sum_n E^{(Rn)}, \qquad  D_1= \widetilde F - F,
\end{equation}
and we note the following properties
$$
\nor D_1\nor_{\hat{f}}\leq \hat{f}(R),\qquad   \nor D_2\nor_{R\bbZ,\hat{f}}\leq C.
$$ 
The bound on $D_1$ follows from Lemma \ref{lem: small difference f}. The bound on $D_2$ is by Lemma \ref{lem: loc and liebrobinson} item ii) and Lemma \ref{lem: splitting single edge}.  
We now plug the decomposition~(\ref{Decomp of H-F}) of $H-F$ into \eqref{eq: difference widetilde k}. Lemma \ref{lem: loc and liebrobinson} items i,iii) then yields that $\widetilde Z$ has a representation similar to $H-F$, namely
\begin{equation}\label{Ztilde estimates}
    \widetilde Z = D_1'+D'_2,\qquad   \nor D_1'\nor_{\hat{f}} \leq   \hat{f}(R),\qquad  \nor D_2'\nor_{R\bbZ,\hat{f}} \leq   C.
\end{equation}
We recall now that $Z$ is the split version of $\widetilde Z$, over the intervals $J_n$. Since $D_2'$ is anchored at the intersection of intervals $I_n$, $(D_2')_S \neq 0$ implies then $\diam(S) \geq R/2$ if $S$ intersects more than a single interval $J_n$. We can therefore conclude as in the proof of the previous lemma that
$$
  \nor Z-\widetilde Z \nor_{\sqrt{\hat{f}}} \leq  \nor D_1' \nor_{\sqrt{\hat{f}}}+ \sqrt{\hat{f}(R)} \nor D_2'\nor_{R\bbZ,\hat{f}}
$$  which, together with~(\ref{Ztilde estimates}), yields the claim.
\end{proof}

\subsubsection{$F$ and $Z$ generate quasiloops}\label{sec: quasiloops}

We have established that $F$ and $Z$  are manifestly split;  $F$ is decoupled over intervals $I_n$ and $Z$ is decoupled over intervals $J_n$. 
So in particular we have that both $\phi \circ \alpha_{F}$ and $\phi \circ \alpha_{Z}$ are factorized over intervals of length $R$.  
We now investigate how close those states are to $\phi$, namely how close $F, Z$ are to generating genuine loops.  
\begin{lemma}\label{lem: almost loop f}
Let $\psi=\phi \circ \alpha_{F}(1)$. Then
$$||(\psi  - \phi)|_{{I_n}}|| \leq \hat{f}(R).$$
\end{lemma}
\begin{proof}
First of all, Lemma~\ref{lem: small difference f} together with Lemma~\ref{lem: evolution by similar tdi} imply that $\Vert \alpha_{F}(1)[A]- \alpha_{\widetilde F}(1)[A]\Vert \leq  R \hat f(R)\Vert A\Vert$ for any $A\in\caA_{I_n}$. We absorb $R$ into $\hat f$ and we conclude that $|| (\psi-\phi\circ\alpha_{\widetilde F})_{I_n} ||\leq \hat f(R)$.  It remains therefore to compare $\phi$ with $\phi\circ\alpha_{\widetilde F}$. For this, we  partition, see Figure~\ref{fig: intervals},  
$$I_n= I_- \cup I_+, \qquad  I_-=I_n \cap J_{n}, \qquad I_+=I_n \cap J_{n+1}.$$
 Let us first estimate $||(\psi-\phi)|_{I_{+}}||$.
Let $F_n=H+E^{(Rn)}$. Then $F_n$ generates a loop with basepoint $\phi$ by Lemma~\ref{lem: splitting single edge} and we have, for $A\in \caA_{I_+}$,
\begin{align}
\phi(A)-\phi\circ\alpha_{\widetilde F}(A)& = \phi(\alpha_{F_n}(1)[A]- \alpha_{\widetilde F}(1)[A] ).
\end{align}
We note now that the TDI $\widetilde F-F_n=\sum_{n'\neq n } E^{(Rn')}$ is anchored at distance at least $R/2$ from $I_+$, see Figure~\ref{fig: intervals} and Lemma \ref{lem: loc and liebrobinson}.  Therefore, by Lemma \ref{lem: evolution by similar tdi},  the integrand is bounded by  $\hat{f}(R) ||A ||$.
We next estimate $||(\psi-\phi)|_{I_{-}}||$ by the same argument but this time using $F_{n+1}=H+E^{(R(n+1))}$, so that, again $\widetilde F-F_{n+1}$ is anchored at distance $R/2$ from the relevant region $I_-$.  We have thus obtained 
$$
||(\psi-\phi)_{I_{\pm}}||  \leq \hat{f}(R).
$$
These estimates on each half $I_\pm$ of $I_n$ suffice to obtain the claim by applying Lemma~\ref{lem: bipartite}, where the projectors correspond to the (pure) restriction of the product state $\phi$ to $I_\pm$, and the density matrices are the reduced density matrices of $\psi$ to $I_\pm$.
\end{proof}
\begin{lemma}\label{lem: almost loop k}
Let $\psi=\phi \circ \alpha_{Z}^{-1}(1)$. Then
$$||(\psi  -\phi)|_{J_n}|| \leq \hat{f}(R).$$
\end{lemma}
\begin{proof}
Recalling that $\alpha_{\widetilde Z}=\alpha_{F}^{-1} \circ \alpha_H$, we have
$$\psi= \phi \circ \alpha_{H}^{-1}(1) \circ \alpha_{F}(1) \circ \alpha_{\widetilde Z}(1) \circ 
\alpha_{Z}^{-1}(1)= \phi \circ \alpha_{F}(1) \circ \alpha_{\widetilde Z}(1) \circ 
\alpha_{Z}^{-1}(1),
$$
where the last equality follows because $H$ generates a loop with basepoint $\phi$. 
Now we estimate
\begin{equation}\label{eq: almost loop k}
|\psi(A)-\phi(A)| \leq    |\phi \circ \alpha_{F}(1)[A]-\phi[A]| +   ||  A- \alpha_{\widetilde Z}(1) \circ 
\alpha_{Z}^{-1}(1)[A] ||
\end{equation}
for $A\in\caA_{J_n}$, and we deal with both terms separately. For the first term, we use that $J_n$ is contained in the union $I_n \cup I_{n+1}$. Employing Lemma  \ref{lem: almost loop f} for the intervals $I_n,I_{n+1}$ and Lemma \ref{lem: bipartite}, we can bound this term by $||A || \hat{f}(R)$.
  For the second term in \eqref{eq: almost loop k}, we note that $\alpha_{\widetilde Z} \circ 
\alpha_{Z}^{-1}$ is generated by $W= \alpha_Z[\widetilde Z - Z]$, see Lemma \ref{lem: manipulation of evolutions}. The smallness of $\widetilde Z - Z$, see Lemma 
\ref{lem: small differences k} leads via Lemma \ref{lem: loc and liebrobinson} to  $||W||_{\hat{f}}\leq \hat{f}(R)$ and hence  to the estimate
$$
 ||  A- \alpha_{W}(1)[A] || \leq \int_0^1 ds  ||  [W(s), A]||  \leq    \hat{f}(R)  |J_n| || A ||
$$
Since $|J_n|$ grows linearly in $R$, it can be absorbed in $\hat{f}(R)$. Therefore, \eqref{eq: almost loop k} is bounded by $\hat{f}(R) ||A||$, as required. 
\end{proof}

Let us check that we achieved the result described at the beginning of Section \ref{sec: approximate splitting}. We write
\begin{equation}\label{First Splitting Loop}
    \alpha_H=   \alpha_{F} \circ  \alpha_{F}^{-1} \circ \alpha_H = \alpha_{F} \circ  \left(\alpha_{\widetilde Z} \circ \alpha^{-1}_Z\right) \circ \alpha_Z =  \alpha_{F} \circ  \alpha_{W} \circ \alpha_Z
\end{equation}
and the desired properties of  $F,Z,W$ are manifest from the results above.

\subsection{Final stage of the proof of Proposition \ref{prop: loop contractible to small loop}} \label{sec: final stage proof}

\subsubsection{Closing the product quasi-loops}

Up to now, we have defined the TDIs $F,Z$ that are split on intervals of length $R$. They do not generate loops, but they fail to do so by an error which decays fast in~$R$. We now show that, by modifying $F,Z$ by a small term that is split as well, the quasi-loops are made into genuine loops.
\begin{lemma}\label{lem: closure quasi loops}
There are TDIs $E=E^{(R)},L=L^{(R)}$ such that 
\begin{enumerate}
\item  $\nor E \nor_{\hat{f}}\to 0 $ and $\nor L \nor_{\hat{f}}\to 0 $, as $R\to\infty$.
\item  $E$ is split over the intervals $I_n$ and $L$ is split over the intervals $J_n$.
\item  $\phi \circ  \alpha_{F} \circ \alpha_{E}$  and $\phi  \circ \alpha_{L} \circ \alpha_{Z}$ are loops with basepoint $\phi$.
\end{enumerate}
\end{lemma}
\begin{proof}
Both $\phi \circ \alpha_{F}(1)$ and $\phi$ are factorized exactly over intervals $I_n$. We consider each interval separately and we note that we are exactly in the framework of Lemma \ref{lem: parallel transport simple}. For each $n$, we obtain a time-dependent observable $E_{n}(\cdot)$ in $ \caA_{I_n}$
that generates parallel transport from $(\phi \circ \alpha_{F}(1))_{I_n}$ to $\phi_{I_n}$, with norm bounded by 
\begin{equation}\label{En s}
\nor E_n \nor \leq 8 || (\phi \circ \alpha_{F}(1)-\phi)_{I_n}    || \leq  \hat{f}(R),
\end{equation}
see Lemma~\ref{lem: almost loop f}. We then assemble the different $E_n(\cdot)$ into a TDI: Let $
E_S(s)= E_n(s) $ for $S=I_n$ for some $n$ and $E_S(s)=0$ otherwise.
This TDI $E$ is manifestly split and it satisfies
$$
\phi \circ \alpha_{F}(1) \circ \alpha_E(1)=\phi.
$$
Since the non-zero contributions to $E$ are given explicitly by $E_n$ that are strictly supported in the finite intervals $I_n$, the norm of $E$ is immediately bounded, for any $h \in \caF$, by
\begin{equation}\label{E epsilon}
\nor E \nor_h\leq \sup_n \frac{\sup_{s}||E_n(s)||}{h(R)}.
\end{equation}
We choose $h=\sqrt{\hat{f}}$, and so (\ref{En s}) and (\ref{E epsilon}) yield $\nor E \nor_{\sqrt{\hat{f}}}\leq \sqrt{\hat{f}(R)}$. This proves all claims concerning $E$.
We proceed similarly with the generator of $\alpha_Z^{-1}$ given by Lemma~\ref{lem: manipulation of evolutions}, but replacing now $I_n$ by $J_n$ and calling the resulting (split and small) interaction $Y$. This achieves that $\phi = \phi  \circ \alpha_{Z}^{-1}(1) \circ \alpha_{Y}(1) 
$ and hence also 
$$
\phi = \phi  \circ \alpha_{Y}^{-1}(1) \circ \alpha_{Z}(1) 
= \phi  \circ \alpha_{L}(1) \circ \alpha_{Z}(1).
$$
where $L$ generates the inverse of $\alpha_Y$.
\end{proof}

\subsubsection{Contracting product loops}

In the above sections, we have produced $G$-loops factorized over a collection of intervals $Y_n$. In our case $Y_n=I_n$ or $Y_n=J_n$.
We will now show that such a loop is contractible, i.e.\ homotopic to a constant loop.  
\begin{lemma}\label{lem: contractibility of product loop}
Consider a $G$-loop $\psi$ that is factorized over intervals $Y_n$, with $Y_n$ forming a partition of $\bbZ$. If $\diam(Y_n)$ is bounded uniformly in $n$, then $\psi$ is $G$-homotopic to a constant loop with the same basepoint. 
\end{lemma}
It is intuitively clear that this lemma relies on the fact that zero-dimensional systems don't allow for nontrivial loops because $\caB(\caH)$ is simply connected.  However, our notion of loops and their homotopy for states on spin chain algebras, does not simply reduce, in the case the chain is finite, to the standard notion of continuous loops and homotopy on topological spaces.  The reason is that we require the loops to be generated by a TDI, which, even for loops on $\caB(\caH)$, is a stronger notion than continuity. Actually, the requirement of being generated by a zero-dimensional TDI  $E(\cdot)$, as defined before Lemma \ref{lem: parallel transport}, is in that case equivalent to absolute continuity of the loop. 
For this reason, we actually need an additional lemma on zero-dimensional loops.
\begin{lemma}\label{lem: contractibility zero dim} Let $s\mapsto\nu(s)$ be a loop of pure normal states on $\caB(\caH)$, with $\caH$ a finite-dimensional Hilbert space, 
and  $E(\cdot)$ a zero-dimensional TDI with $ |||E|||= \sup_s||E(s)|| <\infty$ such that 
$$ \nu(s)=  \nu(0) \circ  \alpha_E(s).$$
Then, there is a two-dimensional family of pure normal states $\nu_{\lambda}(s)$ with $\lambda, s \in [0,1]$, and zero-dimensional TDIs $E_\lambda(\cdot),F_s(\cdot)$ such that 
\begin{enumerate}
\item $\sup_\lambda|||E_\lambda||| \leq  80|||E||| $ and  $\sup_s |||F_s||| \leq  208 $
\item $\nu_1(s)=\nu_{\lambda}(0)=\nu_{\lambda}(1)=\nu(0)=\nu(1)$ for all $s,\lambda$. 
\item  $\nu_{\lambda}(s) =  \nu(0) \circ  \alpha_{E_\lambda}(s)  $
\item  $\nu_{\lambda}(s)  = \nu(s) \circ  \alpha_{F_s}(\lambda)   $
\item If the loop $\nu(\cdot)$ is $G$-invariant, then the family $\nu_\lambda(s)$ and the families $E_\lambda(s),F_s(\lambda)$ can also be chosen $G$-invariant. 
\end{enumerate}
\end{lemma}
We refer to the appendix for proof of this lemma, relying on standard Hilbert space theory.

\begin{proof}[Proof of Lemma \ref{lem: contractibility of product loop}]
Let $E$ be the TDI generating the factorized loop $\psi$, with basepoint $\phi$. 
Then $(\psi(s))_{Y_n}$ is a loop of pure, G-invariant states on $\caA_{Y_n}$.
Since $\caA_{Y_n}$ is isomorphic to $\caB(\caH_{Y_n})$ through the map $\pi^{-1}$, we can apply Lemma \ref{lem: contractibility zero dim} for every $n$ separately. We first obtain then the time-dependent operators $E_{\lambda,n}(\cdot)$ and $F_{s,n}(\cdot)$ in $\caB(\caH_{Y_n})$.
Proceeding as in the proof of Lemma~\ref{lem: closure quasi loops}, we assemble these time-dependent operators in two TDIs $E_\lambda,F_s$.
Both have finite $f$-norm for any $f\in\caF$ since $\sup_n|Y_n|<\infty$ and $\sup_n||E_{n,\lambda}||<\infty$.
\end{proof}

\subsubsection{Proof of Proposition~\ref{prop: loop contractible to small loop}}

We start from the identity (\ref{First Splitting Loop}) and use Lemma~\ref{lem: closure quasi loops} by bringing in the almost local evolutions $E,L$ defined therein and  bracketing terms judiciously to get:
$$
\alpha_H=   \left(\alpha_{F} \circ \alpha_E\right)   \circ  \left(\alpha^{-1}_E \circ  \alpha_{\widetilde Z} \circ \alpha^{-1}_Z \circ \alpha^{-1}_L\right) \circ \left(\alpha_L \circ \alpha_Z\right).
$$ 
The first and third bracketed evolutions generate loops with basepoint $\phi$. Since $\alpha_H$ also generates a loop with basepoint $\phi$, we conclude that also the second bracketed term generates a loop with basepoint $\phi$. By Lemma \ref{lem: composition of loops is homotopic to concatenation},  the above almost local evolution is $G$-homotopic to the concatenation of 3 loops\footnote{Concatenation of loops is not an associative operation. However, Lemma \ref{lem: composition of loops is homotopic to concatenation} shows that $(\alpha_{H_1}\loopc\alpha_{H_2})\loopc \alpha_{H_3}$ is homotopic to $\alpha_{H_1}\loopc(\alpha_{H_2}\loopc \alpha_{H_3})$, since the composition $\circ$ is associative. Therefore we allow ourselves here to write simply $\ldots\loopc\ldots\loopc\ldots$.}.
$$
\left( \alpha_{F} \circ \alpha_E \right) \loopc \left(\alpha^{-1}_E \circ  \left(\alpha_{\widetilde Z} \circ \alpha^{-1}_Z\right) \circ \alpha^{-1}_L\right) \loopc\left( \alpha_L \circ \alpha_Z \right)
$$
 The first and third loop are product loops, hence contractible by Lemma \ref{lem: contractibility of product loop}. It follows that the loop $\phi\circ \alpha_H$ is $G$-homotopic to the loop $\phi\circ \alpha_{\widetilde H}$, where
$$
\alpha_{\widetilde H}=\alpha^{-1}_E \circ  \left(\alpha_{\widetilde Z} \circ \alpha^{-1}_Z\right) \circ \alpha^{-1}_L.
$$
Just as all of $E,\widetilde Z,Z,L$ constructed above, the  TDI $\widetilde H$ depends on the parameter $R$, and we denote it $\widetilde H^{(R)}$
Since $Z-\widetilde Z$ is small (Lemma \ref{lem: small differences k}) and $E,L$ are small in the same sense, we get from Lemma \ref{lem: inverse small diff} and Lemma~\ref{lem: loc and liebrobinson}, that  $\widetilde H^{(R)}$ is also small, namely
$||\widetilde H^{(R)}||_{\hat{f}}\leq \hat{f}(R)$. 
\hfill \ensuremath{\Box}

\section{Contractibility of short loops with product basepoint}\label{sec: contractibility of short loops with product basepoint}

In this section, we prove that short loops, i.e.\ where the TDI has small norm, with product basepoint, are $G$-homotopic to a constant loop. The more precise statement is below in Proposition~\ref{prop: small loop is contractible}.  As before, we fix $\phi$ to be a product $G$-state.
\begin{proposition}\label{prop: small loop is contractible}
For any $f\in\caF$, there is an $\epsilon_1(f)$, such that, if $|||H|||_f\leq \epsilon_1(f)$ and the $G$-invariant TDI $H$ generates a loop $\phi\circ\alpha_H$, then this loop is $G$-homotopic to the trivial loop with basepoint~$\phi$, and the homotopy can be chosen to have constant basepoint. 
\end{proposition}
\subsection{Construction of a loop of ground states}\label{sec: construction of loop of gs}
First we rescale the given $G$-loop $\phi\circ\alpha_H$ to
\begin{equation}\label{Rescaled loop}
\psi(s)=\begin{cases}  \phi\circ\alpha_H(2s)  &  s\leq 1/2, \\  
\phi &   s>1/2,
 \end{cases}
\end{equation}
so that now $\psi(1/2)=\phi$. The resulting loop $\psi$ is $G$-homotopic to the original one by Lemma \ref{lem: composition of loops is homotopic to concatenation} item i).  Next, we construct a specific $G$-invariant interaction $F$, tailored to the product state $\phi$:
$F$ has only on-site terms $F_{\{i\}}= \id-P_i$, which is the orthogonal projection on the kernel of the state $\phi|_i$. It follows that for any finite $S$, $\sum_{i \in S} F_{\{i\}}$ has a simple eigenvalue $0$, corresponding to the state $\phi|_S$, and it has no other spectrum below $1$.  
We now define the $G$-invariant interaction
$$
Z(s)= \begin{cases}  \alpha^{-1}_H(2s)[F] & s \leq 1/2,  \\[2mm]
 (2s-1)F+    (2-2s)Z(1/2)   & s >1/2.
 \end{cases}
$$
We have obtained a loop of interactions, since $Z(1)=Z(0)=F$. We contract it to a point (namely $F$) by setting
$$
Z_{\lambda}(s)  =Z(0)+ (1-\lambda)( Z(s)-Z(0)), \qquad  \lambda \in [0,1].
$$ 
This two-parameter family of interactions has the property that it remains close to $F$, as we remark now. 
\begin{lemma}\label{lem: loop is small}
If  $|||H|||_f\leq  \epsilon \leq  1$, then, for $\hat{f}$ affiliated to $f$,
\begin{enumerate}
\item  $
||Z_\lambda(s)-F||_{\hat{f}} \leq   \epsilon  
$
\item  There are $G$-invariant TDIs $X_\lambda(\cdot), X'_s(\cdot)$ such that 
$$
Z_{\lambda}(s)-Z_{\lambda}(0)=\int_{0}^{s} du X_\lambda(u), \qquad    Z_{\lambda}(s)-Z_{0}(s)=\int_{0}^{\lambda} d\lambda' X'_s(\lambda') 
$$
\item  The TDIs  from item $ii)$ are uniformly bounded 
$$ \sup_{s,\lambda} (||| X_\lambda |||_{\hat{f}}+ ||| X'_s |||_{\hat{f}}) \leq \epsilon
$$
\end{enumerate}
\end{lemma}
\begin{proof}
Item i) follows directly from the definition of $Z$, using Lemma \ref{lem: loc and liebrobinson}, and $||F||_f =f(1)$. 
For items ii) and iii), we pick
$$
X_\lambda(s)= \begin{cases} -2\iu \left[H(2s),\alpha_H^{-1}(2s)[F]\right] &  s\leq 1/2  \\[2mm]
2F-2Z(1/2)  & s>1/2  \end{cases},\qquad  X'_s(\lambda)= - (Z(s)-Z(0))
$$ 
and we obtain the bounds similarly to item i). 
\end{proof}

To continue, we recall the notion of a ground state in the infinite-volume setting of quantum lattice systems. We use implicitly the discussion of Section~\ref{subsec:alal}.
\begin{definition}\label{def: ground state infinite volume}
A state $\psi$ on a spin chain algebra $\caA$ is a ground state associated to an interaction $K$ iff\
$$
\psi(A^*[K,A]) \geq 0,\qquad \forall A \in \caal
$$
\end{definition}
We will need two particular properties of a ground state. The first one is its invariance under the dynamics generated by the constant TDI $K(s)=K$:
\begin{equation}\label{eq: invariance of ground states}
\psi=\psi\circ\alpha_K(s), \qquad s\in [0,1]
\end{equation}
It follows by using that $[K,\cdot]$ is a derivation. Indeed, $[K,A^*A]=  A^*[K,A] - (A^*[K,A])^*$
and hence $\psi([K,A^*A])=0$ for a ground state $\psi$, which implies \eqref{eq: invariance of ground states}. The other property is the following variational statement, see Theorem~6.2.52 in~\cite{BratRob2}.
\begin{lemma}\label{lem: variational principle}
Let $X$ be a finite subset of $\bbZ$ and let $K$ be an interaction. Consider the functional  $e_{X,K}:\caP(\caA)\to \bbC$ (recall that $\caP(\caA)$ is the set of states) given by
$$e_{X,K}(\psi')= \sum_{S: S \cap X \neq \emptyset} \psi'(K_S)$$ 
Then, any ground state $\psi$ satisfies
$$
e_{X,K}(\psi)=\min_{\psi':\psi'|_{X^c}=\psi|_{X^c} }  e_{X,K}(\psi')
$$
\end{lemma}
This property makes explicit the fact that ground states minimize the energy locally. 
The relevance to our problem is that we have a loop of ground states of explicitly given  interactions, namely
\begin{lemma}
For all $s\in[0,1]$,  $\psi(s)$ defined in~(\ref{Rescaled loop}) is a ground state of the interaction $Z(s)$. 
\end{lemma}
\begin{proof}
Definition~\ref{def: ground state infinite volume} makes it explicit that the property of being a ground state depends only on the derivation association to the interaction $Z(s)$.  For  $s \in [0,1/2] $, we have $$\psi(s)(A\str[Z(s),A]) = \phi(B(s)\str[F,B(s)]),\qquad B(s) = \alpha_H(2s)(A),$$ and therefore the fact that $\psi(s)$ is a ground state of $Z(s)$ follows directly from $\phi$ being a ground state of $F$.  For $s>1/2$, we have $Z(s) \eder (2s-1)F+    (2-2s)\alpha^{-1}_H(1)[F]$ and we use that $\phi\circ\alpha_H(1)=\phi$. Therefore $\phi$ is a ground state of both $F$ and     $\alpha_H^{-1}(1)[F]$ and hence also of their convex combinations. 
\end{proof}

\subsection{Construction of the homotopy via spectral flow}\label{sec: homotopy from spectral flow}
Our main idea to prove Proposition \ref{prop: small loop is contractible} is to consider the 2-parameter family of states given by  groundstates of the two-parameter family $Z_\lambda(s)$. We first remark, in Proposition \ref{prop: uniqueness of ground state}, that these groundstates are unique in the perturbative regime (small perturbations of product states). Results of that kind have been established in great generality (see~\cite{nachtergaele2020quasi} for results that are very close to ours, and also~\cite{yarotsky2006ground,bravyi2010topological,michalakis2013stability, roeck2017exponentially,del2021lie}) but the specific claim we need does not seem to appear in the literature. Therefore, we provide a proof in Appendix \ref{sec: app stability}
\begin{proposition}\label{prop: uniqueness of ground state}
For any $h\in\caF$, there is a $\epsilon_2(h)>0$ such that, if $||W||_h \leq \epsilon_2(h)$, then 
the interaction $F+W$ has a unique ground state. 
\end{proposition} 
 If we choose $\epsilon_1(f)$ in Proposition \ref{prop: small loop is contractible} small enough, then, by Lemma \ref{lem: loop is small}, all of the 
interactions $Z_{\lambda}(s)$ satisfy the condition of the above proposition and hence they have a unique ground state, that we call $\psi_\lambda(s)$. This two-dimensional family  $(s,\lambda)\mapsto \psi_\lambda(s)$ is the  homotopy that will realise Proposition \ref{prop: small loop is contractible}. Indeed, for any $\lambda$, we have $Z_{\lambda}(0)=Z_{\lambda}(1)=F$, which has $\phi$ as the unique ground state, and hence $\psi_\lambda(0)=\psi_\lambda(1)=\phi$. 
 For $\lambda=0$, the family $\psi_\lambda(s)$ reduces to the loop $\psi(\cdot)$, as we had already remarked that this was a loop of ground states of $Z(\cdot)$.
For $\lambda=1$, we have that $Z_1(s)=F$, and hence $\psi_1(\cdot)$ is the constant loop $\phi$. It remains to show that the family $(s,\lambda)\mapsto \psi_\lambda(s)$ is generated by a TDI in $\lambda$ and $s$ directions.\\ Such TDI will be obtained from the so-called \emph{spectral flow} \cite{hastings2005quasiadiabatic,bachmann2012automorphic,moon2020automorphic}, combined with  results on stability of the spectral gap. In particular, Propostion \ref{prop: bhm} below, follows from \cite{nachtergaele2020quasi}, except that the condition on spatial decay of potentials is relaxed slightly. The possibility of doing this is clear when inspecting the proofs of \cite{nachtergaele2020quasi}.
To avoid a clash of notation, we use $z \in[0,1]$ (instead of $s,\lambda$) as time-parameter in a TDI.
\begin{proposition}\label{prop: bhm}
For any $h\in\caF$, there is $0<\epsilon_3(h)< 1$ such that, for any pair of TDIs $W$ and $Y$ satisfying
\begin{equation}\label{eq: w and e}
 |||W|||_h\leq \epsilon_3(h), \quad   |||Y|||_h <\infty, \qquad   W(z)-W(0)=\int^{z}_{0} d z' Y(z'),
\end{equation}
the following holds: There exists a TDI $\hat{Y}$ satisfying $|||\hat Y|||_{h'} \leq ||| Y|||_{h} $ for  $h' \in \caF$ depending only on $h$, such that, if the pure state $\nu(0)$ is a ground state of $F+W(0)$, then 
\begin{equation}\label{eq: states determined}
\nu(z)=\nu(0)\circ\alpha_{\hat Y}(z),\qquad z \in [0,1] 
\end{equation}
is a ground state of $F+W(z)$. If $Y$ is $G$-invariant, then also $\hat{Y}$ can be chosen $G$-invariant.
\end{proposition}
Based on this result we can finish in a straighforward way the
\begin{proof}[Proof of Proposition \ref{prop: small loop is contractible}]
By choosing $\epsilon$ in Lemma \ref{lem: loop is small} small enough, we can apply the spectral flow technique, i.e.\ Proposition \ref{prop: bhm}, to get generating TDIs for the paths $s\mapsto \psi_\lambda(s)$, and  $\lambda
\mapsto \psi_\lambda(s)$. They are uniformly bounded in $|||\cdot|||_{h'}$ by Proposition \ref{prop: bhm} and hence they provide the desired homotopy. 
\end{proof}

By combining Proposition~\ref{prop: small loop is contractible} with Proposition \ref{prop: loop contractible to small loop} we conclude 
\begin{theorem}\label{thm: contractibility products}
Let $\phi$ be a pure $G$-invariant product state and let $\psi(\cdot)=\phi\circ\alpha_{H}(\cdot)$ be a $G$-loop such that $I(\phi,H)=0$, i.e.\ its index is zero, then the loop is homotopic to the constant loop $\Id_\phi$, via a homotopy with constant basepoint. 
\end{theorem}

\section{Classification of loops with product basepoint} \label{sec: classification for product loops}

We are now ready to prove the full result for the case of loops with a fixed product basepoint. Throughout this section, $\caA$ is a fixed chain algebra equipped a fixed group action and $\phi$ is again a fixed pure, $G$-invariant product state.

\subsection{Well-definedness of the index}
Up to now, the index $I(\phi,H)$ was defined as a function of the TDI $H$ and the basepoint $\phi$, see Section~\ref{sec: states cut loop}. We prove now that for product basepoint, it actually depends only on the loop itself and not on the choice of TDI that generates it. 
\begin{lemma}\label{lem: index indep of h}
Let $H_1,H_2$ be $G$-invariant TDIs that generate loops with basepoint $\phi$. If the loops are equal, namely $\phi\circ\alpha_{H_1}(s)=\phi\circ\alpha_{H_2}(s)$ for all $s\in [0,1]$, then the relative charge of the pumped states  $\phi\circ\alpha_{H'_1}(s)$, $\phi\circ\alpha_{H'_2}(s)$ is zero and hence 
$$
I(\phi,H_1)=I(\phi,H_2).
$$ 
\end{lemma}
\begin{proof}
Let $\xi(s)=\phi\circ\alpha_{H'_1}(s) \circ \alpha^{-1}_{H'_2}(s) $. Then
$$
\xi(s)-\phi= \phi \circ \left(\alpha_{H'_1}(s) \circ \alpha^{-1}_{H'_2}(s)-  \alpha_{H_1}(s) \circ \alpha^{-1}_{H_2}(s) \right)
$$
By similar reasoning as in the proof of Lemma \ref{lem: cut state}, we find then some $f\in\caF$ such that
\begin{equation}\label{eq: uniform approx of xi}
   \sup_{s\in [0,1]} || (\xi(s)-\phi)_{I_r^c} || \leq f(r), 
\end{equation}
In particular, by Lemma \ref{lem: condition for normality}, $\xi(s) $ are pure states that are normal with respect to $\phi$; they can be represented as density matrices $\rho(s)$ on $\caB(\caH)$. Using \eqref{eq: uniform approx of xi}, Lemma \ref{lem: condition for normality} and its proof imply then that $\sup_{s\in [0,1]}||\rho(s)-\rho_r(s)\otimes \Pi_{I_r^c}|| \to 0$, as $r\to\infty$, with $\rho_r(s) \otimes \Pi_{I_r^c}$ as in the proof of Lemma \ref{lem: condition for normality}. 
Moreover,  $s\mapsto \xi(s)$
is weakly-*continuous (by the strong continuity of $s\mapsto \alpha_{H_{i}}(s)$ with $i=1,2$) and therefore the function $s\mapsto  \rho_r(s) \otimes \Pi_{I_r^c}$ is actually norm-continuous, since its range is contained in a finite dimensional vector space. 
Hence, $s\mapsto \rho(s)$ is the uniform limit of a sequence of norm-continuous functions $s\mapsto \rho_r(s) \otimes \Pi_{I_r^c}$, and thereby itself norm-continuous. Therefore, the states $\rho(s)$ are $G$-equivalent to $\rho_{\phi}$ (i.e.\ to $\Pi$), and so item i) of Proposition \ref{prop: zerodim} implies that the states $\rho(s)$ have zero relative charge with respect to $\rho_\phi$. Item iv) of the proposition and the definition of the index (Section~\ref{sec: states cut loop}) for mutually normal states then imply that $\phi\circ\alpha_{H'_1}(s)$ and $\phi\circ\alpha_{H'_2}(s)$ have zero relative charge.  
The claim of the lemma then follows by item ii) of Proposition~\ref{prop: zerodim}.
\end{proof}

 
Since the basepoint is fixed in this section, this result means that we can now consistenly use the notation $h(\psi)$ for a loop $\psi$, instead of the more tedious object $I(\phi,H)$.

\subsection{Additivity of the index}

\begin{lemma} \label{lem: additivity}
Let the TDIs $H_1$ and $H_2$ generate $G$-loops $\psi_1,\psi_2 $ with common product basepoint $\phi$.  Then
$$
h(\psi_2\loopc \psi_1)=h(\psi_2)+h(\psi_1)
$$
\end{lemma}

\begin{proof}
The pumped state of the loop $\psi_2\loopc \psi_1$ is 
\begin{equation}  \label{eq: double pumping}
\phi \circ \alpha_{H'_1}(1) \circ \alpha_{H'_2}(1)
\end{equation}
where, as before, $H'_{j}$ are the truncated TDI. We let  $H'_2(r)$ denoted the TDI truncated between sites $-r$ and $-r+1$ instead of the sites $0$ and $1$. By Lemma \ref{lem: local perturbation tdi}, we can replace $H'_2$ by $H'_2(r)$ , such that the resulting pumped state
$$
\psi'(r)=\phi \circ \alpha_{H'_1}(1) \circ \alpha_{H'_2(r)}(1)
$$
has the same charge as  the state \eqref{eq: double pumping}, for any $r$. 
Let $C_H$ be a constant such that $|||H_1|||_f,|||H_2|||_f\leq C_H$ for some $f\in\caF$.
Let now $K_1,K_2(r)$ by the parallel transport TDIs constructed in Lemma \ref{lem: parallel transport} satisfying
\begin{equation}\label{eq: parallel for truncated pair}
\phi \circ \alpha_{H'_1}(1)=\phi \circ \alpha_{K_1}(1) ,\qquad\phi \circ \alpha_{H'_2}(1)=\phi \circ \alpha_{K_2}(1)
\end{equation}
and the bounds 
$$
||K_1||_{\{0,1\}, \hat{f}} \leq  1,\qquad  \sup_r ||K_2||_{\{-r,-r+1\}, \hat{f}} \leq  1 
$$
for $\hat{f}$ affiliated to $f$ and depending on $C_H$.
Then we calculate 
\begin{align}
\psi'(r) &= \phi \circ \alpha_{K_1}(1) \circ \alpha_{H'_2(r)}(1) 
=\phi \circ \alpha_{H'_2(r)}(1) \circ \alpha_{K_1+E(r)}(1) \nonumber \\
&=\phi \circ \alpha_{K_2(r)}(1) \circ \alpha_{K_1+E(r)}(1)\label{eq: double pumped repeat} 
\end{align}
The first and third equality follow from \eqref{eq: parallel for truncated pair}, and the second equality follows by the commutation formula \eqref{eq: commutation autos integrated}
with $E(r)=\alpha^{-1}_{H'_2(r)}[K_1]-K_1$.  Locality estimates based on the fact that $H'_2(r)$ is anchored in a region at distance $r$ from the origin, yield  $||E(r)||_{\{0,1\},\hat{f}}\to 0$ as $r\to\infty$ and therefore also
\begin{equation}\label{eq: psiprime cauchy}
||   \psi'(r) -  \phi \circ \alpha_{K_2(r)}(1) \circ \alpha_{K_1}(1) || \to 0
\end{equation}
As $K_1,K_2(r)$ are anchored in a finite set, we have $K_1=\Adjoint(V_1), K_2(r)=\Adjoint(V_2(r)) $ and by the locality properties we deduce that  $ ||[V_1,\gamma(g)[V_2(r)]|| \to 0$.
We let the pure density matrix $\rho^{(r)}$ represent the (pure) pumped state $\psi'(r)$ and we write $W_j = \pi(V_j)$.  Then \eqref{eq: psiprime cauchy} reads in the GNS representation
\begin{equation}\label{eq: factor charges}
 || \rho^{(r)}-  \Adjoint(W^*_1 W^*_2(r))\Pi  ||_1 \to 0,
\end{equation}
with $\Pi$, as before, the pure density matrix corresponding to the state $\phi$, and 
\begin{equation}\label{eq: commutation unitaries}
||[W_1,U(g)W_2(r)U^*(g)|| \to 0
\end{equation}
by the corresponding property for $V_1,V_2(r)$ above. We now bound
\begin{equation} \label{eq: control of additivity}
d(h(\psi_2\loopc\psi_1),h(\psi_2)+h(\psi_1)) \leq {C_{\bbS^1}}\sup_{g\in G}\left| \Tr [ \rho^{(r)} U(g)] - e^{i(h_1(g)+h_2(g))}   \right|, \qquad h_j=h(\psi_j)
\end{equation}
for any  $r\in\bbN^+$, with $C_{\bbS^1}=\sup_{\theta\in[0,\pi]}\frac{\theta}{|e^{i\theta}-1|}$.    We use \eqref{eq: factor charges}, \eqref{eq: commutation unitaries} and the definitions of the separate charges
$$
U(g)\Adjoint(W^*_j)\Pi =   e^{i h_j(g)}  \Adjoint(W^*_j)\Pi,\qquad  j=1,2,
$$
to argue that the right hand side of \eqref{eq: control of additivity} vanishes as $r\to\infty$, which concludes the proof. 
%
\end{proof}

\subsection{Homotopy invariance}

\begin{proposition} \label{prop: homotopy invariance products}
Let $\psi_0$ and $\psi_1$ be $G$-invariant loops with common basepoint $\phi$. If  $\psi_0$ and $\psi_1$ are $G$-homotopic, then $h(\psi_1)=h(\psi_0)$.
  \end{proposition}
\begin{proof}
We consider a homotopy $\psi_\lambda(s)$ relating $\psi_0$ and $\psi_1$, as defined in Section \ref{sec: homotopy}, with the families of TDIs $H_\lambda$ and $F_s$,  and we set
$$
C_0= \sup_{s,\lambda} ( \nor H_\lambda \nor_f+ \nor F_\lambda \nor_f ),
$$
with $C_0$ finite by the condition \eqref{Homotopy: Uniform bound}.
Moreover,  for every $\lambda$, the TDI  $H_\lambda$ generates a loop.
For later reference, let us also define the loop 
\begin{equation}\label{eq: funny loop nu}
\lambda\mapsto \nu(\lambda)= \psi_\lambda(0)=\psi_\lambda(1).
\end{equation}
where the last equality follows from the fact that $\psi_\lambda(\cdot)$ is a loop for every $\lambda$. Moreover, $\nu$ is a loop with basepoint $\phi$ because the loops $\psi_0$ and $\psi_1$ both have basepoint $\phi$.

We now take the set $\caS=\caS(f,\phi,C_0)$ as defined in Section \ref{sec: states cut loop} and we take $\epsilon$ as furnished by Proposition \ref{prop: finite index set}. Note that $\epsilon$ depends on $\caS$.
We set $N=1/\lfloor C_0\epsilon\rfloor$ and $\caT_N=\{k/N:k=0,\ldots,N\}$. For any $t\in\caT_{2N} \setminus \{1\}$, we denote by $\hat{t}$ its successor, i.e.\ $\hat{t}=t+\tfrac{1}{2N}$.
We consider functions $x: \caT_{2N} \to \caT_{N}^2 $ that satisfy
\begin{enumerate}
\item $x(0)=(0,0)$
\item $x(1)=(1,1)$
\item $x(\hat{t})-x(t)$ equals either $(1/N,0)$ or $(0,1/N)$.
\end{enumerate}  
and we call such functions \emph{admissible walks}, 
By drawing a straight line between $x(t)$ and $x(\hat{t})$, we obtain a piecewise linear path in the unit square, starting at the bottom-left corner $(0,0)$ and ending in the top-right corner $(1,1)$, and always moving either upwards or to the right, see Figure \ref{fig: homotopywalks} and  \ref{fig: connected walks} for examples.
To every admissible walk $x$, we associate a TDI $D_x(\cdot)$ as follows: for $s\in (t,\hat{t}]$, we set 
\begin{equation}
D_x(s)= 
 \begin{cases}
2H_{x_1(t)} ((x_2(t)+2(s-t)) &  \quad \text{if} \quad   x(\hat{t})-x(t)= (0,\tfrac{1}{N})   \\[1mm]
2F_{x_2(t)} ((x_1(t)+2(s-t)) &  \quad \text{if} \quad  x(\hat{t})-x(t)= (\tfrac{1}{N},0)   
\end{cases}
\end{equation}
This construction is illustrated in Figure \ref{fig: homotopywalks}. Every one of such TDIs $D_x$ generates a loop with basepoint $\phi$. This follows from the definition of homotopy given in Section \ref{sec: homotopy}, in particular from the fact that $H_\lambda$ and $F_s$ generate motion on the two-dimensional sheet of states $(s,\lambda)\mapsto\psi_\lambda(s)$.

\begin{figure}
    \centering
    \begin{minipage}{0.475\textwidth}
        \centering
        \includegraphics[width=0.9\textwidth]{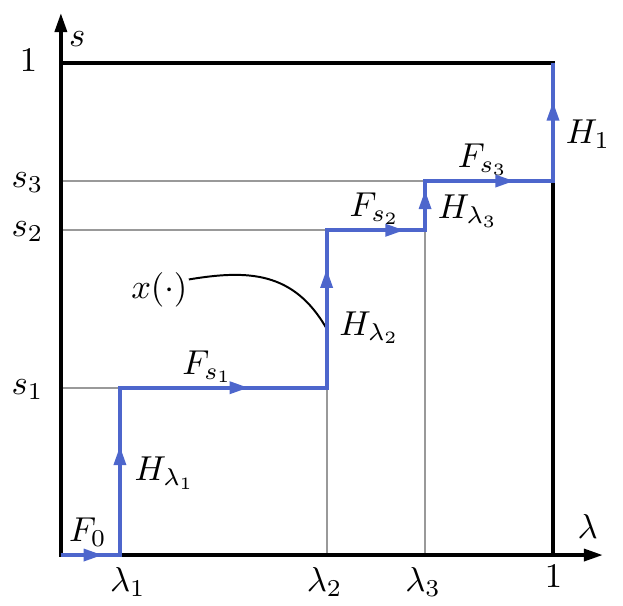}
        \caption{The TDI associated to the walk $x$ is represented in the square $[0,1]^2$. The upward steps are generated by $H$, the right-moving steps are generated by $F$.}
        \label{fig: homotopywalks}
    \end{minipage}\hfill
    \begin{minipage}{0.475\textwidth}
        \centering
        \includegraphics[width=0.9\textwidth]{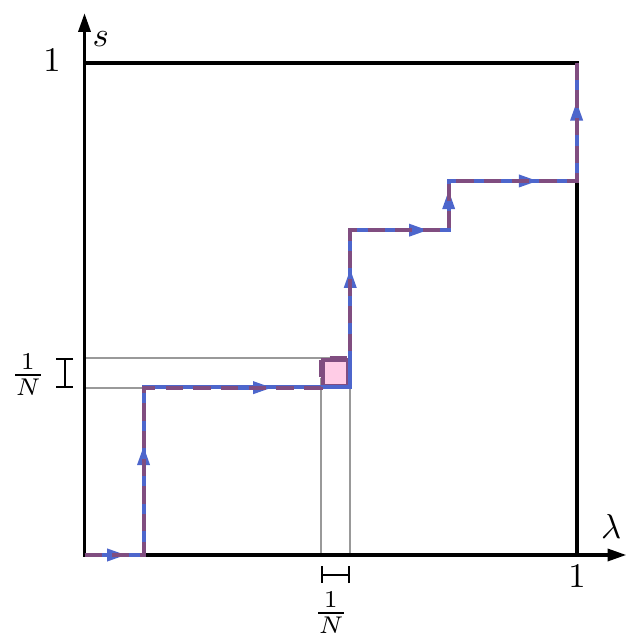}
        \caption{Two walks that are adjacent. They delineate a small square with sides of length~$1/N$.  }
        \label{fig: adjacentpaths}
    \end{minipage}
\end{figure}

We say that two distinct admissible walks $x,y$ are adjacent if $x(t)=y(t)$ for all, but one value of $t\in\caT_{2N}$, that we call $t_*$.  It follows that $x,y$ are obtained from each other by flipping the steps between preceding and following the time $t^*$, i.e. and upwards step becomes right-moving, and vice versa, see Figure \ref{fig: adjacentpaths}.  

We then note that for $x,y$ adjacent admissible walks, we have
$$
||| D_x-D_y |||^{(1)}_f \leq C_0/(2N)
$$
By Proposition \ref{prop: finite index set}, and the choice of $N$ made above, we have $h(D_x)=h(D_y)$.  Therefore, two admissible walks give rise to loops with the same index if those walks can be deformed into each other by switching a walk into an adjacent one. 
 Consider then the two admissible walks
$$
x(t)=\begin{cases} (0,2t) &  t \leq \tfrac{1}{2} \\
(2(t- \tfrac{1}{2}),1) &  t  > \tfrac{1}{2} \\
 \end{cases}
 \qquad
 y(t)=\begin{cases} (2t,0) &  t \leq \tfrac{1}{2} \\
(1, 2(t- \tfrac{1}{2}) ) &    t  > \tfrac{1}{2}  \\
 \end{cases}
$$
These walks can be transformed into each other, as we illustrate in Figure \ref{fig: connected walks}.
\begin{figure}[htb!]
\begin{center}
\includegraphics[width=0.9\textwidth]{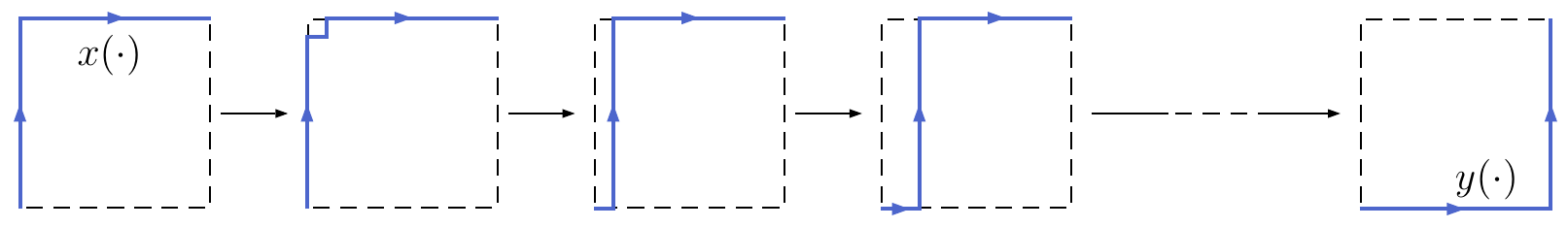}
\caption{Transformation of the walk $x(\cdot)$ into $y(\cdot)$ via a path of adjacent walks. 
}
\label{fig: connected walks}
\end{center}
\end{figure} 
The walk $x$ correspond to the loop $\psi_1\loopc\nu$ and the walk $y$ corresponds to $\nu\loopc\psi_0$. These two loops have hence the same index and by using the additivity Lemma \ref{lem: additivity}, we conclude that $h(\nu)+h(\psi_0)= h(\psi_1)+h(\nu)$ which implies  $h(\psi_0)= h(\psi_1)$. 
\end{proof}

\subsection{Homeomorphism property of $h$}\label{sec: group homomorfism}

We have already established that $h$ maps the concatenation of loops into the sum of charges. We would like to state that $h$ is a homeomorphism, but as it stands, the concatenation  of loops admits neither an identity nor an inverse. This can be remedied by using ($G$)-homotopy to define equivalence classes of loops. The following lemma shows that $h$ acts as a homeomorphism from homotopy classes to $H^1(G)$. 
Let $\Id_\phi$ be the constant loop with basepoint $\phi$, i.e.\ $\phi(s)=\phi$, and recall the definition of the inverse ${\psi^{\theta}}(s)=\psi(1-s)$ associated to a loop  $\psi$.
\begin{lemma}\label{lem: inverse loop}
Let $\psi$ be a $G$-invariant loop with basepoint $\phi$. Then
\begin{enumerate}
\item $h(\Id_\phi)=0$
\item  ${\psi^{\theta}}\loopc\psi$ is $G$-homotopic to $\Id_\phi$.
\item $h(\psi)=-h({\psi^{\theta}})$
\end{enumerate}
 \end{lemma}
\begin{proof}
Item $i)$ follows because a possible TDI generating constant loops is $H=0$, and for that TDI, the pumped state is equal to the basepoint. 
To show item $ii)$, we consider the following homotopy:
\begin{equation}
\zeta_\lambda(s) = \begin{cases} \psi((1-\lambda)2s) &  s\leq 1/2 \\
{\psi^{\theta}}(\lambda+(1-\lambda)(2s-1))  &  s >1/2
\end{cases}
\end{equation}
We note that
$$
\psi((1-\lambda) 2s)=  {\psi^{\theta}}(\lambda+(1-\lambda)(2s-1)), \qquad \text{for $s=1/2$}
$$
by the definition of ${\psi^{\theta}}$, resolving the apparent incongruence between $s\leq 1/2$ and $s>1/2$ in the definition of $\zeta_\lambda(s)$. 
The homotopy property of $(s,\lambda)\mapsto \zeta_\lambda(\cdot)$ now follows because $\psi,{\psi^{\theta}}$ are loops and using Lemma \ref{lem: composition of loops is homotopic to concatenation} item i). 
We then remark that $\zeta_0(\cdot)= {\psi^{\theta}}\loopc\psi$ and $\zeta_1(\cdot)= \Id_\phi\loopc\Id_\phi=\Id_\phi$ which settles item ii). Finally, by additivity and item i), $
h({\psi^{\theta}})+h(\psi)= h({\psi^{\theta}}\loopc\psi)= h(\Id_\phi)= 0
$
which implies item iii). 
\end{proof}

\subsection{Completeness of classification}
\begin{proposition}\label{lem: completeness products}
If a pair of $G$-invariant loops $\psi_1(\cdot)$ and $\psi_2(\cdot)$ with common basepoint $\phi$ have equal index, then they are $G$-homotopic.  
\end{proposition}
\begin{proof}
The loop $\psi^\theta_2\loopc\psi_1$ has index zero, by Lemma \ref{lem: inverse loop} and Lemma \ref{lem: additivity}.  By Theorem \ref{thm: contractibility products}, this loop is $G$-homotopic to the constant loop $\Id_\phi$.  Let $\zeta_\lambda(s)$ be this homotopy, satisfying also $\zeta_\lambda(0)=\zeta_\lambda(1)=\phi$. 
We now construct two derived homotopies $\nu$ and $\omega$ by 
\begin{equation*}
\nu_\lambda(s)=\begin{cases} \zeta_\lambda(s)  &  s\leq 1/2 \\
\zeta_{\lambda-(2s-1)\lambda}(1/2)   &  s > 1/2
\end{cases},\qquad
\omega_\lambda(s)=\begin{cases} \zeta_\lambda(1-s)  &  s\leq 1/2 \\
\zeta_{\lambda-(2s-1)\lambda}(1/2)   &  s > 1/2
\end{cases}.
\end{equation*}
\begin{figure}[h] 
\begin{center}
\includegraphics[width=0.4\textwidth]{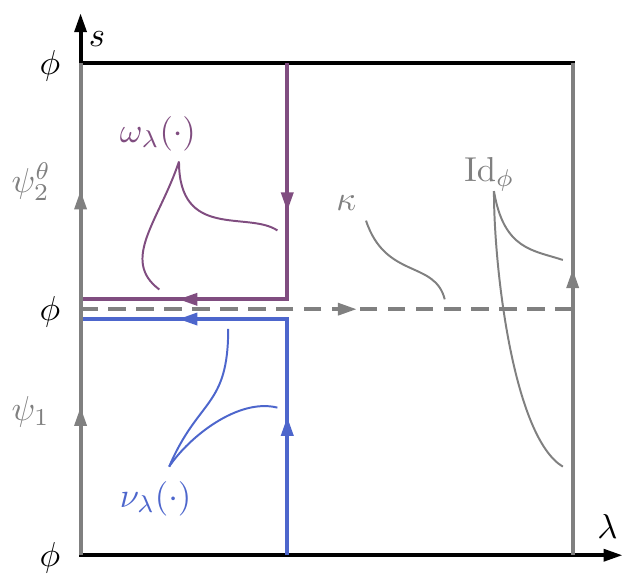}
\caption{The square represents the homotopy $(s,\lambda)\mapsto \zeta_\lambda(s)$. In this square we show the loops $\nu_\lambda(\cdot)$ and $\omega_\lambda(\cdot)$ building the homotopies $\nu$ and $\omega$. }
\label{fig: homotopycontraction}
\end{center}
\end{figure}
The apparent incongruence at $s=1/2$ is resolved by noting that 
$ \zeta_\lambda(s)=  \zeta_{\lambda-(2s-1)\lambda}(1/2)$ for $s=1/2$. 
The fact that $s\mapsto \nu_\lambda(s) $ and $\lambda\mapsto \nu_\lambda(s)$ are generated by uniformly bounded TDIs, follows from the fact that $(s,\lambda)\mapsto \zeta_\lambda(s)$ is a homotopy and from Lemma \ref{lem: time-rescaling almost local evolutions}. 
Then, by construction,  $\zeta_0(1/2)=\phi$ and $\zeta_{\lambda}(0)=\phi$ so that $\nu_\lambda(\cdot)$ are loops with basepoint $\phi$. This confirms that  $(s,\lambda)\mapsto \nu_\lambda(s)$ is indeed a homotopy.  In an analogous way, one checks that  $(s,\lambda)\mapsto \omega_\lambda(s)$ is a homotopy.   In Figure \ref{fig: homotopycontraction}, we illustrate how $\nu$ and $\omega$ are derived from $\zeta$.
We define also the path $\kappa(z)=\zeta_{z}(1/2)$ and note that this path is in fact a loop. Indeed, $\kappa(1)=\phi$ because the original homotopy $\zeta$ reduces for $\lambda=1$ to  $\Id_\phi$, and $\kappa(0)=\zeta_0(1/2)=\phi$ as already used above. 
We are now ready to analyze the homotopies $\nu$ and $\omega$, following the definitions. Since $\zeta_0 = \psi^\theta_2\loopc\psi_1$, we have that $\zeta_0(s) = \psi_1(2s)$ for $s\leq 1/2$. Hence,
$$
\nu_0 = \Id_{\phi}\loopc\psi_1.
$$
For $\lambda=1$, $\zeta_0 = \Id_{\phi}$ while $\zeta_{2-2s}(1/2) = \kappa^\theta(2s)$ and so
\begin{equation}\label{eq: new endpoint homo}
\nu_1 = \kappa^\theta\loopc \Id_{\phi}.
\end{equation}
Similarly, $\zeta_0(1-s) = \psi_2(2s)$ and so
$$
\omega_0 = \Id_{\phi}\loopc\psi_2
$$
and $\omega_1 = \nu_1$. We conclude that $
\Id_{\phi}\loopc\psi_1 
$ and $
\Id_{\phi}\loopc\psi_2 
$
are homotopic to the same loop \eqref{eq: new endpoint homo} and hence to each other. It follows that $\psi_1$ is  homotopic to $\psi_2$.
\end{proof}

\section{Proof of main theorems}\label{sec: proofs of main}

In Section \ref{sec: classification for product loops}, we proved all results of this paper in the special case of loops with product basepoint.  In the present section, we prove our results without this restriction by relying on the results in Section \ref{sec: classification for product loops}. Indeed, by the assumption of invertibility, we can reduce all our results to statements about loops with product basepoint. We use the terminology introduced at the beginning of Section \ref{sec: classification of loops}.

\subsection{Proof of Proposition \ref{prop: equivalence of states}} \label{sec: proof of equivalence propo}
Let $\psi\circ\alpha_H(\cdot)$ be loop. We let $\overline{\psi}$ be an inverse to the basepoint $\psi$, such that
\begin{equation} \label{eq: relation to phi}
\phi=\psi \otimes \overline{\psi} \circ \beta^{-1}.
\end{equation}
where $\phi$ is a product state and $\beta=\alpha_{E}(1)$ for some TDI $E$.  As in the text preceding Proposition~\ref{prop: equivalence of states}, we let $\psi'=\psi \circ \alpha_{H'}(1)$ with $H'$ the truncated TDI. Let also
\begin{equation} \label{eq: relation to phiprime}
\phi'= \psi' \otimes \overline{\psi} \circ \beta^{-1}
\end{equation}
By the uniqueness of the GNS-representation, normality of $\psi'$ w.r.t.\ $\psi$ is equivalent to normality of $\psi'\otimes\overline{\psi}$ w.r.t.\  $\psi\otimes\overline{\psi}$. By \eqref{eq: relation to phi} and \eqref{eq: relation to phiprime}, it is therefore also equivalent to normality of $\phi'$ w.r.t.\  $\phi$. By equation \eqref{eq: commutation autos integrated}, we have (here we write $\alpha_{H'}$ for $\alpha_{H'}\otimes \mathrm{id}$)
$$
\phi'= \phi  \circ \beta \circ \alpha_{H'}(1) \circ \beta^{-1} =   \phi  \circ \alpha_{\beta(H')}(1).
$$
Since $\phi  \circ \alpha_{\beta(H)}(\cdot)$ is a loop with product basepoint, the state $\phi  \circ \alpha_{\beta(H)'}(1)$ is normal with respect to $\phi$ by Lemma~\ref{lem: cut state}.  
Moreover, by Lemma \ref{lem: local perturbation tdi}, and the fact that $\beta(H)'-\beta(H')$ is a TDI anchored in a finite set, we get that 
$\phi  \circ \alpha_{\beta(H')}(1) $  is normal with respect to $\phi  \circ \alpha_{\beta(H)'}(1) $.
Since normality of pure states is a transitive relation, the claim is proven.

\subsection{Loops whose basepoint is $G$-equivalent to a product}\label{sec: loops with equivalent base}

As a first step in dealing with loops with non-product basepoint, we formulate a lemma for loops whose basepoint is $G$-equivalent (as opposed to stably $G$-equivalent) to a product state. 
\begin{lemma}\label{lem: loops with prod equiv}
Let $\psi$ be a $G$-loop generated by the TDI $H$, such that, for some  $G$-invariant TDI $K$ and a product state $\phi$, 
$$\phi=\psi(0)\circ \alpha^{-1}_K(1).$$ 
Then the loop $\psi$ is homotopic to the loop 
\begin{equation} \label{eq: def loop three parts}
\nu=\phi \circ \left(({\alpha_{K}} \loopc \alpha_{H}) \loopc  {\alpha^\theta_{K}}\right)
\end{equation}
which has product basepoint $\phi$. Moreoever, their indices are equal, i.e.\
\begin{equation} \label{eq: equality index inter}
I(\psi(0),H)=h(\nu)
\end{equation} 
\end{lemma}
\begin{proof}
Recalling the definition of Section~\ref{sec:Concatenation of evolutions}, we note that the loop $\nu$ is such that  ${\alpha_{K}}$ is traversed between times $s=0$ and $s=1/4$, $\alpha_{H}$ is traversed between $s=1/4$ and $s=1/2$, and ${\alpha^\theta_{K}}$ is traversed between $s=1/2$ and $s=1$. We define the homotopy, illustated in   Figure \ref{fig: homotopytoproductbase}, 
\begin{figure}[h] 
\begin{center}
\includegraphics[width=0.7\textwidth]{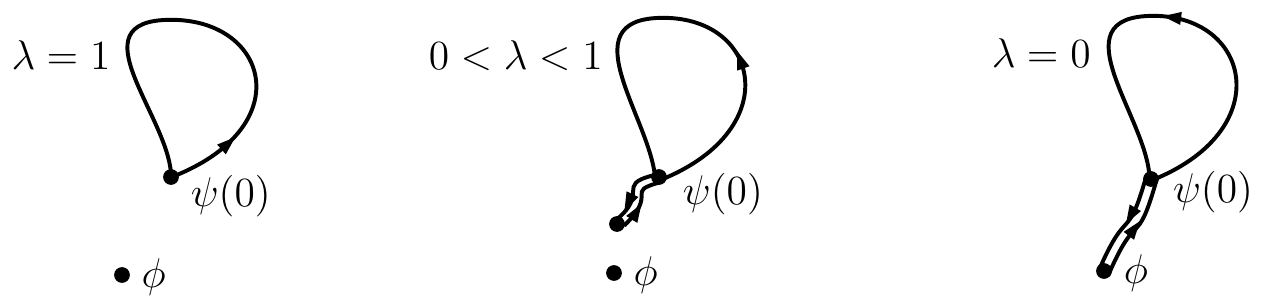}
\caption{The homotopy $\mu$ moves the basepoint $\psi(0)$ up to the product state $\phi$.}
\label{fig: homotopytoproductbase}
\end{center}
\end{figure}
$$
\mu_\lambda(s) = \begin{cases}  \nu(\lambda/4+(1-\lambda)s )  &   s\leq 1/4 \\
\nu(s)&   1/4 < s < 1/2 \\
 \nu(1/2 + (s-1/2)(1-\lambda) )   &     1/2\leq s
\end{cases}
$$
We note first that there is no incongruence at $s=1/4$ and $s=1/2$ since, for any $\lambda$, 
$$
\mu_\lambda(1/4)= \nu(1/4 ) = \psi(0)=\psi(1)=   \nu(1/2 ) =  \mu_\lambda(1/2)
$$
From Lemma \ref{lem: time-rescaling almost local evolutions}, we see that the family $(s,\lambda)\mapsto \mu_\lambda(s)$ is generated by uniformly bounded TDIs in $s$ and $\lambda$ direction. To conclude that this family is homotopy, we check that it is a loop for all $\lambda$, using the definition of time-reversed loops;
$$
 \mu_\lambda(0) =\nu(\lambda/4 )=\psi(\lambda) ={\psi^\theta}(1-\lambda)= \nu(1-\lambda/2)=         \mu_\lambda(1). $$
This homotopy connects the loop $\nu$ with the loop $(\Id_{\psi(0)} \loopc \psi) \loopc  \Id_{\psi(0)}$ and the latter is clearly homotopic to $\psi$, which settles the first claim of the lemma. 
Since we have not yet proven in general that homotopy preserves the index,  we need another argument to get \eqref{eq: equality index inter}.  The pumped state of the loop $\nu(\cdot)$ is
\begin{equation}\label{Complicated Pumped State}
   \phi \circ   \alpha_{K'}(1)\circ \alpha_{H'}(1) \circ \alpha_{(K')^\theta}(1)=  
\phi \circ   \alpha_{K'}(1)\circ \alpha_{H'}(1) \circ  \alpha_{K'}^{-1}(1) = 
\phi \circ \alpha_{\alpha_{K'}(1)[H']}(1) 
\end{equation}
where the first expression follows because  $(K')^\theta=(K^\theta)'$ with $K^\theta$ the TDI generating ${\alpha^\theta_K}$, i.e.\ ${K^\theta(s)}=-K(1-s)$. 
Then, we observe that 
$$
\alpha_{K}(1)[H']-\alpha_{K'}(1)[H']
$$
is a $G$-invariant TDI anchored in $\{0,1\}$. This follows from Duhamel's formula (\ref{Duhamel}) and Lemma \ref{lem: loc and liebrobinson} since $[K-K',\alpha_{K'}(\cdot,\cdot)[H']]$ is anchored\footnote{Here, we assume without loss that all interaction terms are supported on intervals. Indeed, for any finite set $X$, we have that $\caA_X\subset\caA_{[\min(X),\max(X)]}$.} in $\{0,1\}$. This implies then by Lemma \ref{lem: local perturbation tdi} item~ii) that the pumped state~(\ref{Complicated Pumped State}) has zero relative charge w.r.t. the state
$$
\phi \circ \alpha_{\beta[H']}(1)= \phi \circ \beta \circ \alpha_{H'}(1) \circ \beta^{-1}
$$
where we abbreviated $\beta=\alpha_{K}(1)$. 
Hence to prove the claim of the lemma, we have to show that the relative charge of 
$$
\phi\circ \beta \circ \alpha_{H'}(1)\circ \beta^{-1}  \quad \text{w.r.t.}\quad \phi
$$
equals the relative charge of 
$$
\phi\circ \beta \circ \alpha_{H'}(1)  \quad \text{w.r.t.}\quad \phi\circ \beta
$$
since the latter are the pumped state and basepoint, respectively, of the loop $\psi$.
This equality of relative charges follows then from item $v)$ of Proposition \ref{prop: zerodim}, taking $\delta=\beta$.
\end{proof}

\subsection{Tensor products}\label{sec: tensor products}

Given two $G$-loops $\psi_1,\psi_2$ on spin chain algebra's $\caA_1,\caA_2$, generated by TDIs $H_1,H_2$, we can consider the product loop 
$
\psi_1 \otimes\psi_2
$
on the spin chain algebra $\caA_1\otimes\caA_2$,
generated by the TDI
$$
H_1\otimes \Id +\Id \otimes H_2.
$$
The truncation of this TDI to the left equals the sum of truncated TDIs $H'_1\otimes \Id + \Id \otimes H'_2$ and it then follows that the pumped state is a tensor product of pumped states 
$
\psi_1(0) \circ\alpha_{H'_1}(1) \otimes \psi_2(0) \circ\alpha_{H'_2}(1)
$ and so
$$
I(\psi_1(0) \otimes \psi_2(0),H_1\otimes \Id +H_2\otimes \Id ) =
I(\psi_1(0),H_1 ) + I(\psi_2(0),H_2) 
$$
Proposition \ref{prop: zerodim} item iii).

We will establish in the next sections that the index depends only on the loop and not on the generating TDI, but we can already see that, if a loop is generated by the TDI $H=0$, then its index is zero. Hence, any $G$-loop $\psi$  generated by $H$ can be extended to a $G$-loop  $\psi\otimes \Id_\nu$, with $\nu$ a $G$-state, generated by $H\otimes\Id$, such that the index does not change, i.e.\
\begin{equation} \label{eq: index adjoining constant loop}
I(\psi(0)\otimes\nu,H\otimes \Id)=I(\psi(0) ,H).
\end{equation}

\subsection{The associated loop}
The main tool that we will use in the remaining proofs is that every $G$-loop can be related to a loop with product basepoint, such that they have the same index.
We consider hence a general $G$-loop $\psi$ and we let $\overline{\psi(0)}$ be a $G$-inverse to $\psi(0)$. By the discussion in Section \ref{sec: tensor products}, the loop $\psi\otimes \Id_{\overline{\psi(0)}}$ has the same index as $\psi$. However, the former loop has a basepoint that is $G$-equivalent to a product state and hence we can invoke Lemma \ref{lem: loops with prod equiv} to construct a loop $\psi^A$, which we call the \emph{associated loop}, that has a product basepoint, and it is such that 

\begin{lemma} \label{lem: equivalence associated loop}  
Let $\psi$ be a $G$-loop generated by the TDI $H$. There is a $G$-loop $\psi^A$ with product basepoint such that
$
h(\psi^A)=I(\psi(0),H)
$.
\end{lemma}

\subsection{Proof of Theorems \ref{thm: classification loops} and \ref{thm: pump index}}

\subsubsection{Well-definedness of the index}

If we have two TDIs $H_1,H_2$ that both generate a loop $\psi$, then by definition, that the associated loops can be chosen equal. Hence, by the results in Section \ref{lem: index indep of h}, their index is equal and hence also 
$ I(\psi(0),H_1)= I(\psi(0),H_2) $. Therefore, also for loops with a non-product basepoint, the index $I(\cdot,\cdot)$ depends only on the loop itself. 
We can henceforth use the notation $h(\cdot)$. In particular, this means that Theorem \ref{thm: pump index} is proven once we establish all the properties of the index $h(\cdot)$ claimed in Theorem \ref{thm: classification loops}, which we do now.

\subsubsection{Proof of Theorem \ref{thm: classification loops}}

We now check the items of Theorem \ref{thm: classification loops} one by one.

Item i) If a loop is constant, we can choose the generating TDI $H=0$. In that case the pumped state equals the basepoint and the index is manifestly equal to zero.

Item ii)  The examples described in Section \ref{sec: examples} realize every $h\in H^1(G)$.

We prove item iii) by reducing it to the already treated case of a product basepoint using associated loops. Let the $G$-loops $\psi_1,\psi_2$ be stably $G$-homotopic. Then by definition (see Section~\ref{sec: stable homotopy})  these loops are $G$-homotopic to each other upon adjoining constant loops with product basepoint. From Section \ref{sec: tensor products}, we know that this extension does not change their respective index. Hence without loss of generality we assume henceforth that $\psi_1,\psi_2$ are $G$-homotopic. 
Consider an associated loop $\psi_1^{A}$ with product basepoint $\phi$, obtained by adjoining an inverse $\overline{\psi_1(0)}$. Since $\psi_1(0)$ and $\psi_2(0)$ are $G$-equivalent states, so are $\psi_1(0) \otimes \overline{\psi_1(0)}$ and $\psi_2(0)\otimes \overline{\psi_1(0)}$ and hence also $\psi_2$ has an associated loop $ \psi_2^A$ with basepoint $\phi$. By construction, the two loops $\psi_1^A,\psi_2^A$ are loops of states on the same algebra. Moreover, $\psi_j^A$ are stably homotopic to $\psi_j$ for $j=1,2$ and since the latter are $G$-homotopic to each other by assumption, we conclude that $\psi_1^A$ is $G$-homotopic to $\psi_2^A$.  Since they share a product basepoint, these loops have equal index, by the results in Section \ref{sec: classification for product loops}.  By Lemma \ref{lem: equivalence associated loop}  it then follows that also $\psi_1,\psi_2$ have equal index.

Item iv) Here again, we can reduce the problem tho that of loops with product basepoints. For a pair $\psi_1,\psi_2$ of $G$-loops with common basepoint, we find a pair of associated loops $\psi^A_1,\psi^A_2$ with common basepoint $\phi$. 
We note first that
$$
(\psi_2\loopc\psi_1)^A \quad \text{is homotopic to}\quad  \psi^A_2\loopc\psi^A_1
$$
which follows from Lemma \ref{lem: inverse loop} item ii.  Since these are loops with product basepoint, the homotopy implies that they have the same index. The claim then follows from Lemma~\ref{lem: equivalence associated loop} and the additivity of the index for loops with product basepoint.

Item v) This was already proven in Section \ref{sec: tensor products}.

Item vi)  If the $G$-loops $\psi_1,\psi_2$ have stably $G$-equivalent basepoints, then by adjoining a product state, the basepoints are $G$-equivalent. Without loss of generality, we assume hence that the basepoints of 
$\psi_1,\psi_2$ are $G$-equivalent. Let $\overline{\psi_1(0)}$ be a $G$-inverse to $\psi_1(0)$, then we consider the loops
\begin{equation} \label{eq: two stable loops}
\psi_j \otimes \Id_{\overline{\psi_1(0)}}  \otimes \Id_{\psi_1(0)}, \qquad j=1,2
\end{equation}
and we claim that they are $G$-homotopic. This yields the required stable $G$-homotopy of the original loops, since the adjoined product of constant loops is homotopic to a constant  loop with product basepoint. As in item iii), the loops 
$
\psi_j \otimes \Id_{\overline{\psi_1(0)}} 
$
are $G$-homotopic to the associated loops $\psi^A_j$ which can be chosen to have the same product basepoint. The assumption $h(\psi_1)=h(\psi_2)$ implies that $h(\psi^A_1)=h(\psi^A_2)$ and hence $\psi^A_1,\psi^A_2$ are $G$-homotopic by Proposition~\ref{lem: completeness products}. This in turn implies that the loops in~(\ref{eq: two stable loops}) are $G$-homotopic indeed.

\appendix

\section{Appendix:  basic tricks in Hilbert space}\label{app: tricks}
In this appendix we collect a few lemma's dealing with zero-dimensional systems, i.e.\ we deal with Hilbert spaces without any local structure. 

\subsection{An estimate for bipartite systems}
\begin{lemma}\label{lem: bipartite}
Consider a bipartite system $\caH=\caH_a\otimes\caH_b$. Let $P=P_a\otimes P_b$ where $P_a,P_b$ are rank-one projectors acting on Hilbert spaces $\caH_a,\caH_b$. Let $\rho$ be a density matrix on $\caH$ and let $\rho_a,\rho_b$ be the reduced density matrices on $\caH_a,\caH_b$.
 Then 
 $$
 ||\rho-P||_1 \leq 6\sqrt{ ||\rho_a-P_a||_1}+  6\sqrt{||\rho_b-P_b||_1  }
 $$
\end{lemma}
\begin{proof}
We abuse notation by writing $P_i$ to denote $P_i \otimes 1$. 
First,
\begin{align*}
||\rho-P_i\rho P_i||_1  & \leq   ||P_i\rho\bar P_i ||_1+ ||\bar P_i\rho P_i ||_1+||\bar P_i\rho\bar P_i ||_1 \\
 & \leq  || P_i\sqrt{\rho} ||_2 || \sqrt{\rho}\bar P_i ||_2  + ||\bar P_i\sqrt{\rho} ||_2 || \sqrt{\rho} P_i ||_2  +  \Tr ( \bar P_i \rho ) \\
 &  \leq  3 \sqrt{\Tr ( \bar P_i \rho )} 
\leq 3\sqrt{||\rho_i-P_i||_1}  
\end{align*}
The second inequality follows from the Cauchy-Schwarz inequality and the last one is because
$$\Tr ( \bar P_i \rho ) = \Tr ( \bar P_i (\rho_i - P_i) )\leq ||\rho_i-P_i||_1,$$ where we used that $P_i$ is rank one. Then, we bound  $ ||P_b\rho P_b-P_b P_a\rho P_a P_b ||_1 \leq  ||\rho- P_a\rho P_a  ||_1 $ and so  
$$ || \rho- P_b P_a\rho P_a P_b ||_1 \leq \sum_{i=a,b}   ||\rho- P_i\rho P_i  ||_1 \leq \delta := 3\sum_{i=a,b}\sqrt{||\rho_i-P_i||_1}.
$$
Since $P=P_aP_b$ is rank one, we conclude that
$$
|1-\Tr (P\rho)| = |\Tr(\rho-\Tr (P\rho)P)|\leq ||\rho-\Tr (P\rho) P ||_1 = || \rho- P\rho P ||_1 \leq  \delta.
$$
Hence $|| \rho- P ||_1 \leq ||\rho-\Tr (P\rho) P ||_1+||(1-\Tr (P\rho)) P ||_1 \leq 2\delta$. 
\end{proof}

 \subsection{Parallel transport}
We now turn to the
\begin{proof}[Proof of Lemma~\ref{lem: parallel transport simple}]
Let the normalized vectors $\Omega,\Psi \in \caH$ be representatives of  $\omega,\nu$. They can be chosen such that $a:= \langle\Psi,\Omega\rangle$ is real, with $0\leq a\leq 1$. Firstly, we note that $ ||\nu-\omega ||= \sqrt{1-a^2}$. 
We write then $\Psi=a \Omega+\sqrt{1-a^2}\Psi_\perp$ where $\Psi_\perp$ is orthogonal to $\Omega$ and has unit norm. 
Let now
$$
\Psi(s):=y(s)\Omega + \sqrt{1-y(s)^2}\Psi_\perp, \qquad y=a+(1-a)(1-(1-s)^2) 
 $$
Then, 
\begin{enumerate}
\item $||\Psi(s)||=1$,
\item $\Psi(0)=\Psi$ and  $\Psi(1)=\Omega$,
\item $ |\partial_s y|,  |\partial_s \sqrt{1-y(s)^2}| \leq 2\sqrt{1-a} $. 
\end{enumerate}
Let $P(s)$ denote the orthogonal rank-one projector onto the span of $\Psi(s)$.
We consider the adiabatic generator 
\begin{equation}\label{Kato generator}
K(s)=i (P(s)\dot{P}(s)-\dot{P}(s)P(s)),
\end{equation}
which satisfies
\begin{equation*}
\dot P(s) = i[K(s),P(s)].
\end{equation*}
We then have
$$
|| K(s)|| \leq 2 || \dot{P}(s)|| \leq  4 ||\dot{\Psi}(s)|| \leq 8 \sqrt{1-a} \leq 8 || \nu-\omega ||.
$$
This proves items (i,ii). We turn to item (iii). If $h_{\nu/\omega} = 0$, then $U(g) \Omega = z_\omega(g)\Omega$ and $U(g) \Psi = z_\nu(g)\Psi$ with $z_\nu(g)= z_\omega(g)=:z(g)$. Therefore $U(g) \Psi_\perp = z(g)\Psi_\perp$ and so $U(g)\Psi(s) = z(g)\Psi(s)$ for all $s$. It follows that $P(s),\dot P(s)$ are invariant and so is $K(s)$ by definition~(\ref{Kato generator}).
\end{proof}

\subsection{Contracting loops}

In this section, we prove Lemma~\ref{lem: contractibility zero dim} on the contractibility of loops in Hilbert space.

Throughout this work, we often need to appeal to the fundamental theorem of calculus for Banach-space valued functions defined on the interval $[0,1]$. Let $X$ be an arbitrary Banach space. We first recall that a strongly measurable Banach-space valued function $J: [0,1]\to X: s\mapsto J(s)$ is Bochner integrable if and only if $\int_0^s du ||J(u)||<\infty$, see~\cite{diestel1978vector}. We consider functions $F: [0,1]\to X: s\mapsto F(s)$ that satisfy the following property:
$$
F(s)-F(0)=\int_0^s du J(u), \qquad  \sup_{s\in[0,1]} ||J(s)||  <\infty
$$
where $s\mapsto J(s)$ strongly measurable. In particular the function 
$$s\mapsto \alpha_{H}(s)[A],$$
for a TDI $H$ and $A\in \caal$ satisfies this property, with $J=\alpha_{H}(u)\{[H(u),A]\}$.

For the application we have in mind here, it suffices to restrict our attention to the case of finite-dimensional Banach spaces. In that case (more generally in any Banach space with the Radon-Nikodym property with respect to the Lebesgue measure), a function $s\mapsto F(s)$ satisfies the property above if and only if it is Lipschitz continuous, in which case $J(s)$ is its derivative almost everywhere. We say that $C_F$ is a Lipschitz bound if $ \sup_{s\in[0,1]} ||J(s)|| \leq C_F$.   We now come to the
\begin{proof}[Proof of Lemma~\ref{lem: contractibility zero dim}]
As remarked in~(\ref{eq: finite tdi}), the evolved state $\nu(s)=\nu\circ \alpha_E(s)$ is of the form
$$
\nu(s) = \nu(0)\circ \Adjoint( U^*(s)), \qquad  U(s)=\id+i\int_0^s du E(u)U(u).
$$
We choose a unit vector representative $\Omega$ of $\nu(0)$, and we set 
$$
\widetilde\Psi(s)= U(s) \Omega.
$$
Then $s\mapsto\widetilde\Psi(s)$ is Lipschitz, and its almost sure derivative is $i E(s)\widetilde\Psi(s)$. 
Of course, changing the phase of $\widetilde\Psi(s)$ does not affect the state $\nu(s)$ and we use this to introduce a crucial modification of $s\mapsto \widetilde\Psi(s)$.
\begin{lemma}
There exists a Lipschitz function $a: [0,1]\mapsto U(1)$, with Lipschitz bound $ 4|||E||| $, and such that 
$$
\Psi(s)= a(s) \widetilde\Psi(s)
$$ 
satisfies
$$
k(s) := \Re \langle \Psi(s) , \Omega \rangle \geq -1/2.
$$
\end{lemma}
\begin{proof}
Let $\widetilde k(s)=\Re \langle \widetilde \Psi(s) , \Omega \rangle$. As $s\mapsto \widetilde \Psi(s)$ is Lipschitz continuous with bound $|||E|||$, so is $\widetilde k(s)$. Then the sets
$$
\kappa_0=\{ s\in [0,1], \, k(s)> 0   \}, \qquad 
\kappa_1=\{ s\in [0,1], \,  k(s)< -1/2   \}
$$
are open and
$$
\dist(\kappa_0,\kappa_1) \geq \frac{1}{2 |||E|||}.
$$
We can therefore construct a continuous, $U(1)$-valued function $a$ with Lipschitz bound $4 |||E|||$  such that $a(\kappa_0)=1, a(\kappa_1) =-1$
and the vector $a(s)\widetilde\Psi(s)$ satisfies the required properties.
\end{proof}

We started from a loop $\nu(\cdot)$ and we now have a function $s\mapsto \Psi(s)$. The latter function is not necessarily a loop of unit vectors in the Hilbert space, but it is a family of representatives of $\nu(s)$, in particular 
\begin{equation}\label{eq: almost vector loop}
\Psi(0) = \Omega,\qquad \Psi(1) \propto \Omega.
\end{equation}
We now set for $\lambda\in[0,1]$,
$$
\Psi_\lambda(s):= \frac{1}{\sqrt{N(s,\lambda)}} \left(\lambda\Omega +  (1-\lambda) \Psi(s) \right),    \qquad  N(s,\lambda) = \lambda^2+(1-\lambda)^2+2\lambda(1-\lambda)k(s).
$$
With this definition and from the above construction of $\Psi(s)$ and \eqref{eq: almost vector loop}, it follows that 
\begin{enumerate}
\item $N \geq 1/4$,
\item $||\Psi_\lambda(s)||=1$,
\item $\Psi_\lambda(0) = \Omega$ and   $\Psi_\lambda(1)\propto \Omega$,
\item $\Psi_0(s)=\Psi(s)$ and  $\Psi_1(s)=\Omega$.
\end{enumerate}
Since $k(s)$ is Lipschitz, and $N\geq 1/4$, we infer that $\frac{1}{\sqrt{N}}$ is Lipschitz. Therefore, and since $s\mapsto a(s)$ and $\Psi$ are Lipschitz, also $s\mapsto \Psi_\lambda(s)$ is Lipschitz. A short calculation yields $||\partial_s\Psi_\lambda(s)|| \leq 20 |||E|||$. We denote by $P=P(\lambda,s)$ the rank-one projector on the range of $\Psi_\lambda(s)$. Writing $P=|\Psi\rangle\langle\Psi|$, the Lipschitz continuity of $s\mapsto \Psi_\lambda(s)$ gives that $s\mapsto P(\lambda,s)$ is Lipschitz with $|| \partial_sP(\lambda,s) || \leq 40 |||E|||$. 
We can now set
$$
E_\lambda(s)= -i [\partial_sP(\lambda,s) ,P(\lambda,s)]
$$
to satisfy item iii) of the lemma. Moreover, $||| E_\lambda |||\leq 80 |||E|||$.
To construct the family $F_s(\cdot)$, we proceed analogously, but the considerations are simpler because the functions $\lambda\mapsto \Psi_\lambda(s)$ are clearly Lipschitz in $\lambda$. Here we find $||\partial_\lambda\Psi_\lambda(s)|| \leq 52$ and so $||| F_s |||\leq 208$.
\end{proof}

\section{Proof of Proposition \ref{prop: uniqueness of ground state}}   \label{sec: app stability}
 
We will use finite-volume restrictions of interactions $
H^{(L)}_S= \chi(S\subset [-L,L])$ that inherit bounds since
$
||H^{(L)} ||_f \leq   ||H||_f
$. 
We note that $\iota(H^{(L)})$ (recall the definition in Subsection \ref{subsec:alal}) is a Hermitian element of the finite-dimensional $C^*$-algebra $\caA_L=\caA_{[-L,L]}$,  which is isomorphic to a full matrix algebra. We write $\Tr_L(\cdot)$ for the trace on $\caA_L$ and we say that $H^{(L)}$ has a spectral gap $\Delta$ if the minimum of the spectrum of $\iota(H^{(L)})$ is a simple eigenvalue $E_0$ and the next smallest eigenvalue is no smaller than $E_0+\Delta$. We let $P_L$ denote the one-dimensional spectral projector corresponding to $E_0$ and we write $\bar P_L=1-P_L$. 
Further, $\omega_{P_L}$ denotes the state on $\caA_L$ given by 
$$ \omega_{P_L}(A) = \Tr (P_L A), \qquad A \in  \caA_L.
$$

\subsection{Preliminaries}

For a function $v \in L^1(\bbR)$,
we define the map $\caK^H_v: \caA\to\caA$
\begin{equation*}
\caK^{H}_v[A] = \int_{-\infty}^\infty  v(t)\alpha_H(t)[A]dt
\end{equation*}
Note that we used the evolution $\alpha_H(t)$ here for $t\in \bbR$ instead of $[0,1]$.
\begin{lemma}\label{lem: Block diagonalization} 
Let $H^{(L)}$ have a spectral gap $\Delta$, uniformly in $L$.  There is a function $v\in L^1$ such that 
\begin{enumerate}
\item $\int_{-\infty}^{\infty} v(t) dt=1$ and $v\geq 0$;
\item  $\caK^{H}_v[\caal]\subset \caal$
\item $\caK^{H^{(L)}}_v(P_LA\bar{P}_L)=\caK^{H^{(L)}}_v(\bar{P}_LA\bar{P}_L)=0$ for any $A\in \caA_L$.
\item For any $A\in \caal$,  $\lim_L\caK^{H^{(L)}}_v(A)= \caK^{H^{}}_v(A)$.
\end{enumerate}
\end{lemma}
This Lemma is a special case of the,  by now, standard constructions introduced in \cite{hastings2005quasiadiabatic,bachmann2012automorphic}. The following lemma is a collection of facts appearing in the proof of Proposition \ref{prop: bhm}, see again~\cite{nachtergaele2020quasi}.
\begin{lemma}\label{lem: limiting ground state} 
Assume the setup of Proposition \ref{prop: bhm}, with $H=F+W$. Then $H^{(L)}$ has a spectral gap $\Delta=\tfrac{1}{2}$ uniformly in $L$.  Moreover, 
for any $A\in \caA$ with finite support, the weak limit
$$
\nu(A)=\lim_L \omega_{P_L}(A)
 $$
exists. By density, $\nu$ extends to a state on $\caA$ and it is a ground state for $H$.
\end{lemma}

\subsection{Proof of Proposition \ref{prop: uniqueness of ground state}}

For any ground state $\psi$ associated to $H$, \eqref{eq: invariance of ground states} and item $iv$ Lemma~\ref{lem: Block diagonalization} imply that
$$
\psi(A)=  \psi(\caK_v^H(A)) = \lim_L \psi(\caK_v^{H^{(L)}}(A))
$$
for $A$ with finite support and
with  $v$ as in Lemma~\ref{lem: Block diagonalization}.
Next,
\begin{align}
\psi(A) &=  \lim_L ( \psi(\caK_v^{H^{(L)}}(P_LAP_L)) +  \psi(\caK_v^{H^{(L)}}(\bar P_L A\bar P_L)) 
\\
&=  \lim_L ( \psi(P_LAP_L) +  \psi(\caK_v^{H^{(L)}}(\bar P_L A\bar P_L))  \label{eq: origina of decomp}
\end{align}
with the first equality follows from item iii) of Lemma~\ref{lem: Block diagonalization}, and
where the second equality follows because $P_L$ is a one-dimensional spectral projection of $H^{(L)}$. 
For $L$ such that $[-L,L]$ contains the support of $A$, we have hence
$$
\psi(P_LAP_L)= \psi(P_L) \Tr_L(P_LA).
$$
Since $\psi(P_L) \in [0,1]$, there is a subsequence $L_n$ such that the limit $\lim_n\psi(P_{L_n})$ exists. We call this limit $a \in [0,1]$.
We then see from Lemma \ref{lem: limiting ground state} that
$$
\lim_n \psi(P_{L_n}AP_{L_n}) = a \nu(A)
$$
and from \eqref{eq: origina of decomp}, it then follows that 
\begin{equation}\label{eq: convex decomposition}
\psi(A)= a \nu(A) +(1-a)\mu(A)
\end{equation}
where, in case $ a < 1$,  
$$
\mu(A)=\lim_n    \frac{1}{\psi(\bar{P}_{L_n})}  \psi(\bar{P}_{L_n}\caK_v^{H^{(L_n)}}(A)\bar{P}_{L_n})
$$
and we used Lemma \ref{lem: Block diagonalization} item iii) to commute the projectors. The existence of the limit follows because the first term in \eqref{eq: origina of decomp} has a limit.  
We now claim that $\mu$ extends to a state on $\caA$. Indeed, positivity follows from the nonnegativity of the function $v$. Normalization follows from $\caK_v^{H^{(L_n)}}(\bbI) = \bbI$ since $\int v = 1$, and the extension is then by density of operators with local support. Therefore, also \eqref{eq: convex decomposition} extends to any $A\in \caA$.

Now, since $\psi$ is assumed to be pure, \eqref{eq: convex decomposition} means that either $a=0$ or $a=1$, since we can easily check that $\nu$ and $\mu$ are not equal. We will exclude the case $a=0$, which will end the proof.  

We define the boundary operator
$$
B_L= \sum_{S: S \cap [-L,L]\neq 0, S \cap [-L,L]^c \neq \emptyset}  W_S
$$
satisfying $||B_L|| \leq 2f(1) ||W||_f$. By the variational principle Lemma \ref{lem: variational principle}
$$
\psi(H^{(L)}+B_L) \leq (\omega_{P_L}\otimes\psi|_{[-L,L]^c})(H^{(L)}+B_L)
$$
which implies 
$$
\psi(H^{(L)}) \leq \omega_{P_L}(H^{(L)}) + 2||B_L||  
$$
By the gap assumption, we have 
also
$$
\psi(H^{(L)}) \geq \psi(P_L) E_{0,L}+ (1-\psi(P_L)) (E_{0,L} +\Delta)    
$$
and hence
$$
\psi(P_L) E_{0,L}+ (1-\psi(P_L)) (E_{0,L} +\Delta) \leq  E_{0,L} + 2||B_L|| 
$$
from which we conclude that
$$
(1-\psi(P_L))   \leq  \frac{2||B_L|| }{\Delta} \leq  \frac{4 f(1) ||W||_f  }{\Delta}.
$$
If $||W||_f$ is small enough, we get that $a = \liminf_L\psi(P_L)>0$ and so the alternative $a=0$ is indeed excluded. \hfill $\Box$

\bibliographystyle{unsrt}
 
\bibliography{sptpump}

\end{document}